\documentclass[11pt]{article}

\usepackage[margin=1in]{geometry}
\usepackage[T1]{fontenc}
\usepackage[nopatch=footnote]{microtype}
\usepackage[colorlinks=true]{hyperref}

\usepackage{
  amsmath, amsfonts, amssymb, amsthm,
  booktabs, xcolor, setspace, paralist, makecell, longtable,
  bm, orcidlink, dsfont, graphicx, float, mathrsfs, url,
  tabularx, array, multirow, subcaption
}
\usepackage{natbib}
\usepackage{lmodern}
\usepackage{authblk}

\definecolor{darkblue}{rgb}{0,0,.6}
\hypersetup{
  citecolor=darkblue,
  linkcolor=darkblue,
  urlcolor=darkblue
}

\graphicspath{{plots_new/}}

\newsavebox\CBox
\newcommand{\cH}{\mathcal H}
\newcommand{\cI}{\mathcal I}
\newcommand{\cU}{\mathcal U}
\newcommand\numberthis{\addtocounter{equation}{1}\tag{\theequation}}
\newcommand{\<}{\langle}
\renewcommand{\>}{\rangle}
\newcommand{\norm}[1]{\lVert#1\rVert}
\newcommand{\tr}{\mathrm{tr\hspace{2mm}}}
\newcommand{\Cov}{\mathrm{cov}}

\renewcommand{\hat}{\widehat}
\newcommand{\normm}[1]{\left\lVert#1\right\rVert}
\newcommand{\Em}{\mathbb E}
\newcommand{\cC}{\mathcal C}
\newcommand{\Rm}{\mathbb R}
\newcommand{\tp}{^\top}
\newcommand{\td}{\tilde}
\newcommand{\Var}{\mathrm{var}}
\newcommand{\subf}[2]{%
\begin{subfigure}[t]{#2\textwidth}
\centering
\includegraphics[width=\textwidth]{#1}
\end{subfigure}}

\providecommand{\appendixone}{}

\theoremstyle{plain}
\newtheorem{theorem}{Theorem}[section]
\newtheorem{lemma}[theorem]{Lemma}

\newtheorem{proposition}[theorem]{Proposition}

\theoremstyle{definition}
\newtheorem{assumption}[theorem]{Assumption}
\newtheorem{definition}[theorem]{Definition}
\newtheorem{remark}[theorem]{Remark}
\newtheorem{example}[theorem]{Example}

\makeatletter
\newcommand{\leqnomode}{\tagsleft@true}
\newcommand{\reqnomode}{\tagsleft@false}
\makeatother

\title{Attribution of Spurious Factors from High-Dimensional Functional Time Series}

\author[1]{Adam Nie}
\author[1]{Yanrong Yang}
\author[2]{Han Lin Shang}
\author[3]{Yi He}

\affil[1]{Research School of Finance, Actuarial Studies and Statistics\\
The Australian National University, Australia\\
\texttt{adam.nie@anu.edu.au}, \texttt{yanrong.yang@anu.edu.au}}

\affil[2]{Department of Actuarial Studies and Business Analytics\\
Macquarie University, Australia\\
\texttt{hanlin.shang@mq.edu.au}}

\affil[3]{School of Mathematical Sciences\\
Eastern Institute of Technology, Ningbo, China\\
\texttt{yihe@eitech.edu.cn}}

\date{}

\begin{document}

\maketitle

\begin{abstract}
This article explores a general factor structure for high-dimensional nonstationary functional time series, encompassing a wide range of factor models studied in the existing literature. We investigate the asymptotic spectral behaviors of the sample covariance operator under this general data structure. A novel fundamental sufficient condition, formulated in terms of a newly introduced effective rank tailored to this setup, is established under which empirical eigen-analysis yields spurious results, rendering sample eigenvalues and eigenvectors unreliable for accurately recovering the underlying factor structure. This generalizes the results of \cite{OnatskiWang2021} from typical high-dimensional time series (HDTS) to the more intricate functional framework. The newly defined effective rank is rigorously analyzed through a decomposition of the effects attributable to functional factor loadings and functional factors. Contrary to the findings in the HDTS setting, empirical eigen-analysis of models with only a small number of strong non-stationary factors may still produce spurious limits in the functional framework. Therefore, additional caution is warranted when applying covariance-based statistical methods to potentially nonstationary functional data. Simulation studies are performed to determine conditions under which spurious limits occur. Real data analysis on age-specific mortality rate data from multiple locations is conducted for evidence of spurious factors induced by empirical eigen-analysis.
\end{abstract}


\section{Introduction}
Functional data analysis concerns random observations that take values in spaces of functions, typically square-integrable function spaces over compact intervals. Such data have gained substantial attention over the past few decades, both among theoreticians and researchers in applied fields \citep[see, e.g.,][]{WangChiouMueller2016, KonerStaicu2023, RamsaySilverman1997}. To address the inherent infinite-dimensional nature of functional data, dimension-reduction techniques—such as functional principal component analysis and functional factor models—are widely employed \citep{HallHosseini-Nasab2006, BenkoHardleKneip2009, Shang2014a, HallMullerWang2006a}.

Functional time series constitute a special class of functional data in which the observed random functions exhibit serial dependence. In the univariate functional time series setting, a substantial body of literature has extended classical time series methodologies to the functional framework \citep[see, e.g., and references therein]{HormannKokoszka2012, HormannKokoszka2010, HyndmanShang2009, AueNorinhoHormann2015, HorvathKokoszkaRice2014}.

The analysis of multiple functional time series, particularly in high-dimensional regimes where the number of series is comparable to or exceeds the sample size, has received growing attention; see, for example, \cite{GaoShangYang2019, Jimenez-VaronSunShang2024, HallinNisolTavakoli2023, ChangFangQiaoEtAl2024, ChangChenQiaoEtAl2024, GuoQiaoWangEtAl2024, LengLiShangEtAl2024b}. A central difficulty in this setting is the coexistence of multiple sources of dependence: within-curve dependence across the functional domain, serial dependence over time, and cross-sectional dependence across series. Each component may be high-dimensional.
For instance, the Japanese sub-national mortality data analysed in \cite{Jimenez-VaronSunShang2024} comprise of 46 functional time series observed over 50 years, with each annual observation recorded on a grid of 110 points.

Dimension-reduction methods for multivariate and high-dimensional functional time series include extensions of functional principal component analysis \citep[see, e.g.,][]{GaoShangYang2019, ChiouChenYang2014, HappGreven2018, DiCrainiceanuCaffoEtAl2009, ZapataOhPetersen2022} and functional factor models \citep[see, e.g.,][]{Jimenez-VaronSunShang2024, HallinNisolTavakoli2023, ChangFangQiaoEtAl2024, ChangChenQiaoEtAl2024, GuoQiaoWangEtAl2024, LengLiShangEtAl2024b}. These approaches approximate the data by low-rank representations and are effective when the underlying data-generating mechanism is itself approximately low-rank.


Many of the classical methods discussed above are primarily developed under the assumption of stationarity. In recent years, however, increasing attention has been devoted to the non-stationary setting \cite[see, e.g.,][]{vanDelftEichler2018, vanDelftDette2021, ChangKimPark2016, LiRobinsonShang2023, HorvathKokoszkaRice2014}.
For non-functional data, it is well established that directly applying PCA or factor analysis to high-dimensional non-stationary time series can lead to spurious conclusions. In such settings, temporal dependence may dominate cross-sectional dependence, even when strong common factors are present. As a result, the leading eigenvalues obtained from the standard PCA procedure primarily reflect serial dependence rather than genuine cross-sectional correlation. Consequently, the extracted components may fail to represent the underlying factor structure and instead capture persistent time dynamics.

This issue was rigorously studied by \cite{OnatskiWang2021}, who showed that standard PCA may produce spurious factors when applied to high-dimensional integrated time series lacking a genuine low-rank structure. Even in the absence of true common factors, the leading sample principal components can explain a large fraction of the total variance, leading to erroneous inference. They further established that, under high-dimensional asymptotics, the sample eigenvalues and eigenvectors converge to functionals of a Wiener process, irrespective of the underlying covariance structure, indicating that strong temporal dependence can dominate the spectral behavior.
Building on this framework, \cite{HeZhang2024} extended the analysis to models with more general forms of temporal dependence beyond the unit root case. They also derived the asymptotic distributions of the sample eigenvalues, which enable formal tests to distinguish genuine factors from spurious ones.
To mitigate the occurrence of spurious factors in practice, \cite{ZhangGaoPanYang2025} proposed a data-adaptive method to identify the underlying structure of high-dimensional time series before conducting eigen-analysis.

The present paper investigates the emergence of spurious phenomena in high-dimensional functional time series, with an extra functional dimension. 
We consider a general high-dimensional functional factor model with non-stationary factors that encompasses several existing frameworks in the literature \citep[see, e.g.,][]{HallinNisolTavakoli2023, TavakoliNisolHallin2023, GuoQiaoWangEtAl2024, LengLiShangEtAl2024b}.
In parallel with results from the high-dimensional literature, under mild rank-type conditions on the covariance operator of the process, the empirical eigenvalues and eigen-functions converge to a spurious limit that fails to represent the genuine covariance structure. As a result, methodologies relying on functional principal component analysis can produce inconsistent or misleading conclusions in the presence of non-stationarity.

We generalize the core results of \cite{OnatskiWang2021} to the setting of high-dimensional functional time series (see Theorem \ref{theorem - main}). The infinite-dimensional nature of functional observations introduces substantial technical challenges and gives rise to asymptotic behavior that is qualitatively different from that observed in finite-dimensional settings. In particular, the conclusions of our main theorems diverge significantly from those of \cite{OnatskiWang2021}, highlighting structural features that are intrinsic to functional data. We elaborate on these distinctions below.

\cite{OnatskiWang2021} introduce a sufficient condition for the emergence of a spurious limit based on the notion of an “effective rank,” which characterizes the extent to which the model can be well approximated by a low-rank structure. They show that models with a large effective rank are prone to exhibiting spurious limits.
In our setting, however, the concept of effective rank becomes more challenging due to an additional layer of dependence inherent in functional data. For high-dimensional (or multivariate) functional time series, the covariance operator encodes both cross-sectional dependence across coordinates and dependence along the functional domain. Disentangling these two sources of dependence requires additional technical work in order to derive results that admit a clear interpretation; see Theorem~\ref{theorem - upper lower bounds} and Corollary~\ref{corollary}.

Furthermore, the existing theoretical framework does not accommodate models with a fixed number of strong factors. In contrast, this restriction is relaxed in our setting due to the intrinsic high dimensionality of functional data. Specifically, Theorem \ref{theorem - upper lower bounds} establishes that even models driven by only a small number of strong non-stationary factors may converge to the spurious limits characterized in Theorem \ref{theorem - main}.
This constitutes a substantial departure from the settings considered in \cite{OnatskiWang2021} and \cite{HeZhang2024}, where a finite number of strong factors are treated as genuine and thus do not converge to spurious limits. Our findings indicate that, in the functional framework, the presence of an apparently well-defined factor structure does not safeguard against spurious eigenstructure in practice. Consequently, particular caution is warranted when applying covariance-based methods to non-stationary functional time series, even in the presence of strong factors.

The remainder of the paper is organized as follows. Section \ref{section - model} introduces the model under consideration, and Section \ref{section - examples} presents several illustrative examples together with their connections to the existing literature. Section \ref{section - covariance definition} defines the principal mathematical objects of interest, namely a class of sample covariance operators, and introduces the notation and preliminary results required for the development of our theory.
Section \ref{section - assumptions} states and discusses the main assumptions underlying our theoretical analysis, while the primary results are presented in Section \ref{section - main results}. Further discussion of the key sufficient condition, Assumption \ref{assumptions - Omega rank}, is provided in Section \ref{section - remarks on assumption 2}, where we also establish additional theoretical results that clarify the circumstances under which spurious limits arise.
Section \ref{section - simulation} reports simulation results and examines their implications in light of the theory. In Section \ref{section - empirical}, we apply our methodology to empirical datasets. Finally, Section \ref{section - conclusion} summarizes the main findings and proposes a conjecture motivated by several notable simulation patterns for future investigation.

Finally, we will introduce some notations for the rest of the paper. 
Throughout the paper, for a separable Hilbert space $\cH$, we write $\<\cdot, \cdot\>_\cH$ and $\norm{\cdot}_\cH$ for its inner product and the associated norm. 
We will routinely drop the subscript $\cH$ and simply write $\<\cdot, \cdot\>$ and $\norm{\cdot}$ when the context is clear. 
For elements $f,g$ on $\cH$, we write $f\otimes g$ for the bounded linear operator $h\mapsto f \<h, g\>$. For a vector $x = (x_1, \ldots, x_n)\in\Rm^n$, we write
$\norm{x}_0=|\{ i\le n, x_i\ne0\}|$ for the number of non-zero coordinates of $x$  and $\norm{x}_{p} = (\sum_{i\le n} x_i^p)^{1/p}$ for the usual $\ell^p$ norm of $x$.
For $p\in\mathbb N$, we write $\cH^p = \oplus_{k=1}^p \cH$ for the external direct sum of $p$ copies of $\cH$ endowed with the usual inner product. Elements of $\cH^p$ will be denoted as $f = (f_1, \ldots, f_p)$ where $f_i\in\cH$ for $1\le i\le p$. 
For a bounded linear operator $A$ on $\cH$, we write $\norm{A}$ and $A^*$ for the operator norm and the adjoint of $A$ respectively. The rank, the trace, the nuclear norm and the Hilbert-Schmidt norm of $A$ will be denoted as $\norm{A}_0$,
$\<A\>$, $\norm{A}_1$ and $\norm{A}_2$ respectively, whenever they are finite. For Hilbert-Schmidt operators $A,B$ on $\cH$, we write $\<A,B\> = \<A^*B\>$ for the natural inner product in the Hilbert space of Hilbert-Schmidt operators on $\cH$. 
For sequences $(a_n)_{n\in\mathbb N}$ and $(b_n)_{n\in\mathbb N}$, we write $a_n\lesssim b_n$ if $a_n\le c b_n$ for all $n$ for some positive constant $c$ independent of $n$. We write $a_n\sim b_n$ if $a_n\lesssim b_n$ and $b_n\lesssim a_n$. We write $a_n\ll b_n$ or $a_n = o(b_n)$ if $a_n/b_n \to 0$ as $n\to\infty$ and $a_n\gg b_n$ if $b_n\ll a_n$. Finally, we use $O_p$ and $o_p$ to denote the usual notions of stochastic compactness and convergence in probability.

\section{Setup: Model and Assumptions}

\subsection{A General Functional Factor Model}\label{section - model}

Let $\cH=L^2(\cI)$ be the Hilbert space of square integrable functions on some compact interval $\cI$ and $\cH^p=\oplus_{k=1}^p \cH$ be the external direct sum of $p$ copies of $\cH$ endowed with the usual inner product. 
We consider an integrated $p$-dimensional functional time series $X = (X_1, \ldots, X_T)$ of length $T$ taking values on $\cH^p$. We assume $X$ follows the functional factor model below
\begin{align*}
	X_{t}(u)  = (\Psi F_t)(u) + \zeta_t(u),\quad u\in\cI,
	\numberthis\label{equation - model}
\end{align*}
where $F_t$ is an integrated functional time series on $\cH^K$ for some $K<p$, the function $\Psi$ is a bounded linear operator from $\cH^K$ to $\cH^p$, and $(\zeta_t)_t$ is a stationary functional time series on $\cH^p$ with mean zero and covariance operator $C_\zeta$. 
It is natural to view \eqref{equation - model} as a functional factor model with non-stationary factors $F_t$ and factor loading operator $\Psi$. 
This functional factor model takes a general form in the sense of  both factors and loadings being functions, which covers common factor styles in relevant literature \cite{HallinNisolTavakoli2023, GuoQiaoWangEtAl2024} and \cite{LengLiShangEtAl2024b}.

Analogous to \cite{OnatskiWang2021}, the functional factor $(F_t)_t$ is assumed to be a random walk on $\cH^K$ of the form
$F_{t} = \sum_{s=1}^t \epsilon_s$
where $\epsilon_{t} = (\epsilon_{1t}, \ldots, \epsilon_{Kt})\tp$ 
and $(\epsilon_{kt})_{kt}$ is a collection of independent random functions in $\cH$ with mean zero and covariance operator $C_\epsilon$. 
The linear operator $\Psi$ is assumed to be a matrix of integral operators, i.e., there exists a collection
$(\Psi_{ik})_{i\le p, k\le K}$ of integral operators such that
\begin{align*}
	(\Psi \epsilon_t)(u)_i = 
	\sum_{k=1}^K \int_\cI \Psi_{ik}(u,v) \epsilon_{kt}(v) dv,\quad u\in\cI,
\end{align*}
where $\Psi_{ik}(\cdot, \cdot)$ denotes the integral kernel of the integral operator $\Psi_{ik}$. This structure is also adopted by \cite{LengLiShangEtAl2024b}. 
We will write $\Psi(\cdot, \cdot)$ as the matrix of kernel functions $(\Psi_{ik}(\cdot, \cdot))_{ik}$ and
write the above equation in the matrix form
\begin{align*}
	(\Psi \epsilon)(u) = \int_\cI\Psi(u,v) \epsilon(v) dv,\quad u\in\cI. 
\end{align*}
Additional technical assumptions on the model are discussed in Section \ref{section - main results}.

\subsection{Model Generality}\label{section - examples} 
Our framework is built upon a highly general functional factor structure in which both the factors $F$ and the factor loadings $\Psi$ are functional objects. This formulation is conceptually related to the dual functional factor model proposed by \cite{LengLiShangEtAl2024b}; however, their analysis is restricted to stationary time series.
The fully functional factor structure specified in \eqref{equation - model} encompasses a broad class of existing models as special cases. We now illustrate this generality by presenting several concrete examples and discussing their connections to the literature.
\begin{example}\label{eg tavakoli}(Functional Loadings).

	Suppose the factors $F_{kt}$ takes the form of
	\begin{align*}
	 	F_{kt}(\cdot) = \tilde F_{kt} g(\cdot)
	\end{align*} 
	for some deterministic function $g\in \cH$ and real random variables $(\tilde F_{kt})_{kt}$. Set 
	$$\td \Psi(\cdot) = \int \Psi(\cdot, v) g(v) dv,$$ 
	then our factor structure in  \eqref{equation - model} reduces to the functional factor model of \cite{HallinNisolTavakoli2023} and \cite{TavakoliNisolHallin2023} with factors $(\tilde F_{kt})_{kt}$ and factor loading operator $\td \Psi$. In this case the factor loading is a functional object while the factors is given by a multivariate time series on $\Rm^K$.
\end{example}

\begin{example}\label{eg Guo}(Functional Factors).

	Suppose the factor loadings are given by 
	\begin{align*}
	 	\Psi_{ik}(u,v) = \delta(u-v) \td \Psi_{ik}
	\end{align*} 
	where $\td \Psi = (\td\Psi_{ik})_{ik}$ is a $p\times K$ real matrix and $\delta$ is the Dirac $\delta$ function at zero. Then we have
	\begin{align*}
		(\Psi F_t)(u) = \int_\cI \delta(u-v) \td\Psi F_t(v) dv 
		= \td \Psi F_t(u)
	\end{align*} 
	and our factor structure in \eqref{equation - model} reduces to the functional factor model recently proposed by \cite{GuoQiaoWangEtAl2024}. In this case the factors are functional objects while the factor loadings are given by a  $p\times K$ real matrix. 
\end{example}

We note that, under the factor structure in Example \ref{eg tavakoli}, both the cross-sectional dependence and the dependence along the functional domain are entirely induced by the factor loadings. In contrast, under the structure in Example \ref{eg Guo}, the cross-sectional dependence is determined by the factor loading matrix, whereas the dependence along the functional domain is jointly driven by the latent factors and the loading matrix.
In our general framework, both types of dependence emerge from the interaction between the factors and the factor loadings, as will become evident from the main theoretical results.

The fully functional factor structure is also closely related to the study of non-stationarity and cointegration in functional time series, topics that have attracted increasing attention in the recent literature \citep[see, e.g.,][and references therein]{BeareSeoSeo2017, ChangKimPark2016}. We present a simplified example below to illustrate these connections.


\begin{example}(Cointegrated Functional Time Series).

Suppose $p=K=1$ and $\Psi$ is a bounded linear operator on $\cH$ of finite rank. Writing
\begin{align*}
\Delta X_t = \Psi \epsilon_t + \Delta \zeta_t,
\end{align*}
our model in~\eqref{equation - model} becomes the the Beveridge-Nelson decomposition of the cointegrated functional time series studied in \cite{BeareSeoSeo2017} and \cite{ChangKimPark2016}. The rank of $\Psi$ here represents the dimension of the attractor space, i.e. the vectors $h$ in the kernel of $\Psi$ are the directions such that the projection $\< h , X_t\>$ is of $I(0)$.  Our model can be extended to be compatible with recent results that study integrated time series of order $I(d)$ where $d$ can be fractional \citep[see][]{LiRobinsonShang2023, BeareSeo2020}. We do not pursue such extensions in our current work. 
\end{example}

Finally, although our model is formulated with the factors 
$F$ evolving as a random walk on $\cH$, the framework in fact accommodates a broad class of temporal dependence structures, in the spirit of \cite{OnatskiWang2021}. In particular, the specification in \eqref{equation - model} can be interpreted as the Beveridge–Nelson decomposition of $I(1)$ linear process on $\cH$ exhibiting more general short-run dynamics.

As demonstrated by our main results, the random-walk component of the factors asymptotically dominates the behavior of the sample covariance operator, while mild deviations from pure random-walk dynamics have asymptotically negligible effects. We conclude this section with an additional example illustrating the range of temporal dependence structures encompassed by our framework.

\begin{example}\label{example - ARIMA}(Multivariate Functional Time Series).

	Suppose the factors $(F_1, \ldots, F_T)\subseteq \cH^K$ follow the model
	\begin{align*}
		F_{t}(u) = \sum_{n=1}^N   Y_{t}^n  \phi_n(u), 
	\end{align*}
	where $(\phi_n)$ is a complete orthonormal set of functions on $\cH$ and for each $n$, the vector-valued coefficients $(Y_t^{n}, t=1, \ldots, T)$ follow a $K$-dimensional vector ARIMA model of the form
	\begin{align*}
		\Phi_n(L) (1-L) (Y_t^n - Y_0^n) = \Theta_n(L)  \epsilon_t^n
	\end{align*}
	where $\Phi_n$ and $\Theta_n$ are the matrix valued autoregressive and moving average polynomials. Assuming causality of the ARIMA processes and writing $\Psi_n(L):= \Phi_n(L)^{-1} \Theta_n(L)$, the multivariate Beveridge-Nelson decomposition of $Y_t^n$ is given by
	\begin{align*}
		Y_t^n = \Psi_n(1) \sum_{s\le t} \epsilon_s^n + \zeta^n_t,
	\end{align*}
	where $(\zeta^n_t)_t$ is a stationary process given by $\zeta^n_t = -\sum_{k=0}^\infty \sum_{i=k+1}^\infty \Psi_i \epsilon_t^n$. 
	Now, let $\Psi:\cH^K\to\cH^K$ be a linear operator given by $\Psi(a \phi_n) := \Psi_n(1) a \phi_n$ where $a$ is an arbitrary vector in $\Rm^K$. The factor process can then be written as 
	\begin{align*}
		F_t(u) = \Psi \sum_{s\le t} \sum_{n=1}^N     \epsilon_s^n   \phi_n(u)
		+
		\sum_{n=1}^N \zeta_t^n \phi_n(u).
	\end{align*}
    This representation shows that, the multivariate functional time series $F_t(u)$ is also a special case of the functional factor model in (\ref{equation - model}). 
\end{example}
In light of Examples \ref{eg tavakoli}–\ref{example - ARIMA}, the proposed functional factor model accommodates a rich variety of cross-sectional and temporal dependence structures.

\subsection{The Sample Covariance Operator Matrix}\label{section - covariance definition}
The analysis of the model~\eqref{equation - model} and the estimation of factors and factor loadings typically rely on the eigen-analysis of the sample covariance operator $\frac{1}{T}\sum^{T}_{t=1}[X_t\otimes X_t]$ on $\cH^p$. 
In Euclidean space, it is common practice to leverage on the duality between the column space and the row space of a matrix and consider the $T\times T$ Gram matrix instead.
In the context of functional data, this estimation approach has also been taken in existing work \citep{BenkoHardleKneip2009,LengLiShangEtAl2024b,TavakoliNisolHallin2023} when estimating factors and factor loadings. We provide a formal justification of this duality in Lemma \ref{lemma - same spectrum} in the Appendix.

Let $\overline X(u) := \frac{1}{T}\sum_{t=1}^T X_t(u)$ be the sample mean. By Lemma \ref{lemma - same spectrum}, the sample covariance operator of $X$ has the same non-zero spectrum as the Gram
matrix $\hat S = (\hat S_{st})$ given by
\begin{align*}
	\hat S_{st}:& =  \frac{1}{p} \sum_{i=1}^{p} 
	\<X_{is} - \overline X_i, X_{it} - \overline X_i\>
	\numberthis\label{equation - gram matrix}
\end{align*}
Let $M$ denote the orthogonal projection onto the orthogonal complement of the $T$-dimensional vector of ones, i.e. $M:= I_T - T^{-1}1_T 1_T'$. Then we may write
\begin{align*}
	\hat S & =  \frac{1}{p} \int_\cU (X(u) - \overline X(u))'(X(u) - \overline X(u)) du
	 =   \frac{1}{p} \int_\cU M X(u)' X(u)M du
\end{align*}
As will be shown, in cases where spurious behaviors arise, this Gram matrix is dominated by the Gram matrix of the non-stationary part of the original data, which we will denote as 
\begin{align*}
	\tilde S:& =  \frac{1}{p} \int_\cU  M(\Psi F)(u)'  (\Psi F)(u) Mdu.
	\numberthis\label{equation - tilde S}
\end{align*}
Write 
$\Theta$ for the $T\times T$ upper triangular matrix with entries equal to $1$ so that, based on the random walk structure on $F(u)$,  
\begin{align*}
	(\Psi F)(u) = (\Psi \epsilon)(u) \Theta
	= \int_\cI \Psi(u,v) \epsilon(v)  \Theta dv. 
\end{align*} 
By Fubini's theorem, 
the Gram matrix $\tilde S$ can be written as
\begin{align*}
	\tilde S
	& =  
	\frac{1}{p}M \Theta'
	\int_\cI 
	\left( \int_{\cI} \epsilon(v)' \Psi(u,v)' dv \int_\cI \Psi(u,w) \epsilon(w) dw   \right) du \Theta M
	  \\
	& =  
	\frac{1}{p} M \Theta'
	\iint_{\cI^2}  \epsilon(v)'\left( \int_\cI \Psi(u,v)'\Psi(u,w)  du \right)\epsilon(w)  dv dw \Theta M
	 =:
	\frac{1}{p} M \Theta' W \Theta M,
\end{align*}
where $W$  denotes the $T\times T$ random matrix given by
\begin{align*}
	W: = \iint_{\cI^2} \epsilon(v)' \Omega(v,w) \epsilon(w) dv dw,
	\numberthis\label{eq- def of W}
\end{align*}
and $\Omega(v,w)$ is the $K\times K$ matrix given by
\begin{align*}
	\Omega(v,w): =  \int_\cI \Psi(u,v)' \Psi(u,w) du, \quad v,w\in\cI. 
	\numberthis\label{eq- def of Omega}
\end{align*}
Note that the transpose of $\Omega$ is given by $\Omega(v,w)' =  \int_\cI \Psi(u,w)'\Psi(u,v)  du = \Omega(w,v)$, i.e. for any $i,j$ we have $\Omega_{ji} = \Omega_{ij}^*$. 
Proposition \ref{proposition - matrix of op} below shows that $\Omega$ is a self-adjoint operator on~$\cH^K$.

In the analysis of the Gram operator $\tilde{S}$, we will often encounter certain compositions between matrices of operators. For convenience, we provide a new definition on the operator product and state some of its basic properties.  
\begin{definition}\label{definition - C Omega}(A New Operator-Matrices Product)

	Let $C_1$ and $C_2$ be bounded linear operators on $\cH$ and $\Omega = (\Omega_{ij})$ be a $K\times K$ matrix of bounded linear operators on $\cH$ for some fixed $K$. 
	We define 
    $$C_1 \Omega C_2:=(C_1 \Omega_{ij} C_2)_{ij}$$
    to be the $K\times K$ matrix of operators where $C_1 \Omega_{ij} C_2$ is the usual composition between operators. The operator $C_1 \Omega C_2$ acts on $\cH^K$ by
	\begin{align*}
		(C_1 \Omega C_2 f)_i := \sum_{j=1}^K C_1 \Omega_{ij} C_2 f_j, \quad 
		(f_j)_{j=1}^K\in\cH^K.
	\end{align*}
\end{definition}

We first state some linear algebraic properties for these types of operators. The proofs can be found in Section~\ref{section - proofs}.
\begin{proposition}\label{proposition - matrix of op}
	Let $C_1$ and $C_2$ be given as in Definition \ref{definition - C Omega}, and $\Omega$ is defined in (\ref{eq- def of Omega}). Then
	\begin{enumerate}
		\item \label{proposition - matrix of op - adjoint}
		The adjoint of $\Omega$ is given by $(\Omega^*)_{ij} = \Omega_{ji}^*$, hence $\Omega$ is self-adjoint if and only if $\Omega^*_{ij} = \Omega_{ji}$. Furthermore, the adjoint of $C_1 \Omega C_2$ is given by 
		\begin{align*}
			(C_1 \Omega C_2)^* =C_2^* \Omega^* C_1^*.
		\end{align*}
		\item \label{proposition - matrix of op - trace}
		Suppose $C_1 \Omega_{ii} C_2$ is a trace-class operator on $\cH$ for all $i=1, \ldots, K$. Then $C_1 \Omega C_2$ is a trace class operator on $\cH^K$ and its trace, denoted as $\<C_1 \Omega C_2\>$, is given by
		\begin{align*}
			\<C_1 \Omega C_2\> =  \sum_{i=1}^K \<C_1 \Omega_{ii} C_2\>.
		\end{align*}

		\item \label{proposition - matrix of op - HS}
		Suppose the operator $C_1 \Omega_{ij} C_2$ is Hilbert-Schmidt on $\cH$ for all $i,j=1, \ldots, K$. Then the operator $C_1 \Omega C_2$ is Hilbert-Schmidt on $\cH^K$ and
		\begin{align*}
			\norm{C_1 \Omega C_2}_2^2:= \langle C_1 \Omega  C_2 C_2^*\Omega^* C_1^* \rangle
			=\sum_{ij} \norm{C_1 \Omega_{ij} C_2}_2^2. 
		\end{align*}
	\end{enumerate}
\end{proposition}
\begin{remark}
We introduce this new form of product to provide a convenient representation of the covariance structure for a large collection of functional time series. In typical settings, $C_1$ and $C_2$ correspond to the covariance operators of the latent factors, whereas $\Omega$ captures the functional dependence present in the factor loadings. The interaction between these two sources of dependence is described through the product defined above.
\end{remark}

\section{Asymptotic Theory for Spurious Phenomenon}\label{section - main results}
\subsection{Assumptions}\label{section - assumptions}
The following set of assumptions are made throughout the paper. 
\begin{assumption}[Moments Condition]\label{assumptions}
	Assume $p\to\infty$ as $T\to\infty$. Furthermore, suppose that
	\begin{enumerate}
		\item \label{assumptions - epsilon iid}
		The process $\epsilon:=(\epsilon_{it})_{it}$ is a $K\times T$ matrix of independent random elements on $\cH$ with $\Em[\epsilon_{it}] = 0$ and
		covariance operator $C_\epsilon:=\Em[\epsilon_{it}\otimes \epsilon_{it}]$ for all $i,t$.

		\item \label{assumptions - epsilon 4th moemnt} The fourth moment of $\epsilon_{it}$ is uniformly bounded in the sense that
		\begin{align*}
			\kappa_4^*:= \sup_{i,t}\sup_{u,v, \alpha, \beta}
			\frac{\Em[\epsilon_{it}(u) \epsilon_{it}(v) \epsilon_{it}(\alpha) \epsilon_{it}(\beta)]}{C_\epsilon(u,v) C_\epsilon(\alpha, \beta)}<\infty.
		\end{align*}
	\end{enumerate}
	
\end{assumption}
\begin{remark}
The condition $p\to\infty$ places the analysis in a high-dimensional regime, while no restriction is imposed on the relative growth rates of $p$ and $T$. 

Assumption \ref{assumptions}.\ref{assumptions - epsilon iid} ensures that the latent factors $\epsilon$ follow a covariance stationary process with $K$ independent components. 
The requirement that the elements of $\epsilon$ be independent across $t=1, 2, \ldots, T$ may appear restrictive. However, this assumption can be readily relaxed to accommodate common time-series structures such as ARMA models; see Example \ref{example - ARIMA} in Section \ref{section - examples}. The present formulation of the factor structure is adopted for convenience, as it clearly distinguishes the $I(1)$ and $I(0)$ components of the model, which is the main focus of this article.

Assumption \ref{assumptions}.\ref{assumptions - epsilon 4th moemnt} simply ensures that the latent factors have uniformly bounded $4$\textsuperscript{th} moment. The appearance of the covariance function $C_\epsilon$ in the denominator takes into account the fact that $\epsilon_{it}$ is an element of a Hilbert space and has a nontrivial covariance structure. 
\end{remark}

The following assumption on the covariance operators $C_\epsilon$ and $\Omega$ outlines the key condition for the emergence of spurious behaviours in sample covariance-based eigen-analysis.  

\begin{assumption}[Effective Rank]\label{assumptions - Omega rank} 
	As $p, T\to\infty$, the effective rank, based on the operators $C_\epsilon$ and $\Omega$, satisfies
\begin{equation}\label{equation - key condition}
\mathcal{R}:= \langle C_\epsilon^{1/2} \Omega C_\epsilon^{1/2} \rangle/ \norm{C_\epsilon^{1/2} \Omega C_\epsilon^{1/2}} \to\infty.
\end{equation}
\end{assumption}
\begin{remark}
The ratio $\mathcal R$ is the effective rank of the operator $C_\epsilon^{1/2} \Omega C_\epsilon^{1/2}$, a quantity widely used in the high dimensional probability literature \citep{KoltchinskiiLounici2016, RudelsonVershynin2010, Wainwright2019}. A large effective rank indicates that the operator cannot be well approximated by low-rank operators in the operator norm.

It is straightforward to see that the operator $\cC:= C_{\epsilon}^{1/2} \Omega C_\epsilon^{1/2}$, as defined in Definition \ref{definition - C Omega}, is a self-adjoint, non-negative definite operator on $\cH^K$. 
This implies $\norm{\cC}^2\le\norm{\cC}_2^2 \le  \langle \cC \rangle \norm{\cC}$, from which we obtain the estimates ${\langle \cC \rangle}/{\norm{\cC}_2} \le {\langle \cC \rangle^2}/{\norm{\cC}_2^2} \le {\langle \cC \rangle^2}/{\norm{\cC}^2}$. Therefore, Assumption~\ref{assumptions - Omega rank} is equivalent to the condition
\begin{align*}
\norm{C_{\epsilon}^{1/2} \Omega C_\epsilon^{1/2}}_{2}\ll \langle C_\epsilon^{1/2} \Omega C_\epsilon^{1/2}\rangle.			 
\end{align*}
\end{remark}

We now offer several remarks on the interpretation of Assumption~\ref{assumptions - Omega rank}.
\begin{remark}
In the non-functional settings considered by \cite{OnatskiWang2021} and \cite{HeZhang2023}, the effective rank $\mathcal{R}$ is bounded below by the number of non-stationary factors~$K$. Consequently, condition~\eqref{equation - key condition} is ensured under the stronger but interpretable assumption that~$K\rightarrow\infty$. This is also the framework within which \cite{OnatskiWang2021} and \cite{HeZhang2024} develop the asymptotic theory of spurious phenomena.

We work directly under the intrinsic condition~\eqref{equation - key condition} in the functional setting. One reason is that the lower bound $\mathcal R \gtrsim K$ no longer holds due to the presence of the operator $C_\epsilon$. In fact, condition~\eqref{equation - key condition} reflects both the cross-sectional and the functional dependence structures in the data; see Section \ref{section - remarks on assumption 2} for further discussion. This represents a key distinction between the present paper and the existing literature. We show that models with $K \sim p$ factors may exhibit no spurious behaviour, while models with only a few factors may display pronounced spurious effects. Consequently, even models with a small number of genuine non-stationary factors, eigen-analysis based on the sample covariance can produce spurious results in the functional setting. Extra caution is therefore required when applying eigen-analysis to potentially non-stationary functional time series.
It is worth noting that, in Section \ref{section - remarks on assumption 2}, we derive more interpretable results by establishing upper and lower bounds for $\mathcal R$ in common time series models. These results provide a clearer insight into when spurious behaviour arises and how it relates to $K$.
\end{remark}
 
Finally, the following assumption gives a lenient bound on the size of the error time series $\zeta$. 
\begin{assumption}[Error Components]\label{assumptions - zeta} 
Assume that the error term $(\zeta_t, t=1,\ldots, T)$ is a weakly stationary time series on $\cH^p$ with mean zero and covariance operator $C_\zeta  = \Em[\zeta_t\otimes \zeta_t]\in\cH^p$ satisfying
\begin{align*}
\<C_\zeta\> \ll T \<C_\epsilon \Omega\>,\quad T\to\infty. \numberthis\label{equation - condition on zeta}
\end{align*}
\end{assumption}
\begin{remark}
Condition~\eqref{equation - condition on zeta} ensures that $\zeta$ is negligible asymptotically compared to the non-stationary component of the model. 
To illustrate the scope of this condition, suppose that the components $\zeta_{1t}, \ldots, \zeta_{pt}$ of $\zeta_t$ have uniformly bounded covariance operators, in which case condition~\eqref{equation - condition on zeta} is equivalent to $\<C_\epsilon \Omega\> \gg pT^{-1}$. Drawing a parallel with the high dimensional factor modelling literature, it is natural to assume that the strength of the factors increases with $p$, i.e. $\norm{\Omega}\sim p^{1- \delta}$ for some $\delta\in[0,1]$  \citep[see][]{LamYao2012, FanLiLiao2021}. Since $C_\epsilon$ is a trace class operator, condition \eqref{equation - condition on zeta} reduces to $p^\delta T^{-1} = o(1)$, which clearly holds for a wide range of high-dimensional factor models. In particular, condition \eqref{equation - condition on zeta} is trivially satisfied when the factors are ``strong'' in the sense of $\delta=0$. 
\end{remark}

\subsection{Asymptotic Theory}

We are now ready to state the main result of our work.
\begin{theorem}[Asymptotic Spectral Limits]\label{theorem - main}
Let $(\hat\lambda_t)_{t=1}^{T}$ be the eigenvalues of the Gram matrix $\hat S$ arranged in non-ascending order and $(\hat u_n)$ be the corresponding eigenvectors. Suppose Assumptions \ref{assumptions}-\ref{assumptions - zeta} hold, then for any fixed natural number $k$, as $p, T\to\infty$, we have
\begin{enumerate}[(i)]
\item the sample eigenvector $\hat u_k$ satisfies
\begin{align*}
|\<\hat u_k, d_k\>|=1+ o_p(1)
\end{align*}
where $d_k\in \Rm^T$ with the $t$-th coordinate equal to $d_{kt}  = \sqrt{2/T} \cos(\pi k t/T)$;\label{theorem - evec}
\item the sample eigenvalue $\hat \lambda_k$ satisfies
\begin{align*}
\hat \lambda_k = \frac{T^2}{k^2\pi^2 p}\langle C_\epsilon \Omega \rangle  (1+ o_p(1));
\end{align*}
\label{theorem - ev}
\item the percentage of variance explained by $\hat\lambda_k$ satisfies
\begin{align*}
\frac{\hat \lambda_k}{\sum_j \hat \lambda_j} = \frac{6}{(k \pi)^2} + o_p(1). 
\end{align*}
\label{theorem - ev/ev}		
\end{enumerate}
\end{theorem}

\begin{remark}[Spurious Eigenvalues]
Theorem~\ref{theorem - main} establishes the asymptotic limits of the sample eigenvalues and eigenvectors as $T \to \infty$. Part~\eqref{theorem - ev} shows that the largest sample eigenvalues $\hat{\lambda}_k$ converge to deterministic limits. Several observations follow from this result.

First, the true covariance structure of the model, characterized by $C_\epsilon$ and $\Omega$, enters the limit of $\hat{\lambda}_k$ only through a tracial quantity. In particular, the limit of $\hat{\lambda}_k$ depends on the index $k$ solely through the factor $k^{-2}$ and is not related to the $k$-th eigenvalue of $C_\epsilon^{1/2}\Omega C_\epsilon^{1/2}$. Consequently, the sample eigenvalues $\hat{\lambda}_k$ fail to capture meaningful information about the cross-sectional dependence structure of the model.

Second, for any fixed $k$, the empirical eigenvalue $\hat{\lambda}_k$ is of the order $T^2 p^{-1}\langle C_\epsilon \Omega \rangle$, regardless of the actual number of factors present in the model. In particular, although the empirical eigenvalues appear to decay rapidly at the rate $k^{-2}$, any inference on the number of factors based on these eigenvalues would be misleading.

Finally, recall that in the factor-modeling literature the term strong factors typically refers to the regime $|C_\epsilon \Omega| \sim p$. However, even under this assumption, the leading empirical eigenvalue $\hat{\lambda}_k$ remains of order $T^2$, which is substantially larger than the strength of the factors. This indicates that the non-stationarity in the model dominates the population covariance structure, even in the presence of strong factors.

Part~\eqref{theorem - ev/ev} of the theorem provides the asymptotic limits of the proportion of variance explained by each sample eigenvalue. These limits are deterministic and independent of the true factor structure, represented by $C_\epsilon$ and $\Omega$. Consequently, the scree plot carries no useful information about factor strength, and any inference based on it would be unreliable.
\end{remark}

\begin{remark}[Spurious Eigenfunctions]
Part~\eqref{theorem - evec} of the theorem states that the estimated eigenvectors $\hat u_k$, which are used as factor estimates in certain models \citep[see][]{TavakoliNisolHallin2023, LengLiShangEtAl2024b}, converge in probability to a deterministic trigonometric function $d_k$. Importantly, this limit depends only on the index $k$ and is independent of the population covariance structure of the model. Consequently, factors estimated in this manner are spurious, as they bear no relation to the true factors $\epsilon$ or to the factor loadings $\Psi$.

The appearance of trigonometric functions in the limit is not surprising in this context. These functions coincide with the eigenfunctions arising in the Karhunen–Loève expansion of a demeaned standard Wiener process, which itself appears as the limit of a suitably scaled random walk. As discussed above, in the spurious regime the covariance structure is dominated by the non-stationarity of the model rather than by the cross-sectional dependence. As a result, the estimated factors effectively recover the eigenfunctions of the Wiener process instead of the true latent factors.
\end{remark}

\subsection{Separable Effects from Factors and Factor Loadings}\label{section - remarks on assumption 2} 

Assumption~\ref{assumptions - Omega rank} provides a sufficient condition for the emergence of the spurious limit, in a similar sense to that considered in the preceding works of \cite{OnatskiWang2021} and \cite{HeZhang2024}. The key quantity is the effective rank $\langle \cC \rangle / ||\cC||_2$, defined as the ratio between the trace and the Hilbert--Schmidt norm of the operator $\cC := C_\epsilon^{1/2}\Omega C_\epsilon^{1/2}$. This quantity reflects the structures of both common factors and factor loadings in the model.

Under some regular assumptions, we further study concrete bounds for the effective rank $\mathcal{R}$. This, in turn, enables us to formulate more interpretable conditions for Theorem~\ref{theorem - main}. 

In details, by Mercer's theorem, the operator $C_\epsilon$ has the spectral decomposition given by
	\begin{align*}
		C_\epsilon = \sum_{n=1}^\infty c_n \phi_n\otimes \phi_n
	\numberthis\label{equation - C special case in theorem 2}
	\end{align*}
where $(\phi_k)_k$ is a complete orthonormal basis for $\cH$ and $(c_n)_n$ is a non-negative sequence in $\ell^1$. 
Suppose that the factor loading operators $(\Psi_{ik})$ are given by
\begin{align*}
	\Psi_{ik} = \sum_{n=1}^{\infty} a_{nik} \phi_n \otimes \phi_n
	\numberthis\label{equation - Psi_ik special case in theorem 2}
\end{align*}
for some collection of real numbers $(a_{nik})_{nik}$ satisfying $\sup_n \sum_{i,k} a_{nik}^2<\infty$.
For $n\in\mathbb N$, define the matrices $A_n= (a_{nik} )_{ik}$ and $B_n= A_nA_n'$. Essentially, $c_n B_n$ is the covariance matrix of the factor part of the model $\Psi \epsilon_t$ projected component-wise to the one dimensional closed linear subspace of $\cH$ spanned by the function $\phi_n$. 

Under this specific structure (\ref{equation - C special case in theorem 2}) and (\ref{equation - Psi_ik special case in theorem 2}) for common factors and factor loadings, respectively, we can formulate the following two-sided estimates on the effective rank $\mathcal{R}:=\<\cC\>/\norm{\cC}_2$.


\begin{theorem}[Separable Bounds on Effective Rank]\label{theorem - upper lower bounds}
Let $C_\epsilon$ and $\Psi$ be given as in \eqref{equation - C special case in theorem 2} and \eqref{equation - Psi_ik special case in theorem 2}. 
Suppose that the matrices $(B_n, n\in\mathbb N)$ satisfy
$\norm{B_n}_2 \lesssim \norm{B_m}_2$ uniformly in $n,m\in\mathbb N$. 
	Then we have the upper bound
	\begin{align*}
		\mathcal{R}:=\frac{\<\cC\>\ }{\norm{\cC}_2}
		\lesssim 
		\frac{1}{\norm{C_\epsilon}_2}
		\sup_{n\in\mathbb N} \frac{\<B_n\>\ }{\norm{B_n}_2}
		\le
		\frac{1}{\norm{C_\epsilon}_2}
		\sup_{n\in\mathbb N} \norm{B_n}_0^{1/2},
		\numberthis\label{equation - upper bound on ratio in theorem 2}
	\end{align*}
	and the lower bound
	\begin{align*}
		\mathcal{R}:=\frac{\<\cC\>\ }{\norm{\cC}_2}
		\gtrsim
		\frac{1}{\norm{C_\epsilon}_2} 
		\left(  1+ \sup_n c_n \frac{\<B_n\>\ }{\norm{B_n}_2}  \right)
		\gtrsim
		\frac{1}{\norm{C_\epsilon}_2} 
		\left(  1+ \sup_n c_n \alpha(B_n)^{1/2} \norm{B_n}_0^{1/2} \right),
		\numberthis\label{equation - lower bound on ratio in theorem 2}
	\end{align*}
	where $\alpha(B_n)$ denotes the ratio between the least non-zero eigenvalues and the largest eigenvalue of $B_n$. 
	The omitted constants in the above estimates are all uniform in $T$. 
\end{theorem}

\begin{remark}
Theorem~\ref{theorem - upper lower bounds} derives upper and lower bounds for the effective rank $\mathcal{R}$, allowing the contributions of the factor loadings, captured by the sequence of matrices $(B_n, n\in\mathbb{N})$, to be distinguished from those of the common factors, represented by the operator $C_\epsilon$.

The term $\<B_n\>^2/ \norm{B_n}_2^2$ is analogous to the ratio in Assumption A3 of \cite{OnatskiWang2021}. It is bounded above by $\operatorname{rank}(B_n)$ and below by $\operatorname{rank}(B_n)$ scaled by the eigenvalue gap $\alpha(B_n)$. Intuitively, this quantity measures the effective rank of the projected data onto the direction $\phi_n$.

The quantity $\norm{C_\epsilon}_2^{-1}$ can be interpreted as an effective rank of $C\epsilon$, reflecting the decay rate of the eigenvalues $\{c_n\}$. When the eigenvalues decay rapidly, this quantity is close to $1$, whereas slow decay leads it to approach the actual rank of $C_\epsilon$, which may be infinite. For example, if $c_n = q^{-1} 1_{n\le q}$ for some fixed $q>0$, then $\norm{C_\epsilon}_2^{-1}=\sqrt q$, equal to the square root of the rank of $C_\epsilon$. In contrast, if $c_n = 2^{-n}$ for all $n$, then $\norm{C_\epsilon}_2^{-1} = \sqrt 3$ even though the rank of $C_\epsilon$ is infinite.
\end{remark}

To facilitate the discussion, it is helpful to define the following two different regimes for $C_\epsilon$.
\begin{definition}[Localization and Delocalization]\label{definition - local/delocal}
	The operator $C_\epsilon$ is said to be ``localized'' if $\norm{C_\epsilon}_2$ is uniformly bounded away from zero as $T\to\infty$, and ``delocalized'' if $\norm{C_\epsilon}_2\to0$ as $T\to\infty$.
\end{definition}
\begin{remark}
The terminology we used here is partly taken from the random matrix literature, see for example
\citep{JohnstonePaul2018,rudelson2015delocalization}. The sample eigenvectors of large dimensional random matrices typically exhibit two types of asymptotic behaviours depending on the true eigenvalues and the presence or the lack of spikes. In the de-localized regime, the sample eigenvector tends to a limit where the entries are all of a similar size, much similar to our condition $\norm{C_\epsilon}_2\to 0$.

The delocalized regime is when the factors are generated by an increasing number of basis functions with comparable weights as $p\to\infty$. A simple yet very general example of this can be seen by setting $c_n:= (aT)^{-\delta} 1_{n\le (aT)^\delta}$ for some fixed $a>0$ and  $\delta\in (0,1]$ in the spectral decomposition of $C_\epsilon$. The localized regime refers to when the factors are generated by an arbitrary number of basis functions, but the weights $(c_n)$ rapidly decay to zero. In this setting $\norm{C_\epsilon}_2^{-1}$ is bounded even though the rank of $C_\epsilon$ can be infinite. 

As previously discussed in Section \ref{section - assumptions}, 
the sufficient condition for the main theorem of \cite{OnatskiWang2021} implies that the number of stochastic trends $K$ in the data must diverge as $T\to\infty$. As a consequence, their main theorem necessarily exclude models with a finite number of strong non-stationary factors.  
This observation is analogous to the upper bounds in \eqref{equation - upper bound on ratio in theorem 2} in our case. In particular, under the localized regime defined above, condition \eqref{equation - key condition} still implies $K\to\infty$. However, this is no longer the case under the delocalized regime.
Condition \eqref{equation - key condition} can in fact be satisfied by models with a finite number of strong factors, as long as those factors are generated by an increasing number of basis functions with comparable weights as $T\to\infty$. This is an important distinction since it implies that functional time series with few strong non-stationary factors can exhibit spurious behaviour as well. We will also illustrate this through simulations in Section \ref{section - simulation}. 
\end{remark}

For ease of referencing we summarize the above discussion into the following theorem.
\begin{theorem}\label{corollary}
	Under the assumptions of Theorem \ref{theorem - upper lower bounds}, the spurious condition in (\ref{equation - key condition}) holds if either of the following conditions is satisfied
    \begin{enumerate}
        \item[Condition (a).] 
    $C_\epsilon$ is delocalized;
    \item [Condition (b).]
    $c_n \<B_n\>/\norm{B_n}_2$ diverges for some $n$. 
    \end{enumerate}
\end{theorem}
\begin{remark}
In fact, from \eqref{equation - lower bound on ratio in theorem 2} we can conclude that condition \eqref{equation - key condition} holds in the delocalized regime when all the $B_n$'s are comparable  in $\norm{\cdot}_2$, regardless of their actual rank. This is a significant departure from the conclusions in \cite{OnatskiWang2021} for the non-functional case, since it shows concretely that a model with single ($K=1$) integrated functional factor can indeed tend to the spurious limit. 

In the localized regime, a sufficient condition for \eqref{equation - key condition} is the divergence of the quantity $c_n \<B_n\>/\norm{B_n}$ for some $n$. 
By the second lower bound in \eqref{equation - lower bound on ratio in theorem 2}, this is satisfied whenever the rank of $B_n$ diverges as $T\to\infty$, and the quantities $c_n $ and  $\alpha(B_n)$ do not decay too fast. 
This concretely shows that a model with large effective rank in some direction $\phi_n$ with non-diminishing weights tends to the spurious limit, which is in line with the findings of \cite{OnatskiWang2021}.
\end{remark}

\section{Simulation Study}\label{section - simulation}
\subsection{Setup}
We start with a description of the data generating process. 
Recall from \eqref{equation - model} the model 
\begin{equation*}
X_{it}(u) = \sum_{s=1}^t\sum_{k=1}^K (\Psi_{ik} \epsilon_{ks})(u) + \zeta_{it}(u), 
\quad i=1,\ldots, p, \ t=1, \ldots, T,
\end{equation*}
where $(\epsilon_{kt})_{kt}$ are independent random functions on $\cH$ with mean zero and covariance operator $C_\epsilon$  and $(\zeta_{it})_{it}$ is a weakly stationary  functional  times series on $\cH^p$ with zero mean and covariance $C_\zeta$. For simplicity, we focus on the settings discussed in Section \ref{section - remarks on assumption 2} and Theorem \ref{theorem - upper lower bounds}.

Fix $q>0$ and let $(\phi_n)_{n=1}^q$ be a set of Fourier basis functions on $\cH$. 
We set 
\begin{align*}
    C_\epsilon = \sum_{n=1}^\infty c_n \phi_n\otimes \phi_n,
    \numberthis\label{equation - simulation/C_epsilon}
\end{align*}
where $c_n$ is a sequence of non-negative real numbers summing to one, to be specified below. Similar to \cite{LengLiShangEtAl2024b} and \cite{GuoQiaoWangEtAl2024}, we generate the factor process $\epsilon$ by setting
\begin{align*}
    \epsilon_{kt}(u) :=  \sum_{n=1}^q Z_{kt}^n \phi_n(u), 
\end{align*}
where $(Z_{kt}^n)$ is a collection of i.i.d. Gaussian variables with mean zero and variance $\Var(Z_{kt}^n) = c_n$ for all $k,t$. 
For the noise time series $\zeta$, we simply set 
\begin{align*}
    \zeta_{it}(u):= \sum_{n=1}^q W_{it}^n \phi_n(u),
\end{align*}
where $(W_{it}^n)$ is a collection of i.i.d. Gaussian variables with mean zero and variance equal to $\Var(W_{it}^n) = 2^{-n}$ for all $i,t$. 
Following \eqref{equation - Psi_ik special case in theorem 2},  we generate the factor loadings $\{\Psi_{ik}\}$ by setting
\begin{align*}
    \Psi_{ik}(u,v) = \sum_{n=1}^{q} 
        a_{ik}^n \phi_n(u) \phi_n(v),
        \numberthis\label{equation - simulation/A_n}
\end{align*}
where for each $n=1,\ldots, q$, the $p\times K$ matrix $A_n:= \{a_{ik}^n\}_{i,k}$ is to be specified below.

For the model parameters, we set $T=200$, $p=100$ and $q=20$.
Recall from  Assumption \ref{assumptions - Omega rank} and Theorem \ref{theorem - main} that the sample covariance structure tends to a spurious limit whenever the ratio
\begin{align*}
    \mathcal R = \frac{\<C_\epsilon^{1/2} \Omega C_\epsilon^{1/2}\>^2}{\norm{C_\epsilon^{1/2} \Omega C_\epsilon^{1/2}}_2^2}.
\end{align*}
diverges as $p\to\infty$. From Theorem \ref{theorem - upper lower bounds}, it is clear that whether $\mathcal R$ diverges is determined by the parameters $K$, the eigenvalues  $(c_n)_n$ in \eqref{equation - simulation/C_epsilon} and the matrices  $(A_n)_n$ in \eqref{equation - simulation/A_n}. We will perform simulations with various choices of these parameters described below. 

For the number of factors we set $K\in\{50,10,2\}$. 
In 
\cite{HeZhang2023} and \cite{OnatskiWang2021}, 
the $K=2$ setting represents as a model with genuine factors and is shown to be far from the spurious limit, while the $K=50$ setting represents a model with no genuine factors for which spurious behaviours occur. We will see that this is not necessarily the case with functional data. 

For the loading matrices $(A_n, n\in\mathbb N)$,  we will consider the high rank setting where
\begin{align}\label{high rank}
    &A_n \sim G_{p,K}
\end{align}
independently for all $n\in \mathbb N$, where $G_{p,K}$ denotes a $p\times K$ matrix with i.i.d. standard Gaussian random variables, and the low-effective rank setting where
\begin{align}\label{low rank}
    A_n \sim p^{1/2} U_{p}\cdot\mathrm{diag}&(\{2^{-n/2}\}_{n=1}^{K}, 0,\ldots, 0)\cdot U_{K}'
\end{align}
independently for all $n\in \mathbb N$,
with $U_N$ denoting an $N\times N$ random orthogonal matrix uniformly sampled from $\mathbb S^{N-1}$ for $N\in\{p,K\}$.
In the high-rank setting (\ref{high rank}), each $A_n$ has full column rank with high probability, while in the low-effective rank setting (\ref{low rank}), each $A_n$ has low effective rank with
\begin{align*}
    \frac{\<A_n A_n'\>^2}{\norm{A_n A_n'}_2^2}
    = 
    3 \frac{(1-2^{-K})^2}{1-4^{-K}} \approx 3 \ll K. 
\end{align*}
For the covariance operator $C_\epsilon$, we consider the following three scenarios by setting
\begin{align}
 \mathrm{Delocalized}: \ \    c_n  & =   q^{-1} 1_{n\le q},\label{c1}
    \\
 \mathrm{Localized}: \ \    c_n & =  2^{-n}1_{n\le q},\label{c2}
    \\
 \mathrm{Localized}: \ \    c_n & = 2^{-1} 1_{n\le 2}.\label{c3}
\end{align} 
In the first setting (\ref{c1}), each basis function $\phi_n$ has equal weight and we have $\norm{C_\epsilon}_2^2 = q^{-1} = 0.05$ which is very close to zero. This represent the scenario where the operator $C_\epsilon$ is delocalized as defined in Definition \ref{definition - local/delocal}. 
In the second setting (\ref{c2}), the eigenvalues rapidly decay to zero and we have $\norm{C_\epsilon}^2_2 = \frac{1- 4^{-q}}{3}\approx 1/3$. This represents the situation where $C_\epsilon$ is localized. Finally, a more extreme version of this is found in the last setting (\ref{c3}), where $C_\epsilon$ is completely localized on two basis functionsand we have $\norm{C_\epsilon}^2_2 = 0.5$. 

In the simulations, we consider combinations of the two settings for $A$ and the three settings for $C_\epsilon$ described above. Table \ref{table1} summarizes these  combinations of choices for $A_n$ and $C_\epsilon$. 
We also list the asymptotic orders of the effective ranks of each setting, which are obtained from Theorem 2.


\begin{table}[!htb]
\centering
\setlength{\tabcolsep}{0.12in}
\caption{Summary of Simulation Settings}
\begin{tabular}{@{}llccccll@{}}
\toprule
& & $c_n$ & $\norm{C_\epsilon}_2^2$ & $A_n$ &
$\<B_n\>^2/\norm{B_n}_2^2$ & $\mathcal R$ \\
\midrule
\multirow{2}{*}{Localized}
& Setting 1  
& $q^{-1}\mathbf{1}_{n \le q}$ 
& $q^{-1}$ 
& full rank
& $\sim K$ 
& $\sim \sqrt{qK}$
\\
& Setting 2  
& $q^{-1}\mathbf{1}_{n \le q}$ 
& $q^{-1}$ 
& low rank
& $\approx 3$ 
& $\sim \sqrt{q}$
\\
\midrule
\multirow{4}{*}{Delocalized}
& Setting 3 
& $2^{-n}\mathbf{1}_{n \le q}$ 
& $\frac{1}{3}(1-4^{-q})$ 
& full rank
& $\sim K$
& $\sim \sqrt{K}$
\\
& Setting 4  
& $2^{-n}\mathbf{1}_{n \le q}$ 
& $\frac{1}{3}(1-4^{-q})$ 
& low rank
& $\approx 3$  
& $\sim 1$
\\
& Setting 5   
& $2^{-1}\mathbf{1}_{n \le 2}$ 
& $0.5$ 
& full rank
& $\sim K$ 
& $\sim \sqrt{K}$
\\
& Setting 6  
& $2^{-1}\mathbf{1}_{n \le 2}$ 
& $0.5$ 
& low rank
& $\approx 3$
& $\sim 1$
\\
\bottomrule
\end{tabular}
\label{table1}
\end{table}
For each setting in Table~\ref{table1}, we simulate $X$ from the above data generating process with  $K\in\{50,10,2\}$ and compute the leading eigenvalues and eigenvectors of the sample Gram matrix defined in~\eqref{equation - gram matrix}. 
We plot the top five sample eigenvectors for $K\in\{50,10,2\}$ (red, green and blue lines respectively) against the spurious limit (black line) as shown in Theorem \ref{theorem - main}. We also plot the proportion of variance explained by each eigenvalue with the same color coding. 

\subsection{Results}

We start with the setting that shows obvious spurious behaviors.
Figure \ref{figure - delocal iid} plots the sample eigenvectors of Setting 1 where $C_\epsilon$ is in the delocalized regime and $A_n$ is of high rank. 
It can be seen that all sample eigenvectors as well as the sample eigenvalues closely resemble the theoretical spurious limits. 
\begin{figure}[!htb]
    \centering
    \subf{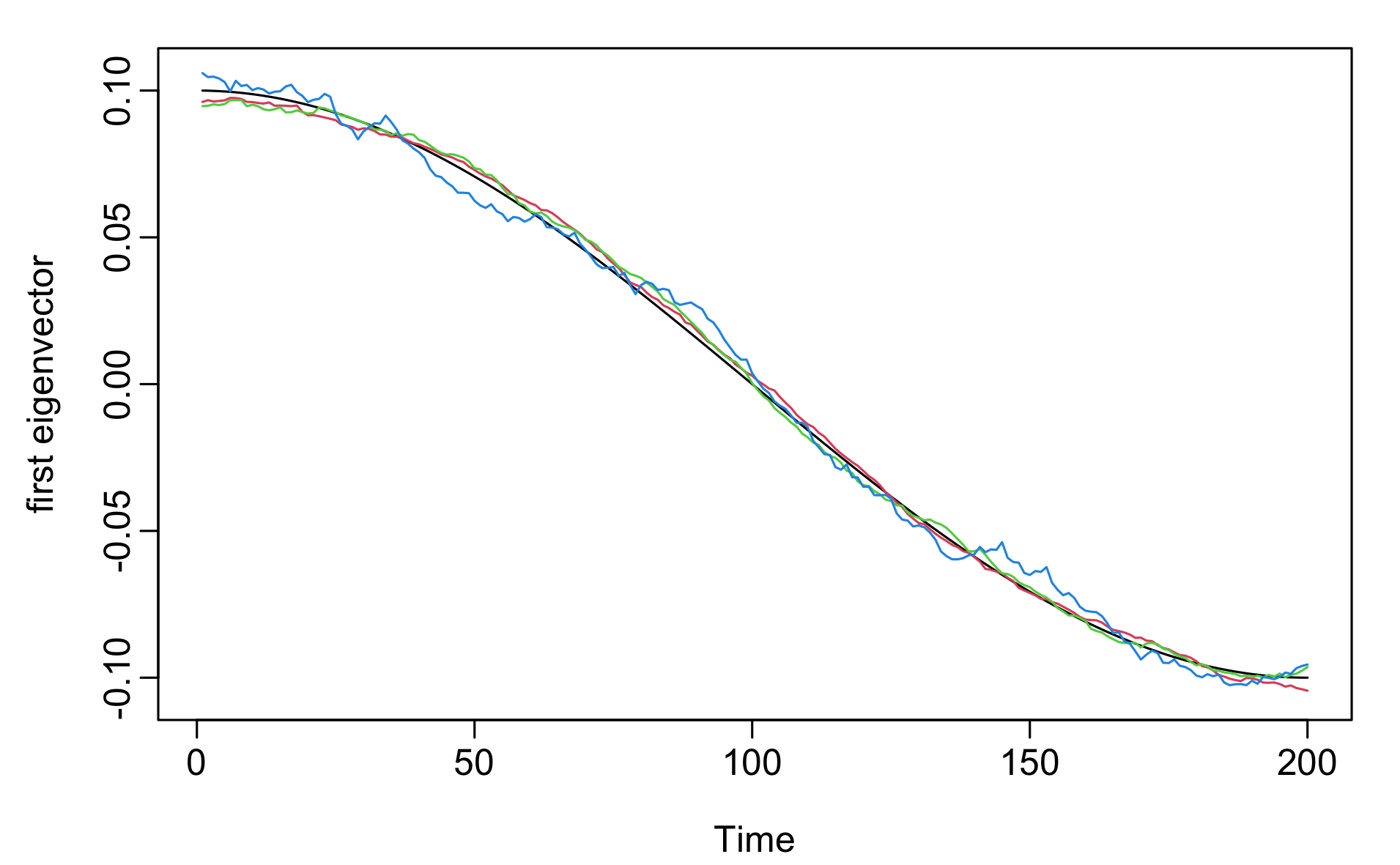}{.32}
    \subf{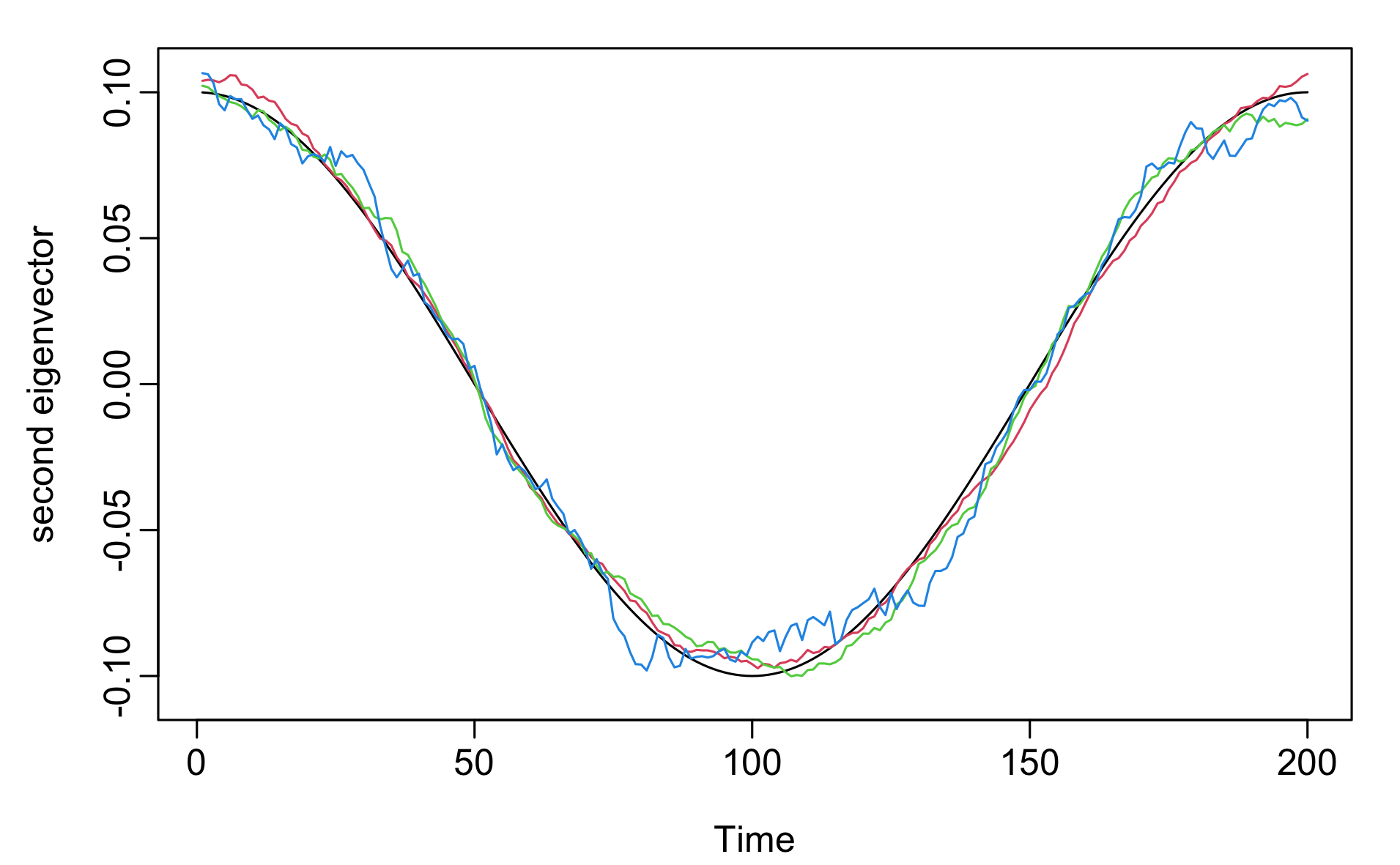}{.32}
    \subf{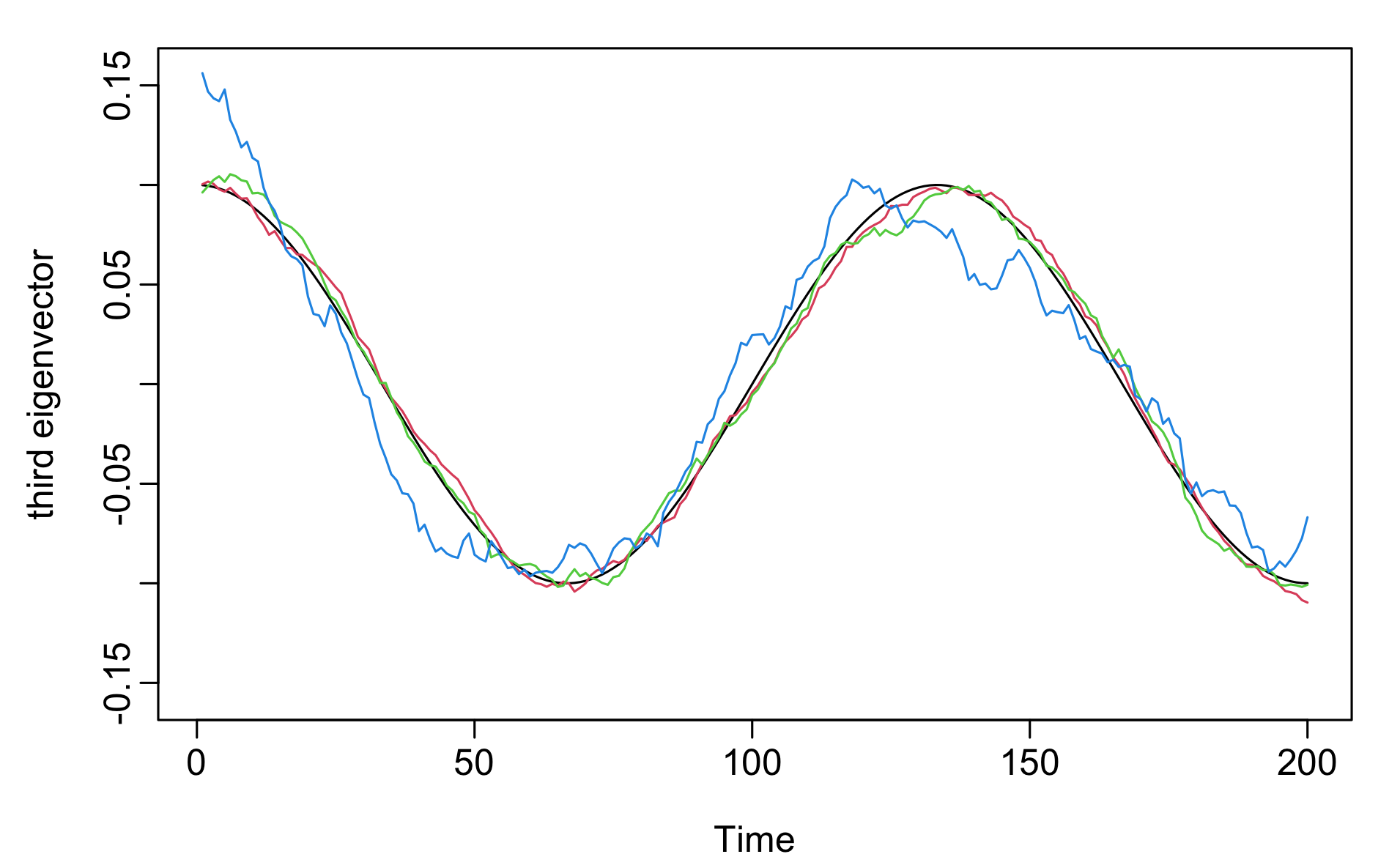}{.32}
    \subf{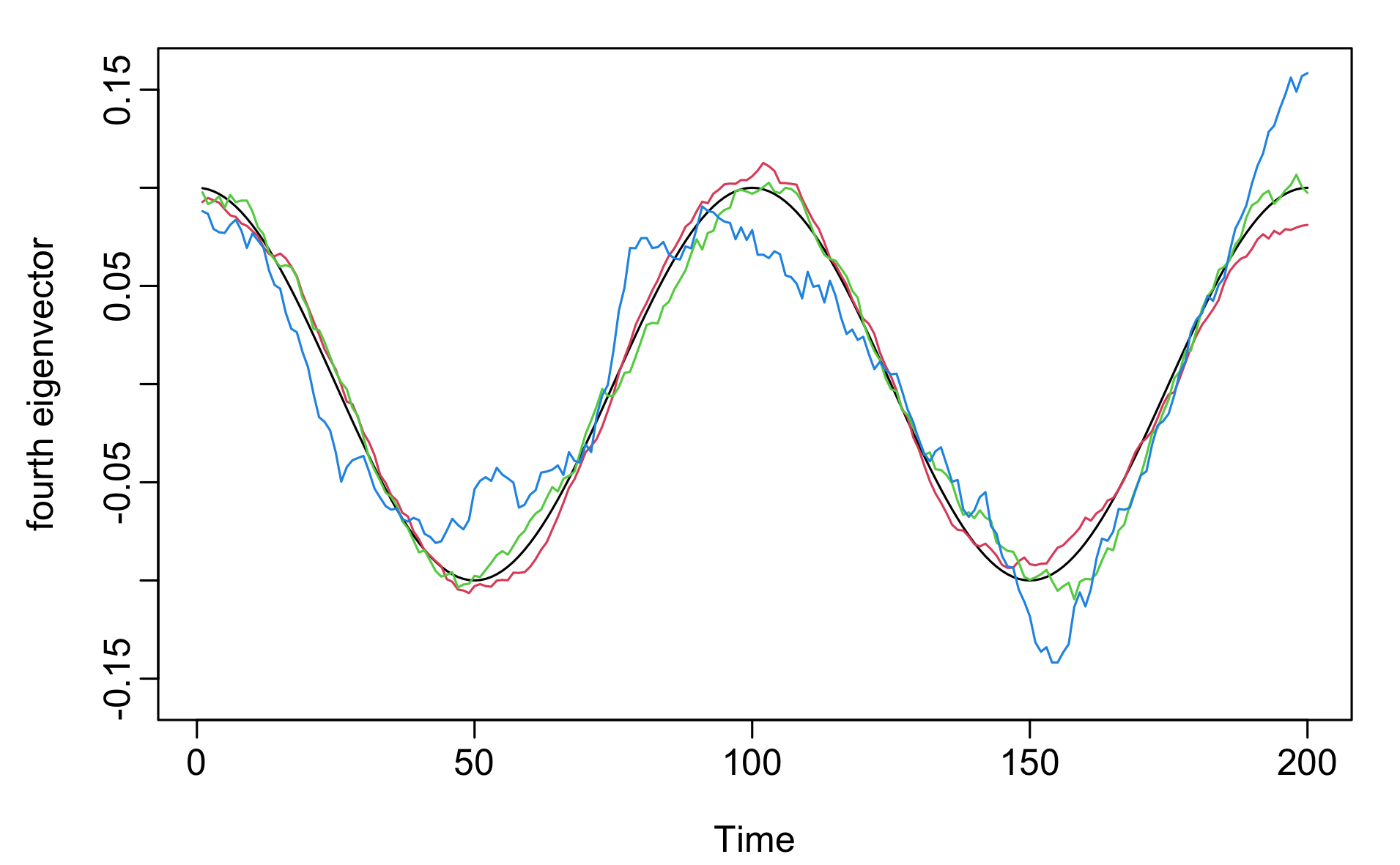}{.32}
    \subf{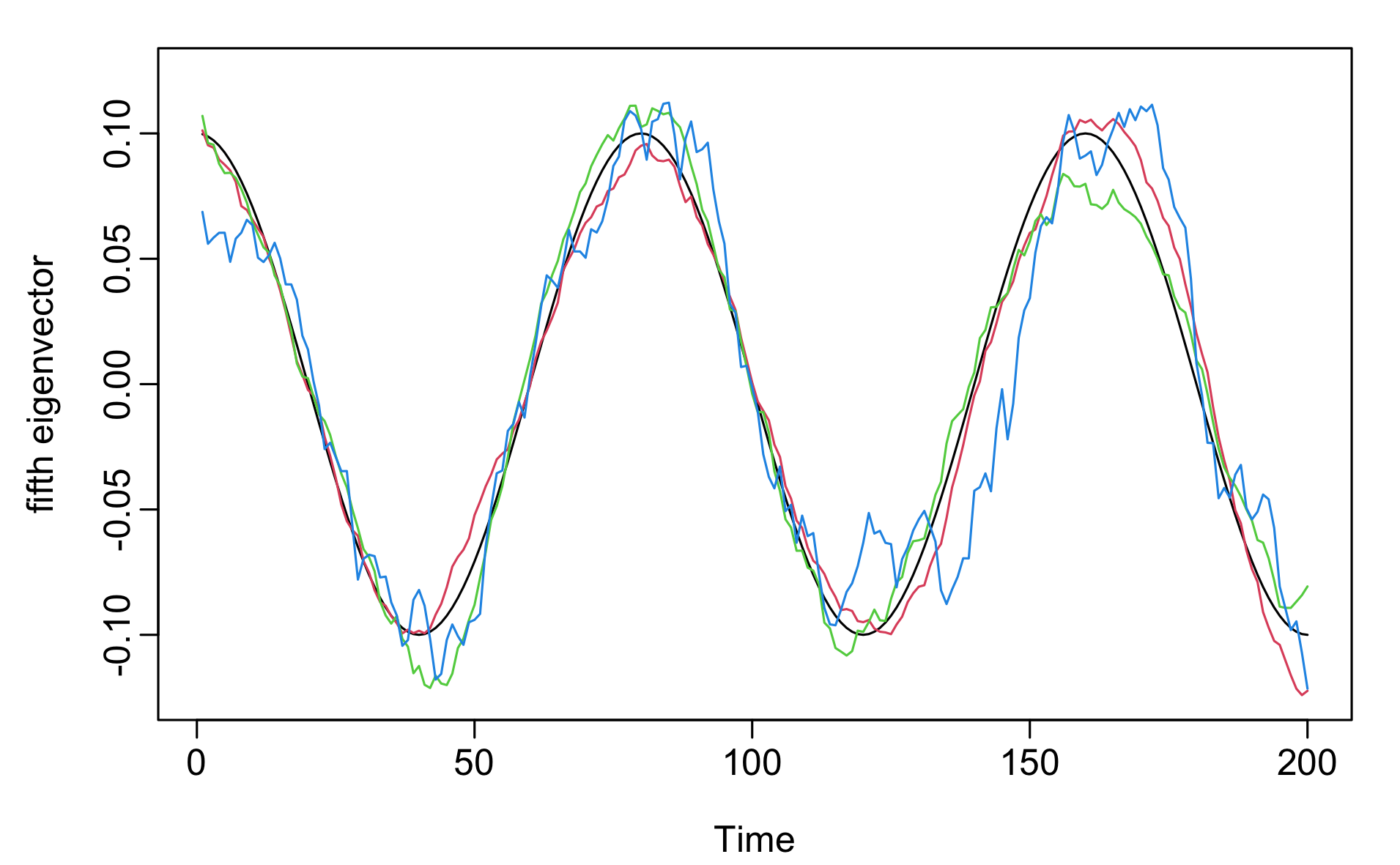}{.32}
    \subf{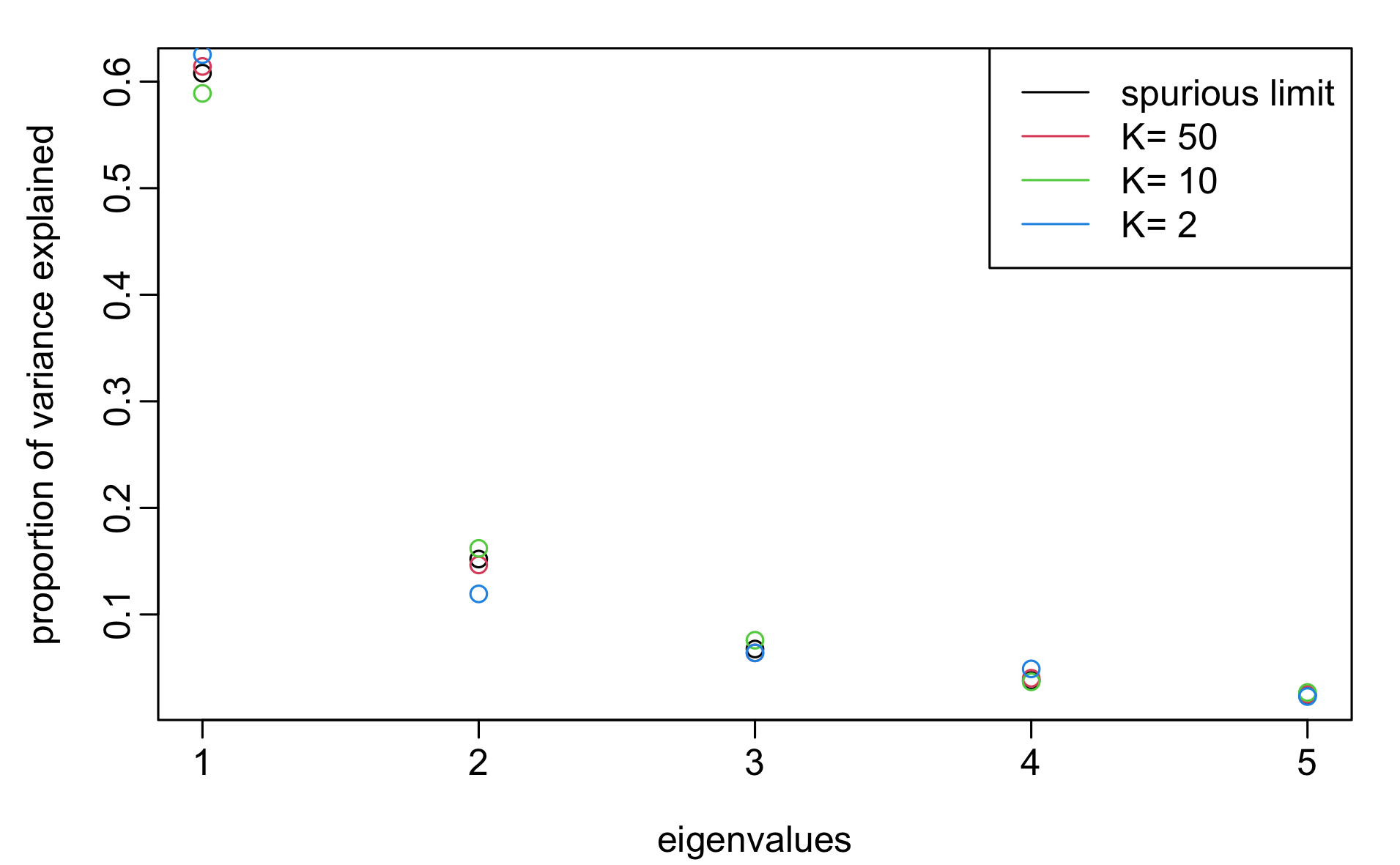}{.32}
     \caption{Setting 1 - $C_\epsilon$ delocalized,  loading matrices are of full-column rank}\label{figure - delocal iid}
\end{figure}

This behaviour is in agreement with our theoretical results, since Theorem \ref{corollary} states that spurious limits must appear whenever $C_\epsilon$ is delocalized. 
The results under Setting 2 (i.e. $C_\epsilon$ delocalized with $A_n$ having low-effective rank) are visually very similar to Figure \ref{figure - delocal iid} under Setting 1 and are hence omitted.  
Setting 2 is designed to show that these limits may arise even in the presence of only $K=2$ factors and when $A_n$ possesses low-effective rank, representing a clear departure from the findings of \cite{HeZhang2024} and \cite{OnatskiWang2021}. This underscores a key difference between HDFTS and HDTS, namely that the proliferation of basis functions in functional representations can itself introduce an additional source of dimensionality.

Figure \ref{figure - local iid} plots the results for Setting 3 where the eigenvalues of $C_\epsilon$ are fast decaying and the rank of each $A_n$ is equal to $K$. 
As can be seen from Figure \ref{figure - local iid}, the sample eigenstructure in Setting 3 closely resembles the spurious limit for $K=50$, however, as $K$ decreases, visible deviations from the spurious limit can be observed. 
This is to be expected, since for Setting 3, the first upper bound in \eqref{equation - upper bound on ratio in theorem 2} is bounded from above by a multiple of $K$. By Theorem \ref{theorem - upper lower bounds} this implies that condition \eqref{equation - key condition} is not satisfied when $K$ is finite. 
The results for this setting is in line with the findings of \cite{OnatskiWang2021} and \cite{HeZhang2024}, where a small value of $K$ represents the presence of genuine factors. 
\begin{figure}[!htb]
    \centering
    \subf{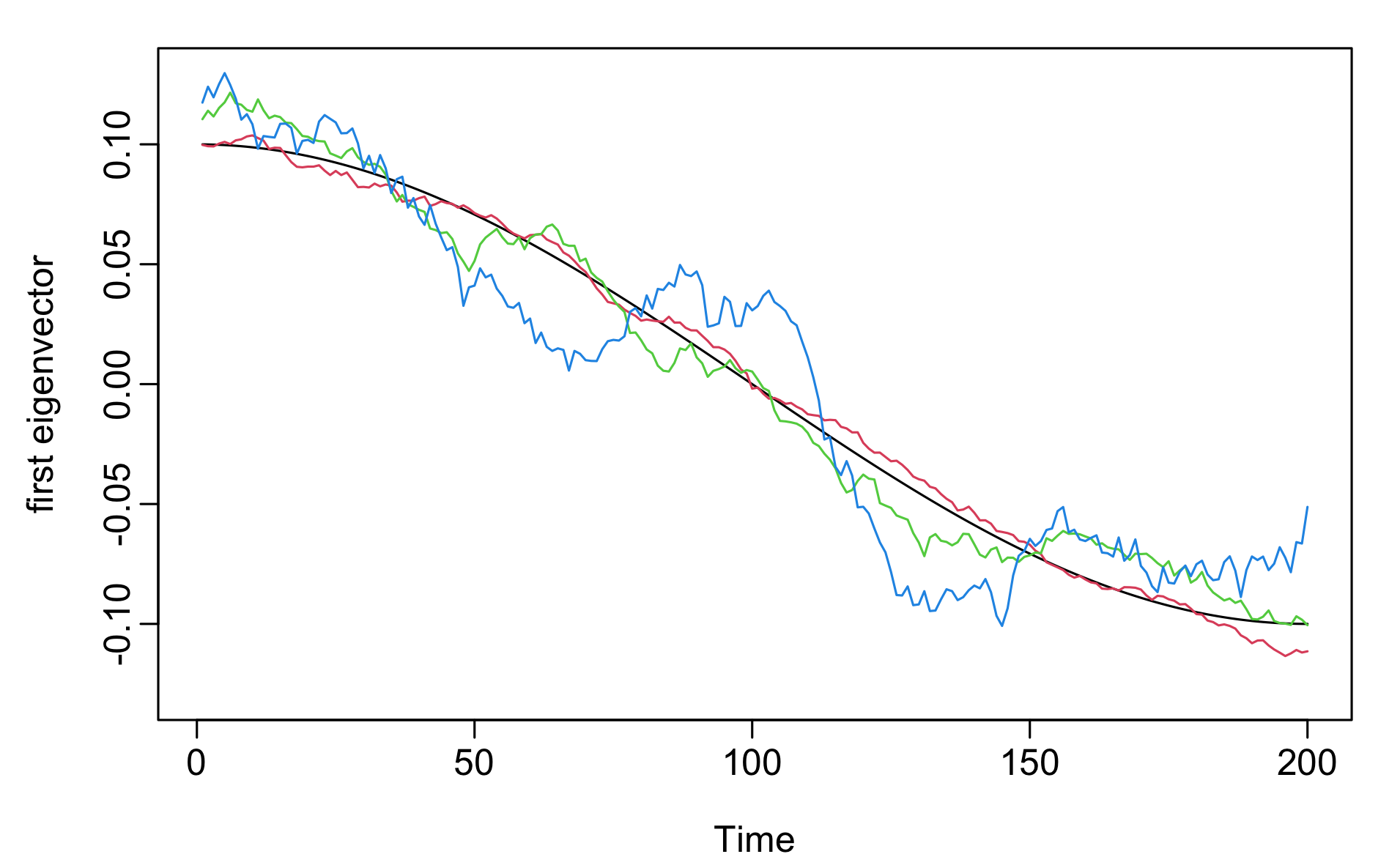}{.32}
    \subf{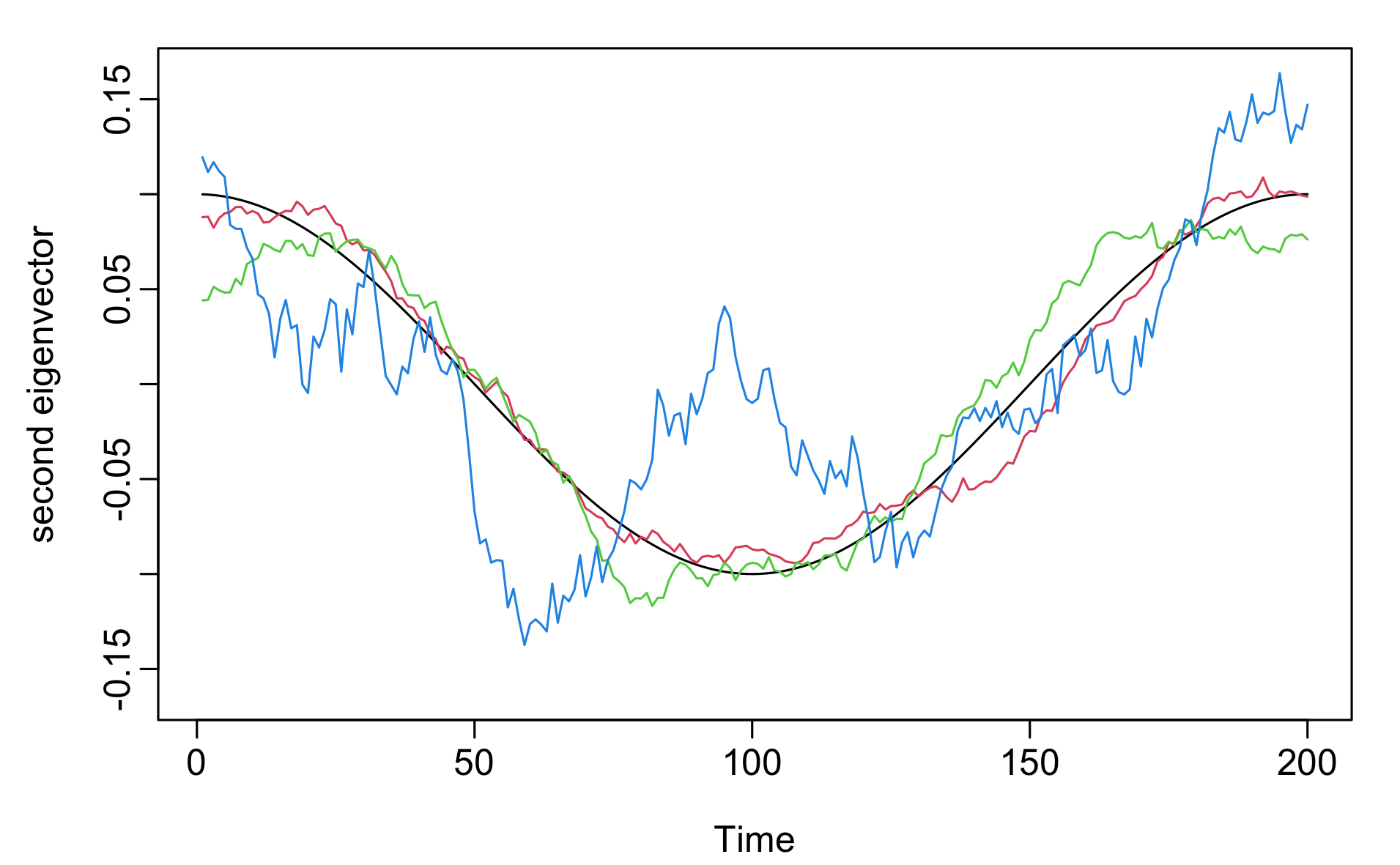}{.32}
    \subf{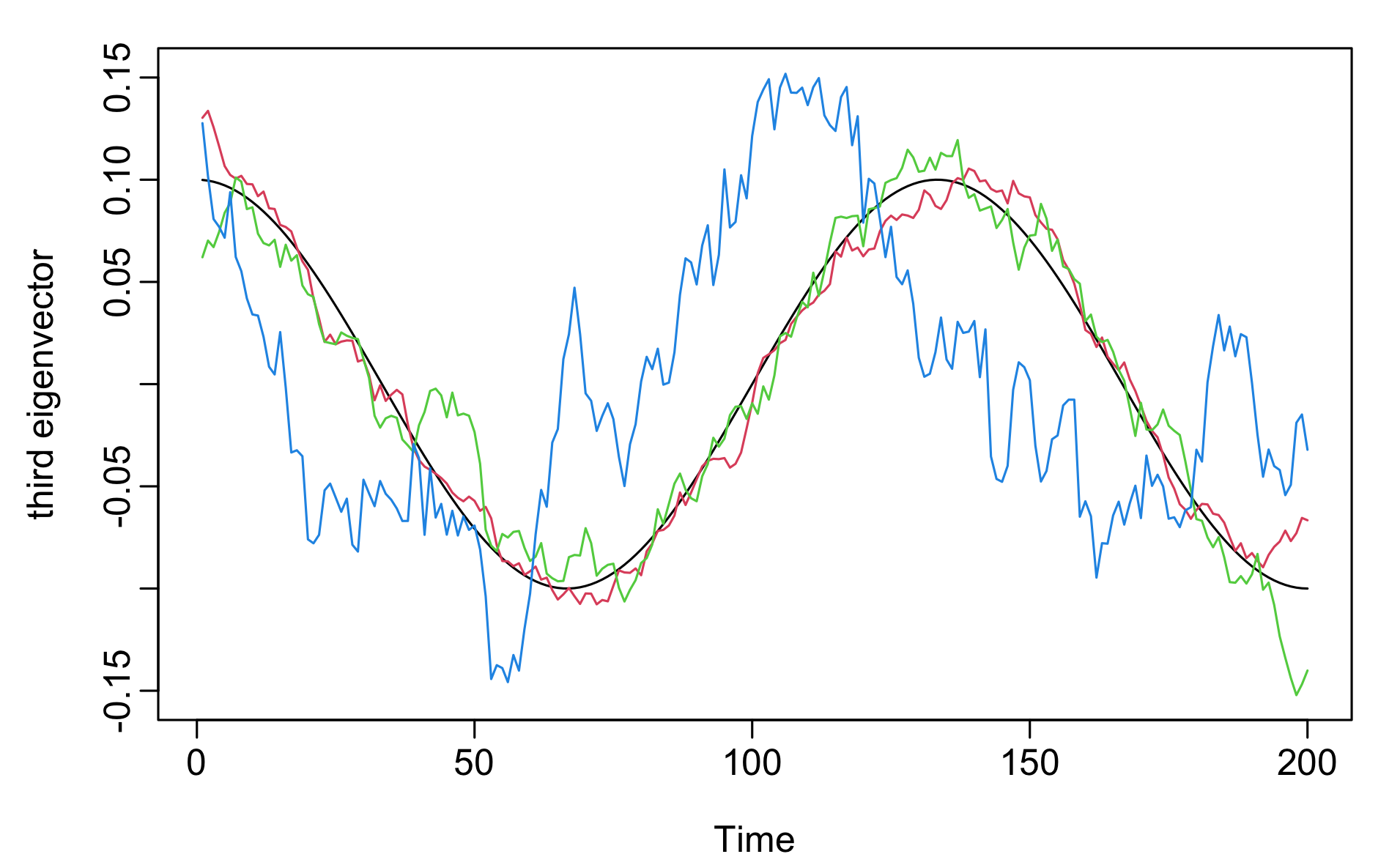}{.32}
    \subf{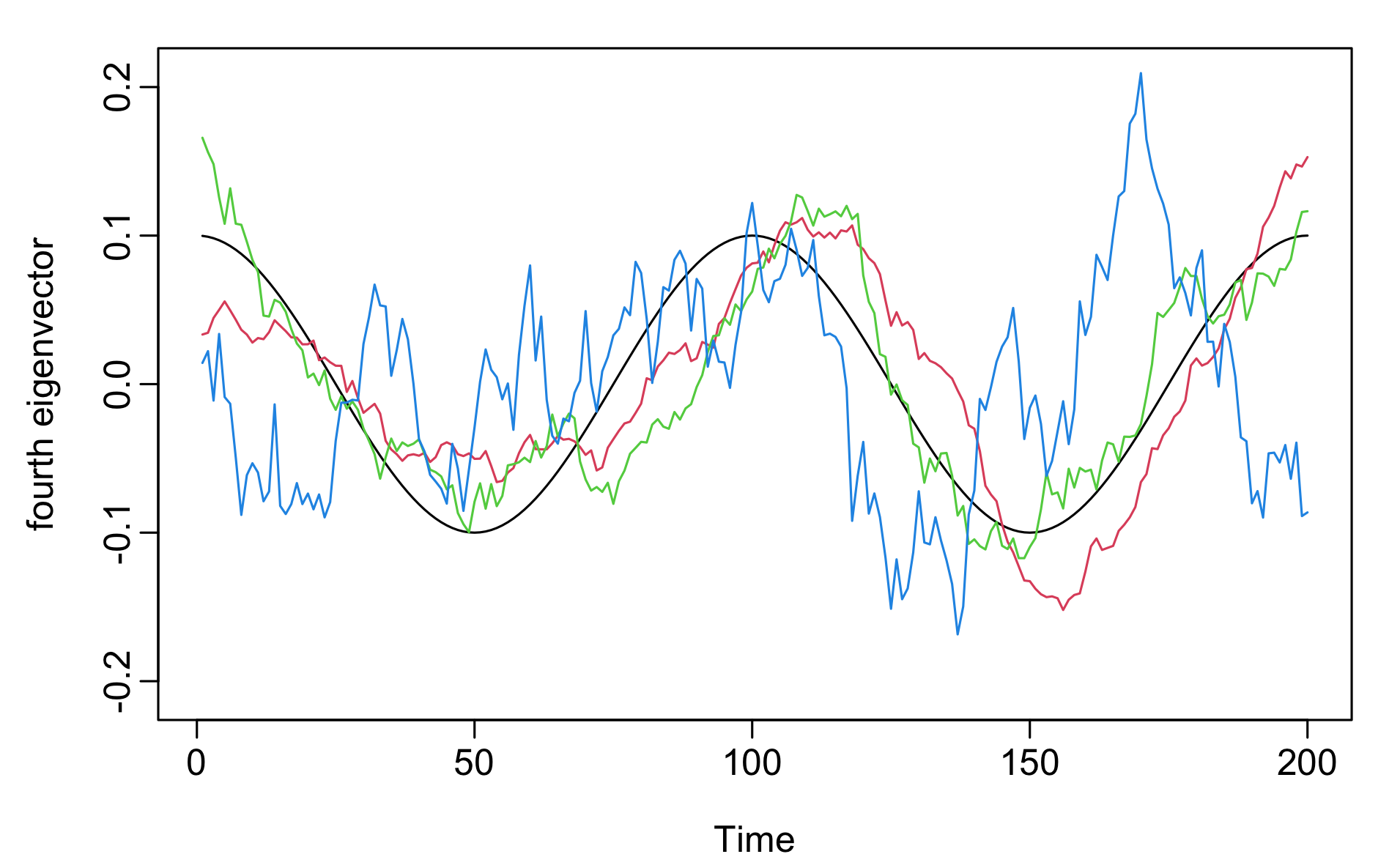}{.32}
    \subf{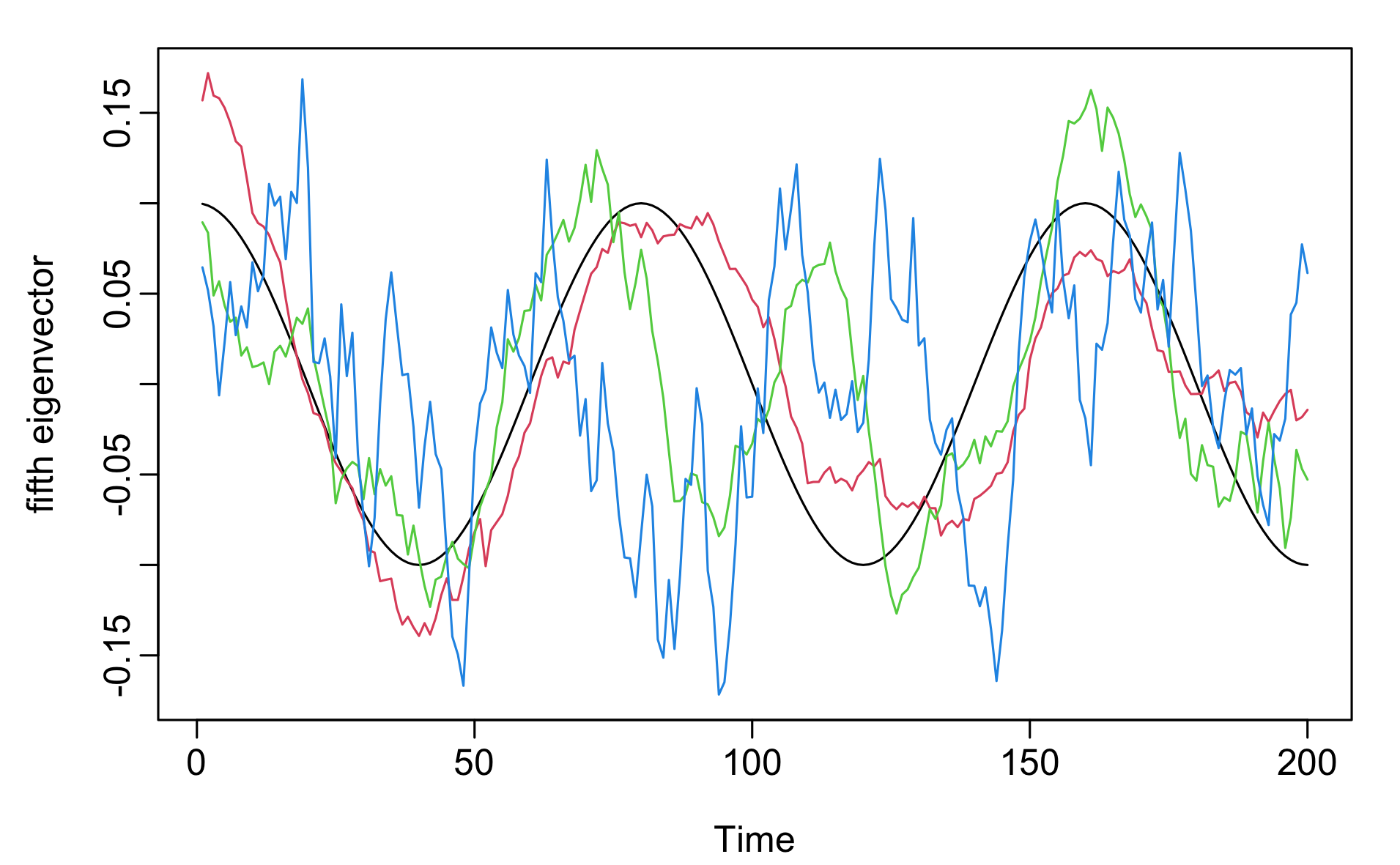}{.32}
    \subf{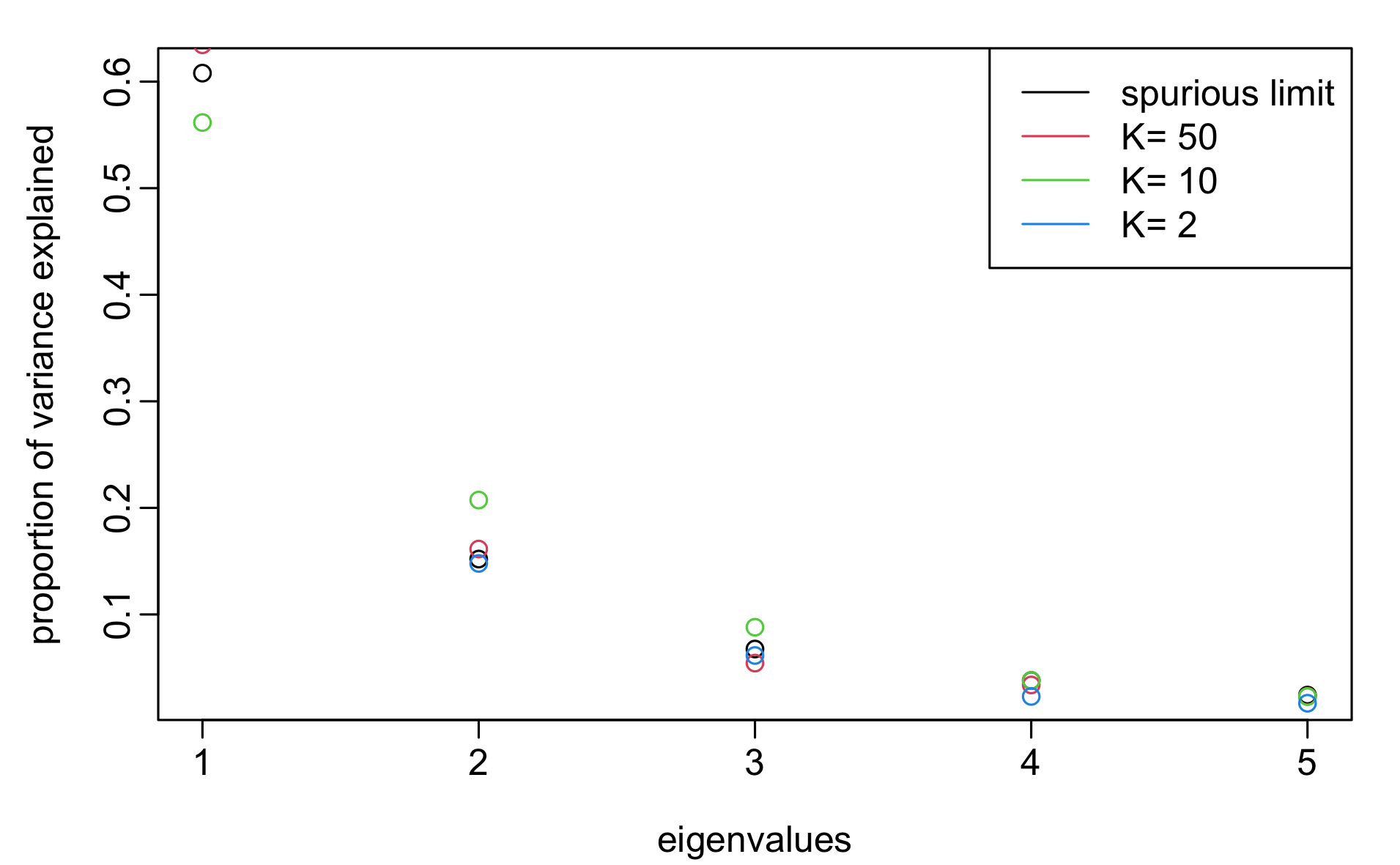}{.32}
    \caption{Setting 3 - $C_\epsilon$ localized,  loading matrices are of full-column rank}\label{figure - local iid}
\end{figure}

Figure \ref{figure - local loweffrank} plots the results for Setting 4 where $C_\epsilon$ has decaying eigenvalues and $A_n$ has low-effective rank. It can be seen that the sample eigenstructure begins to visibly deviate from the spurious limit, 
regardless of the value of $K$. Similar to Setting 3, this can be explained by the upper bounds in \eqref{equation - upper bound on ratio in theorem 2} and Theorem \ref{theorem - upper lower bounds}. Note that in particular, having divergent $K$ does not seem to guarantee the presence of the spurious limit. 
\begin{figure}[!htb]
    \centering
    \subf{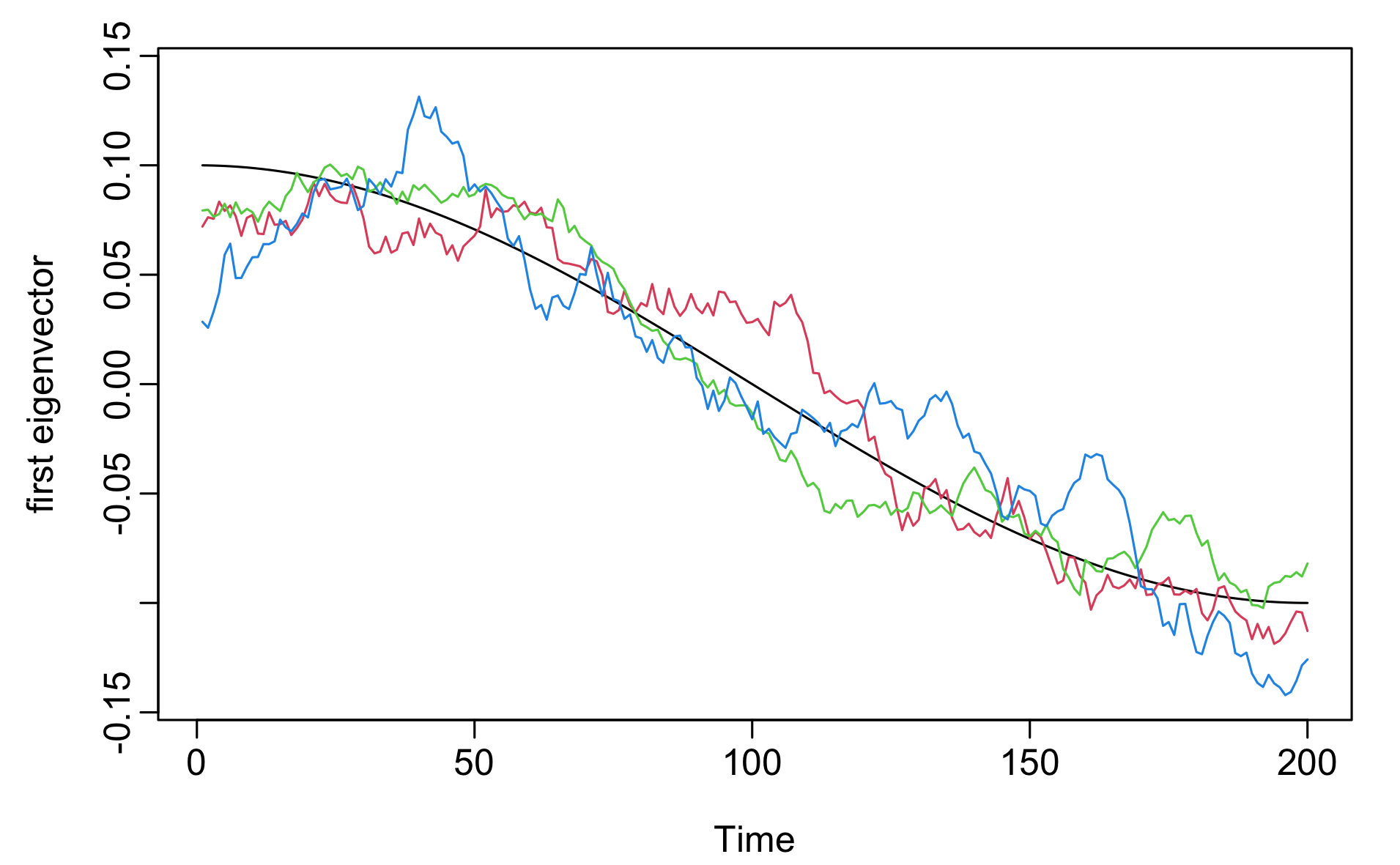}{.32}
    \subf{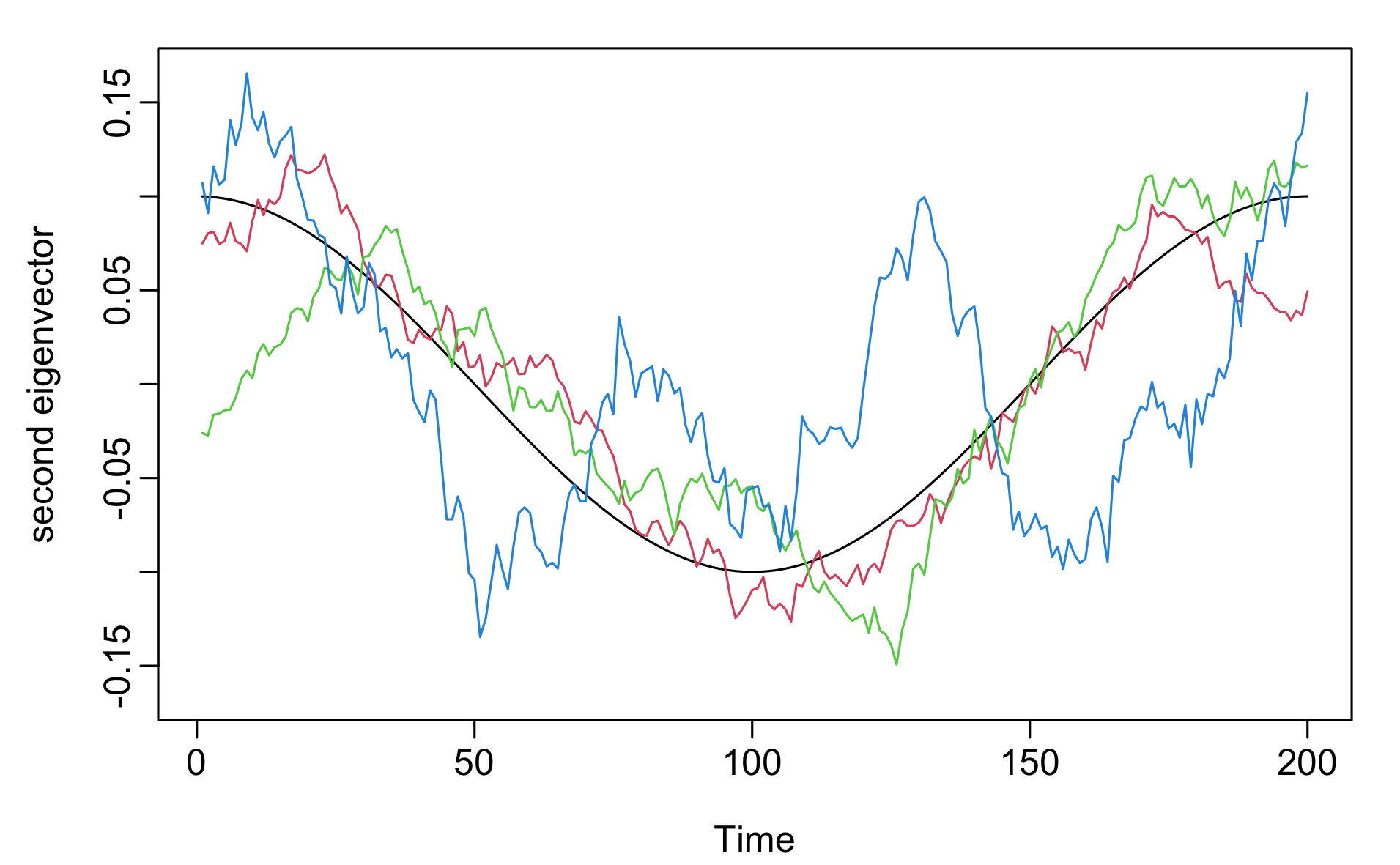}{.32}
    \subf{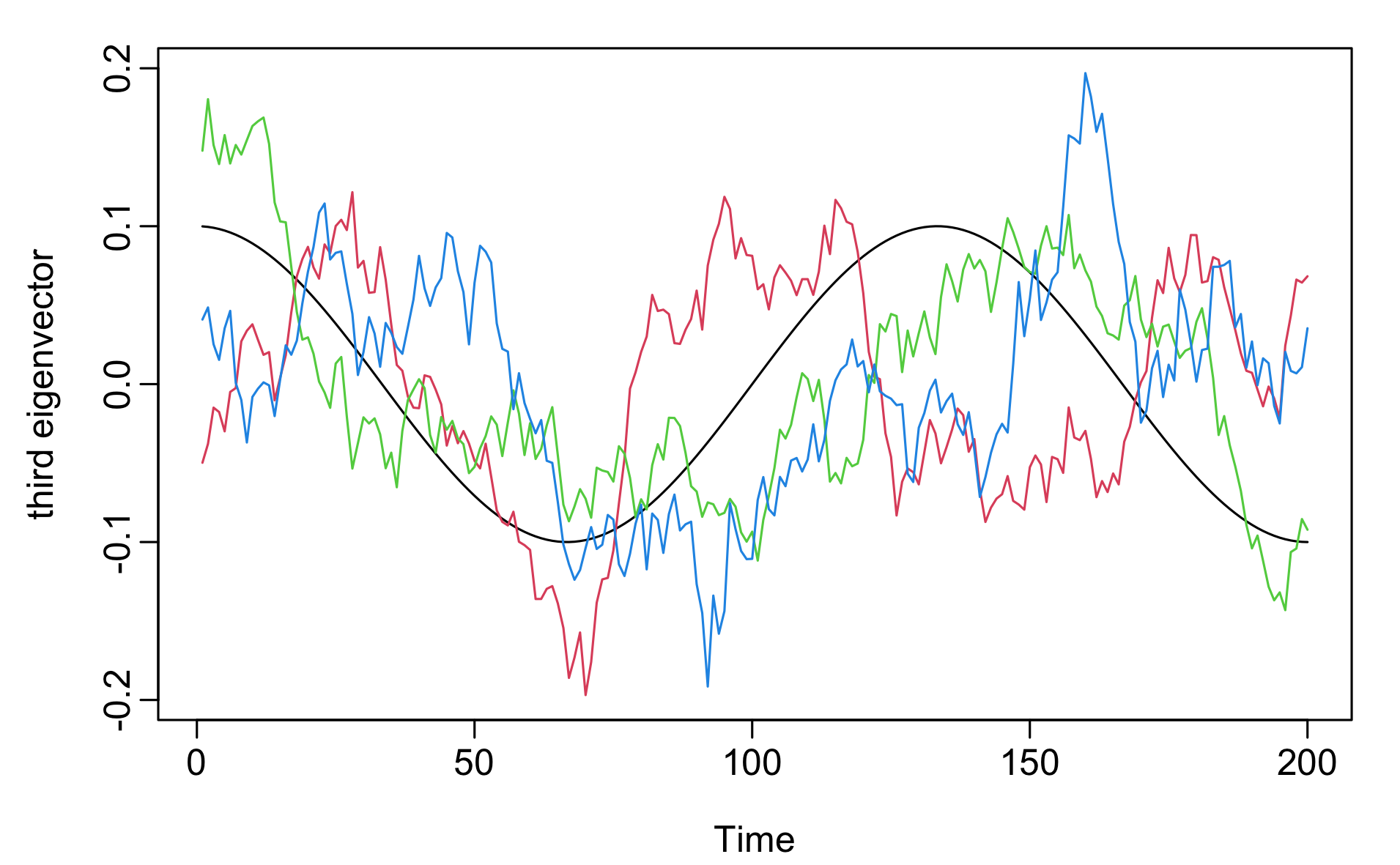}{.32}
    \subf{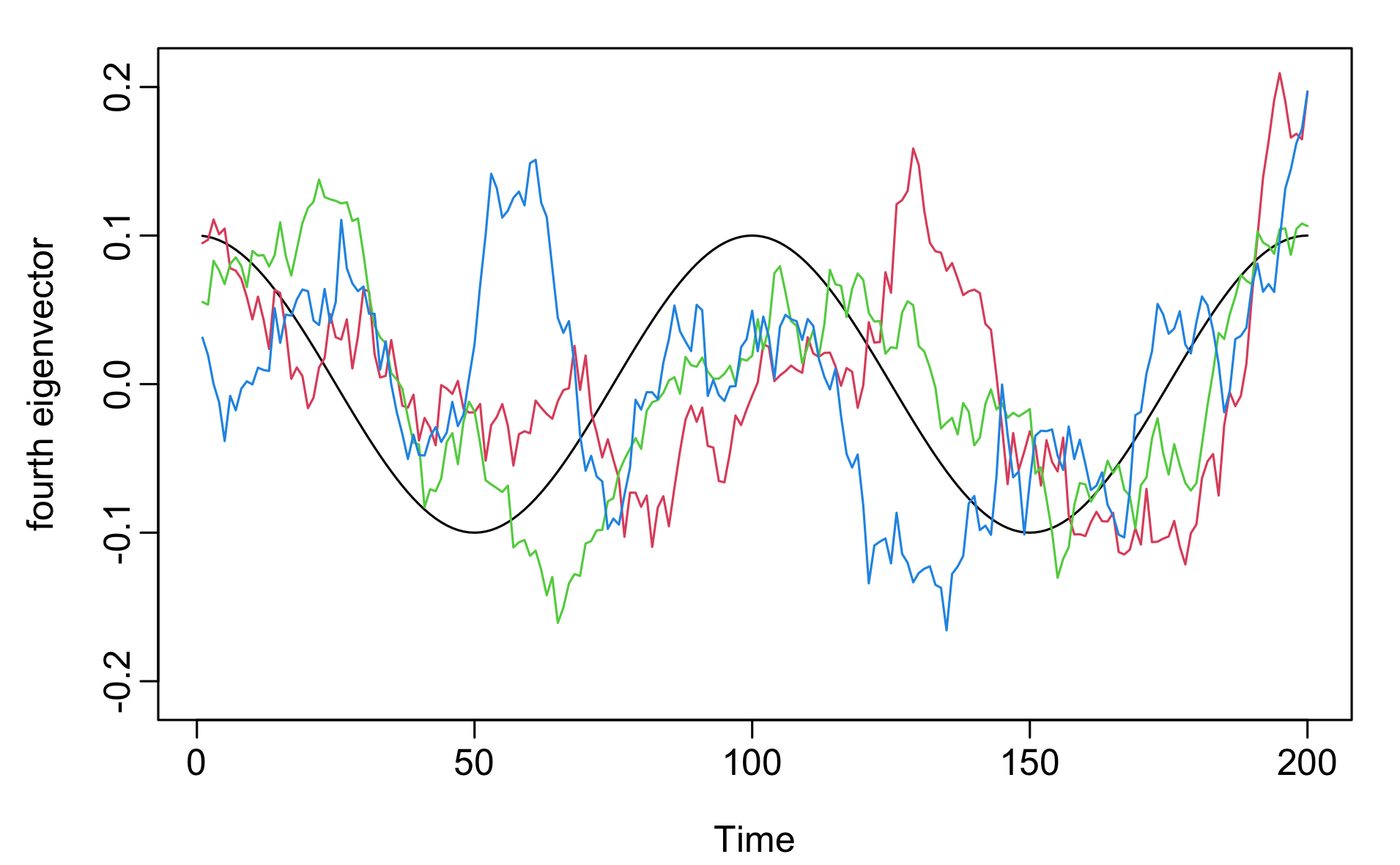}{.32}
    \subf{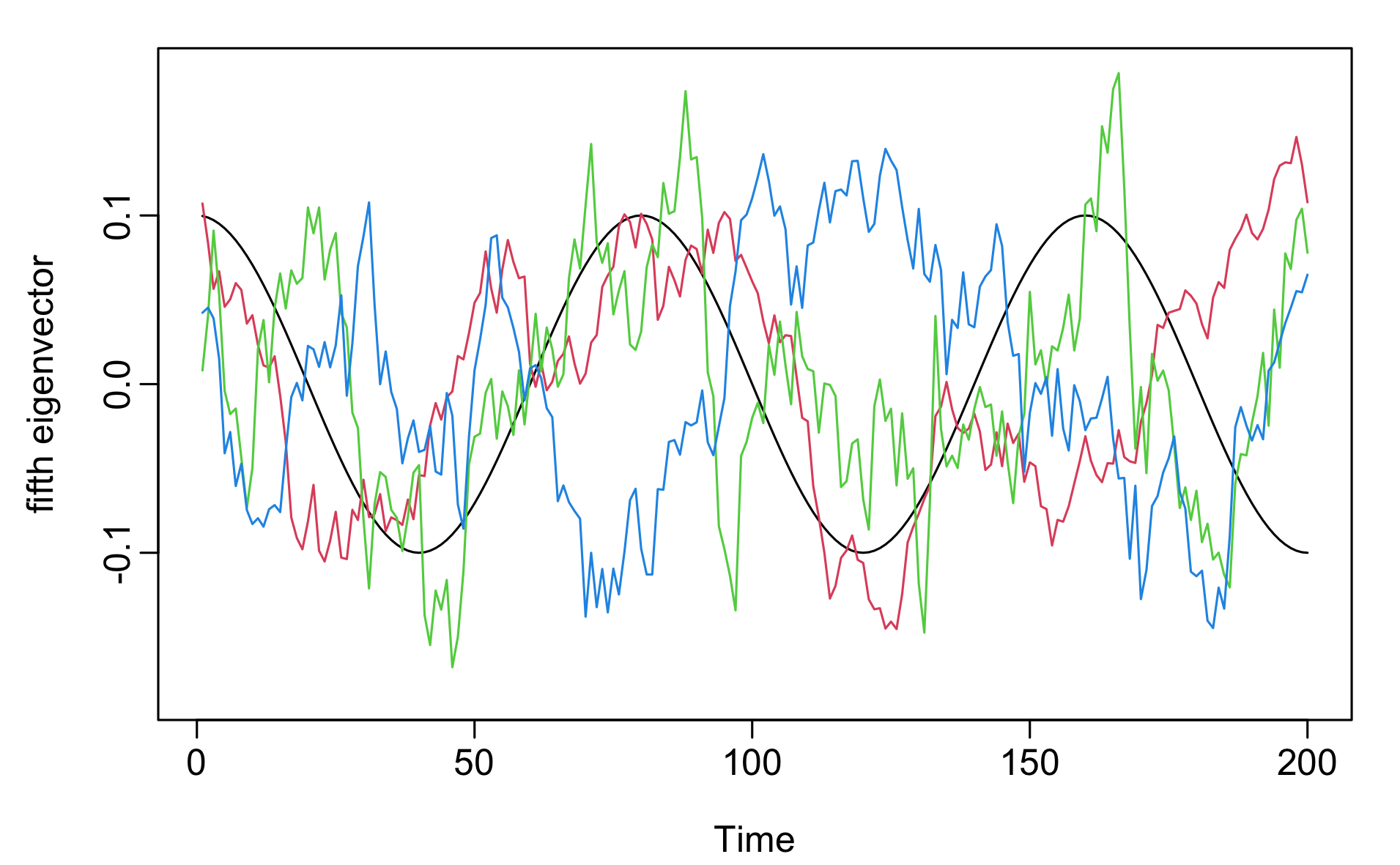}{.32}
    \subf{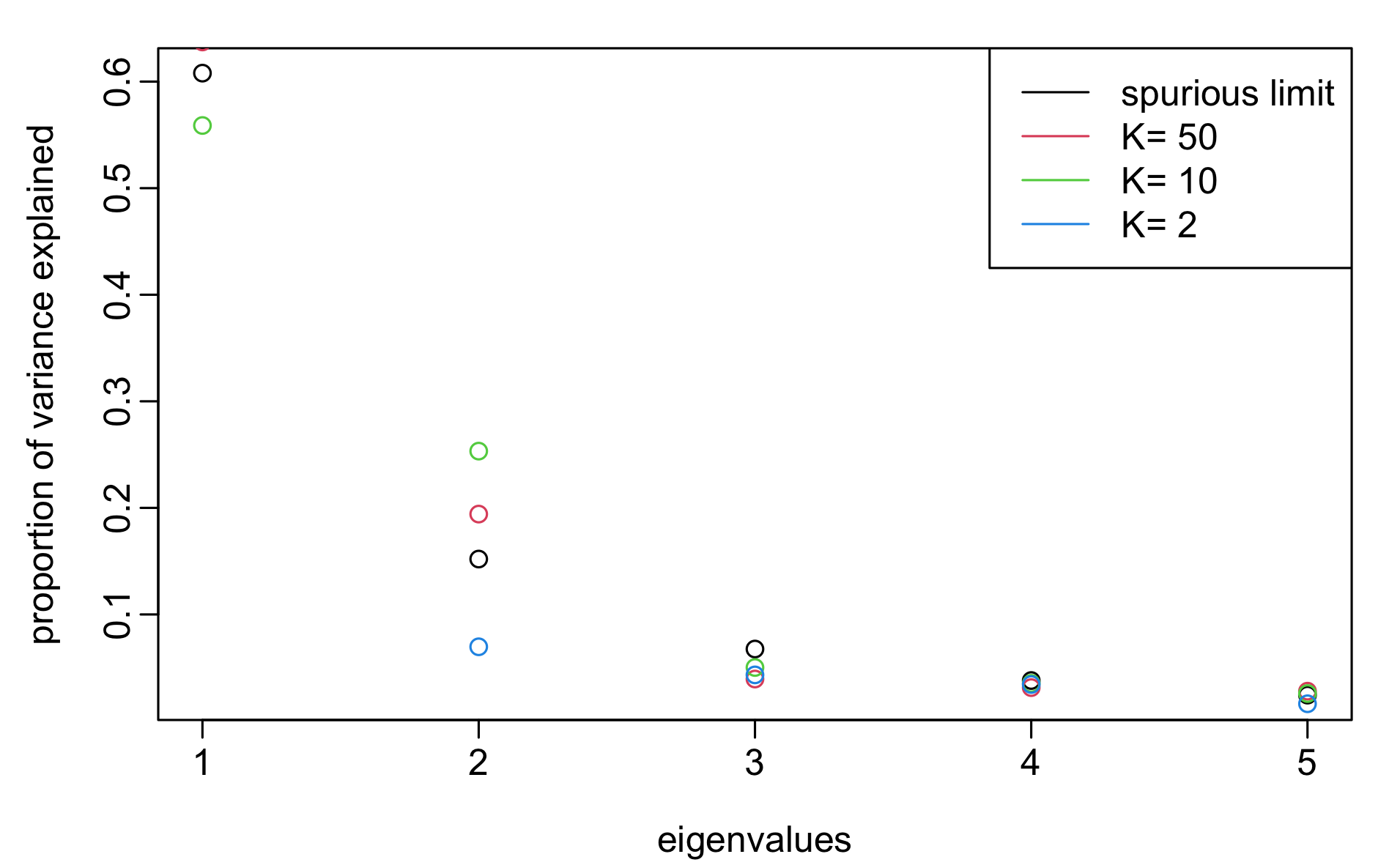}{.32}
    \caption{Setting 4 - $C_\epsilon$ localized, loading matrices are of low-effective rank}\label{figure - local loweffrank}
\end{figure}
 
Figures~\ref{figure - very_local iid} and~\ref{figure - very_local loweffrank} show the results for Settings 5 and 6, where $C_\epsilon$ is completely localized on two basis functions, and $A_n$ has full rank and low-effective rank respectively.

\begin{figure}[!htb]
    \centering
    \subf{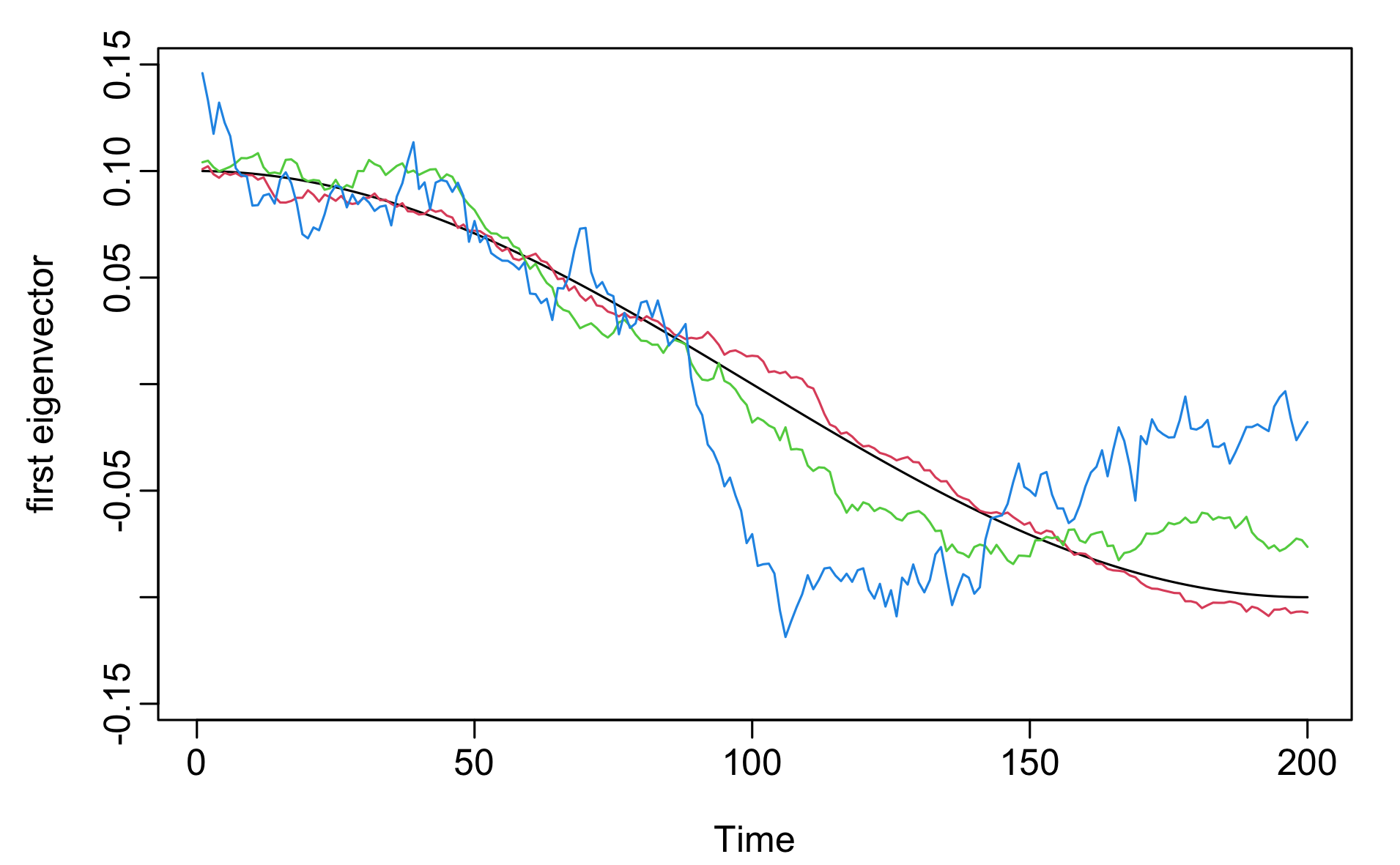}{.32}
    \subf{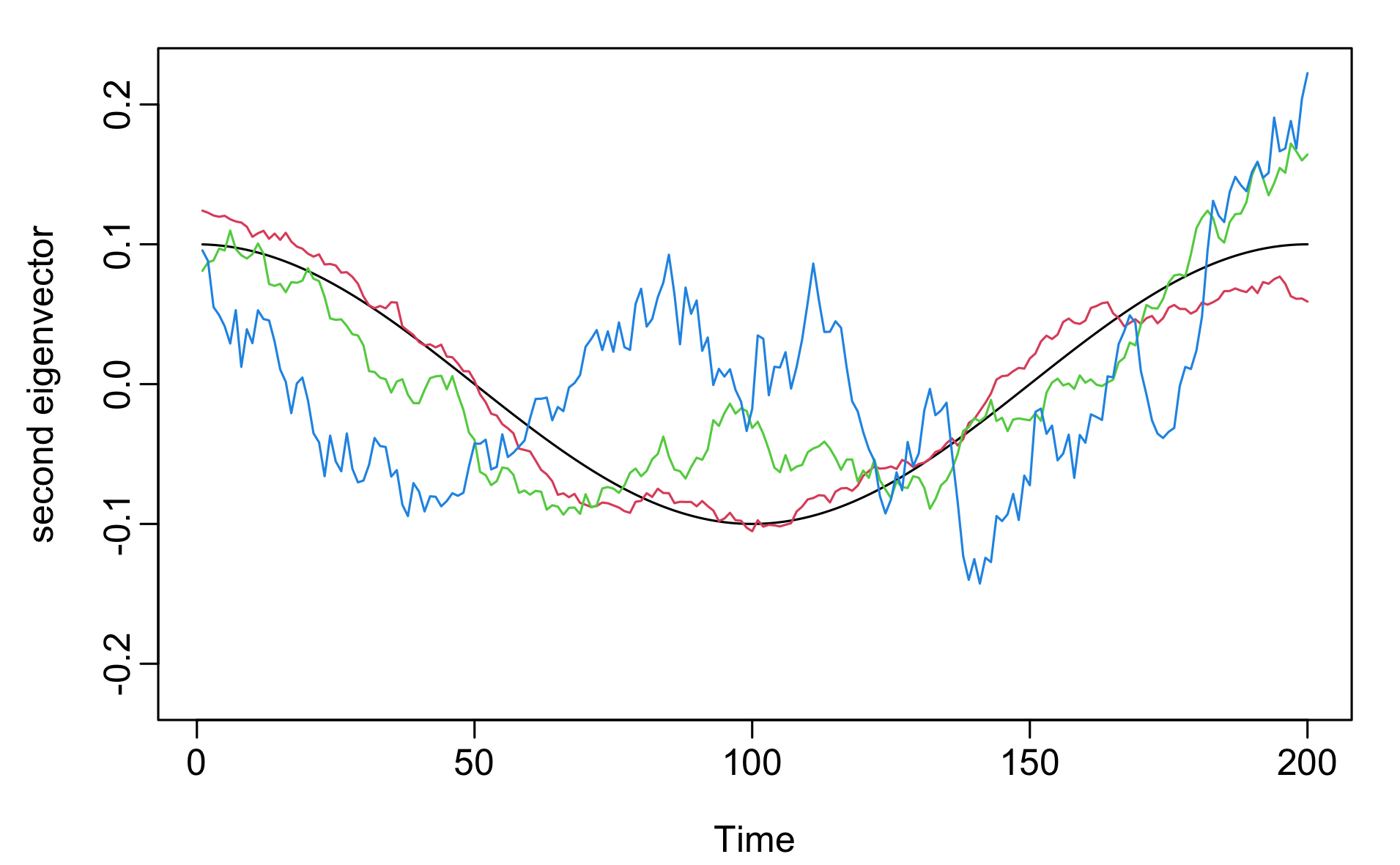}{.32}
    \subf{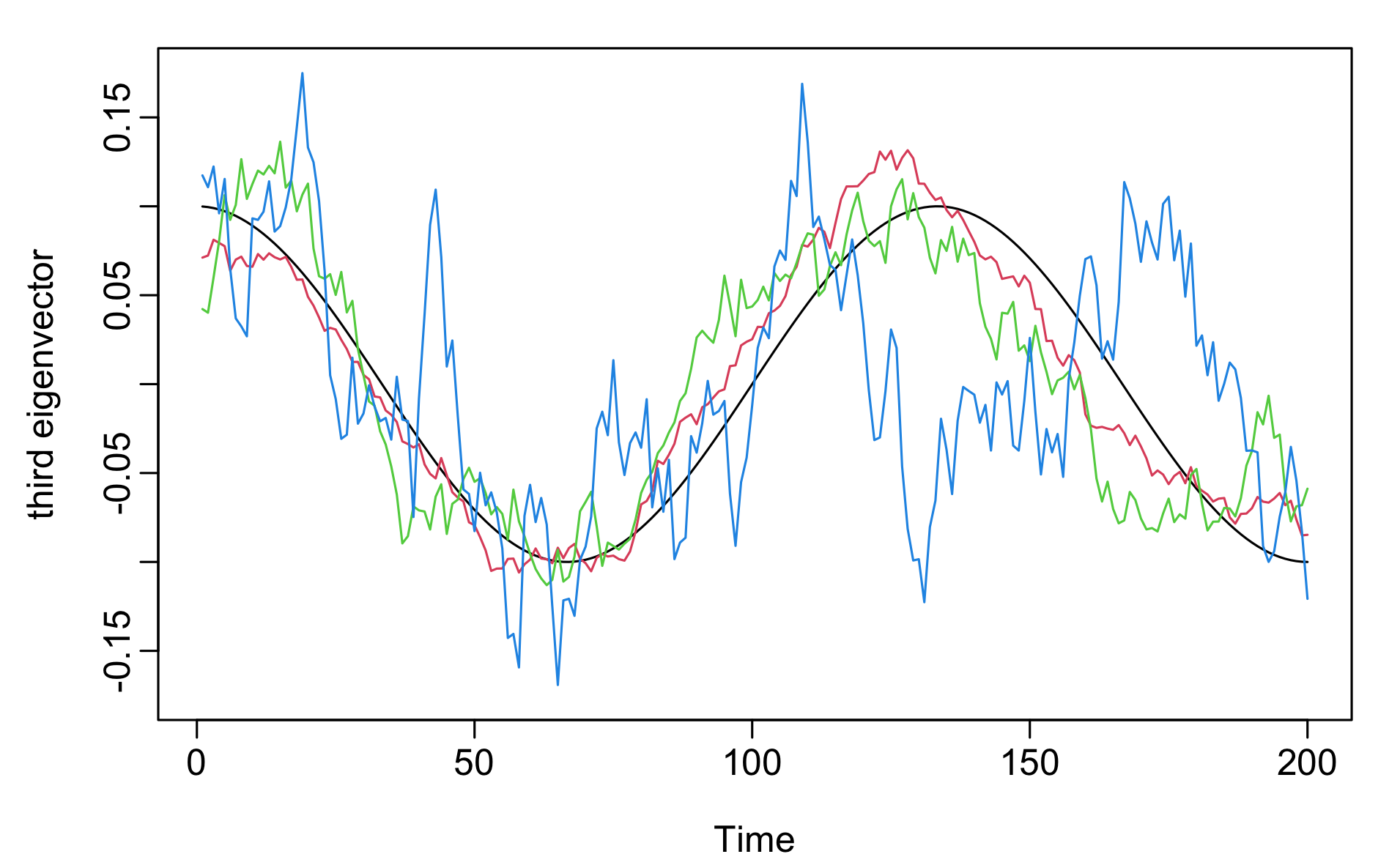}{.32}
    \subf{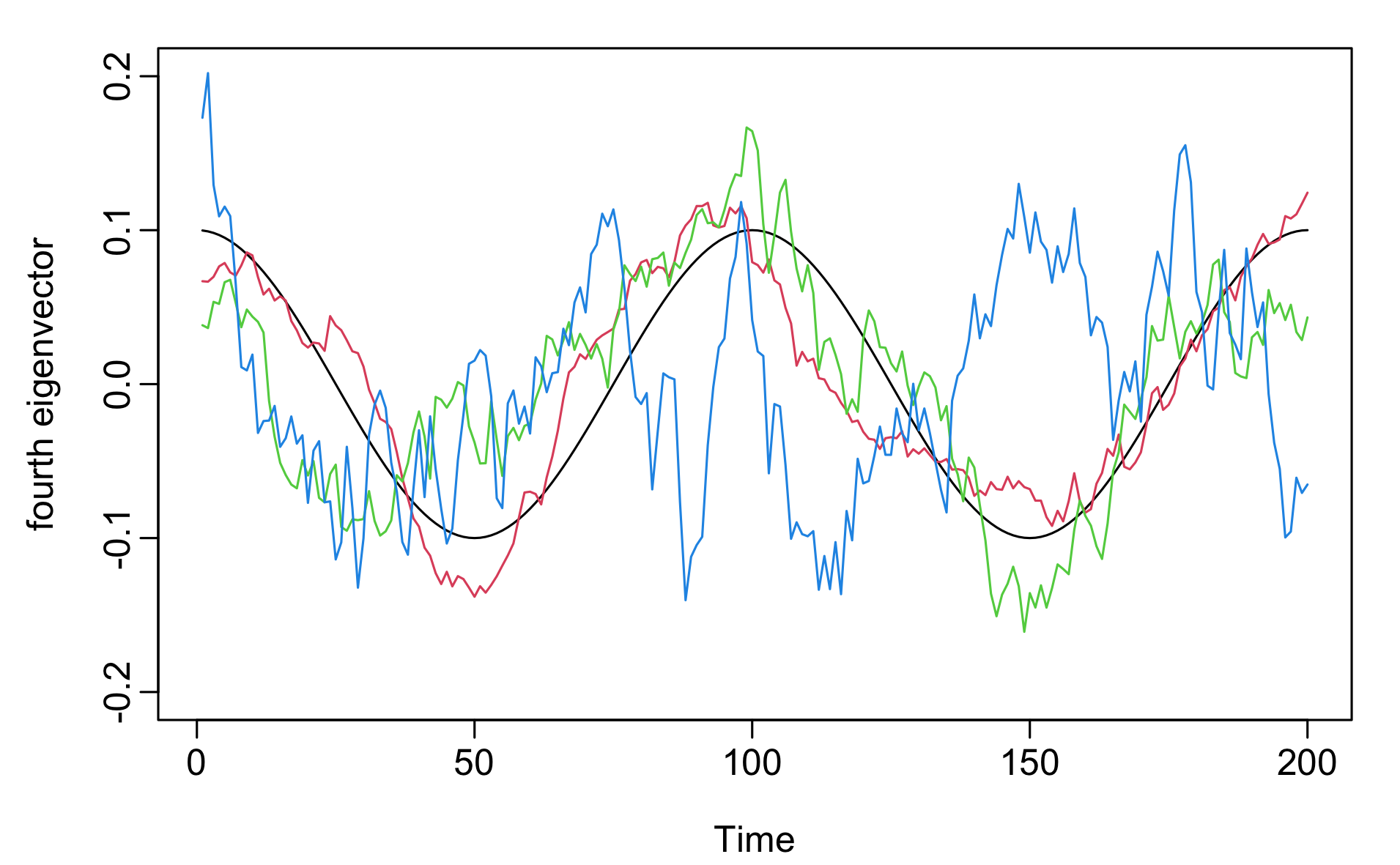}{.32}
    \subf{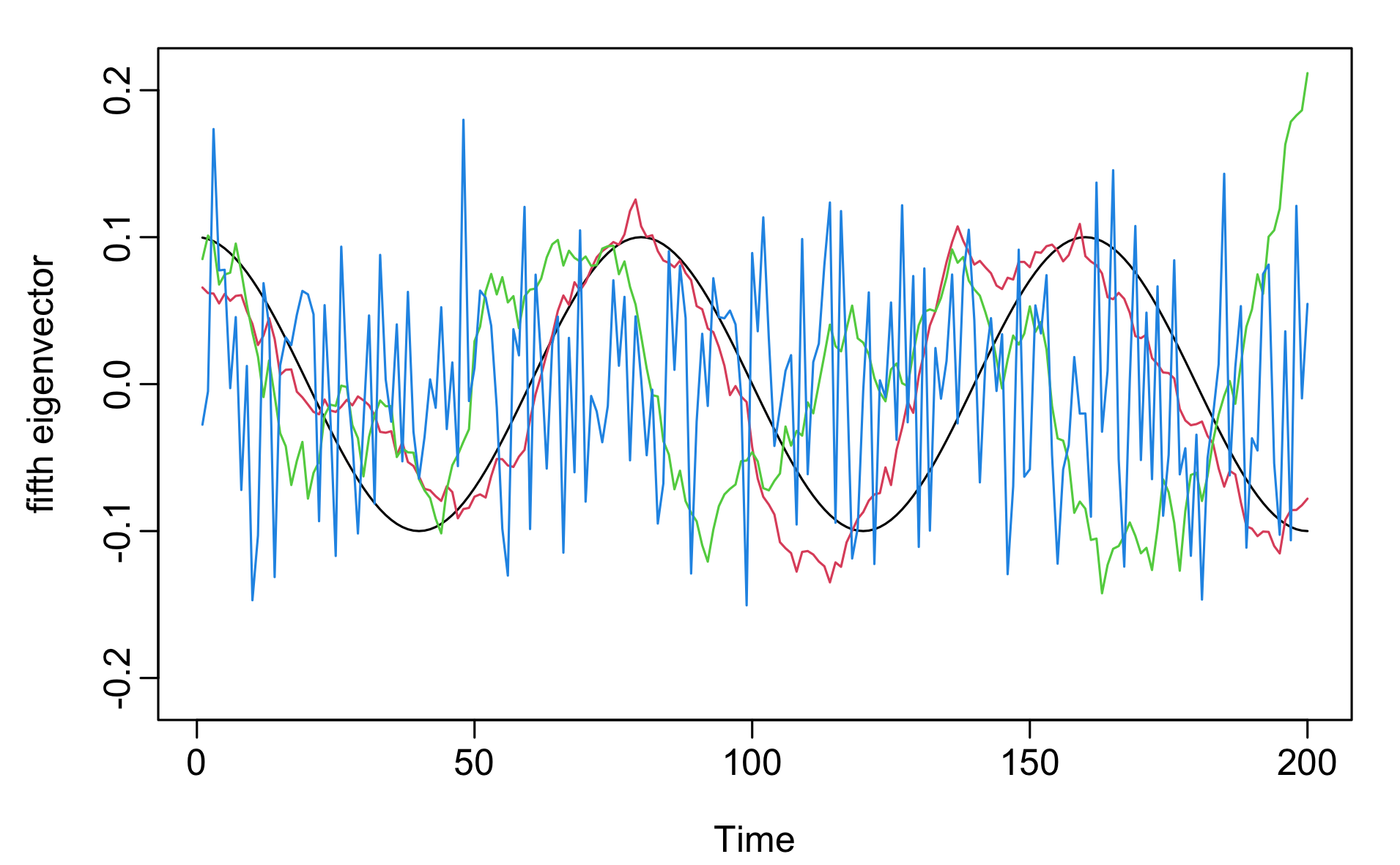}{.32}
    \subf{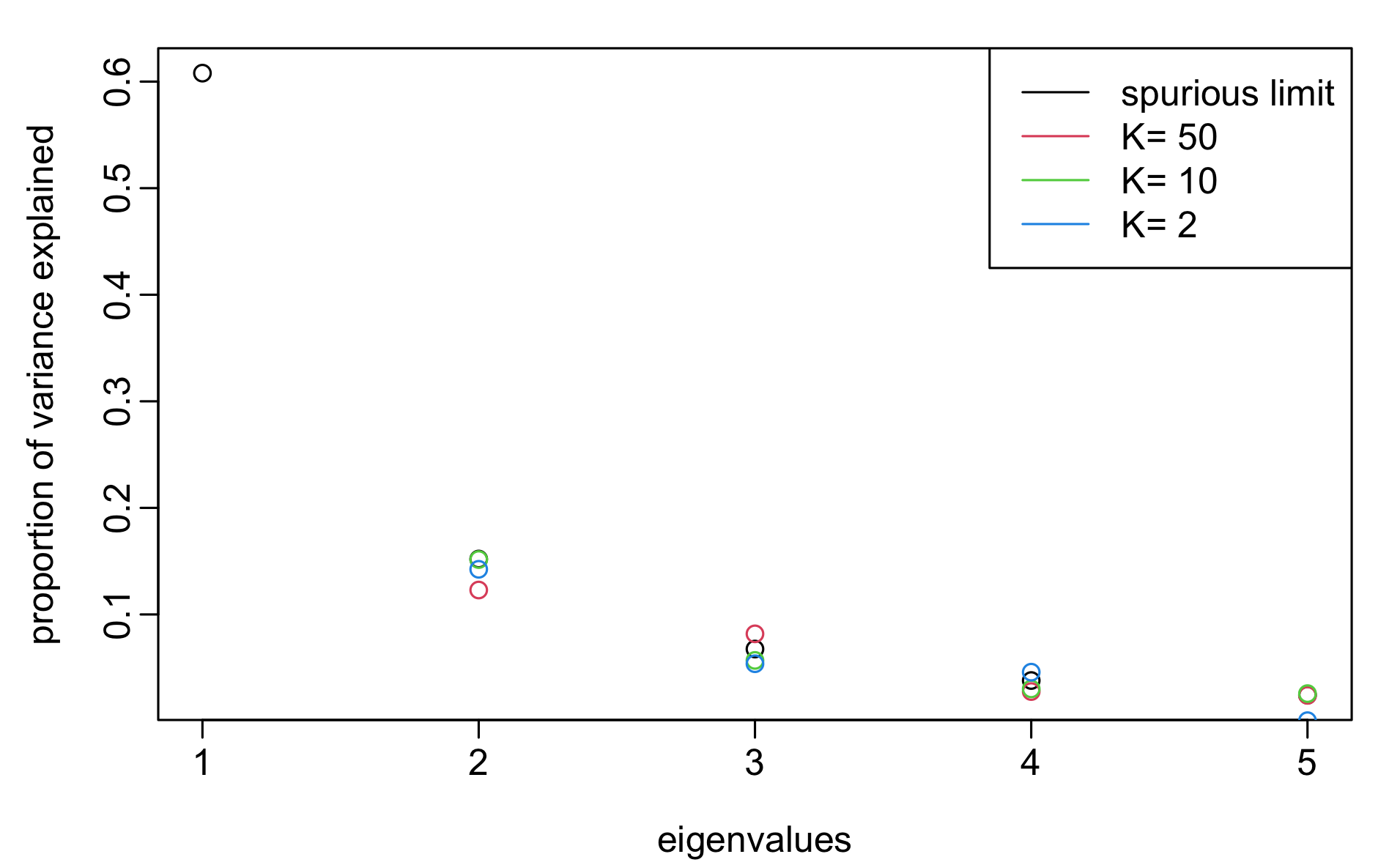}{.32}
    \caption{Setting 5 - $C_\epsilon$ is of rank 2, loading matrices are of full-column rank}\label{figure - very_local iid}
\end{figure}

\begin{figure}[!htb]
    \centering
    \subf{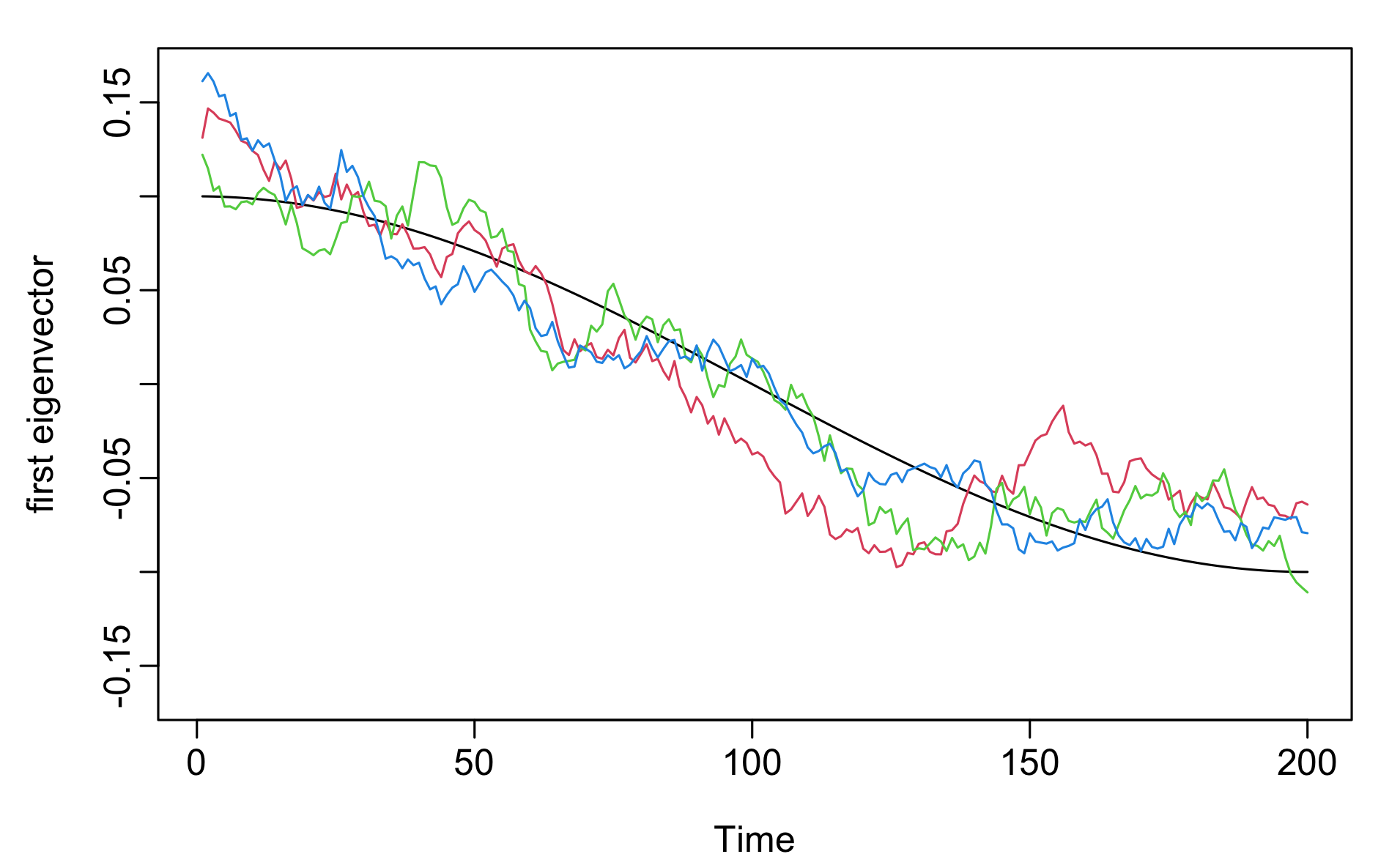}{.32}
    \subf{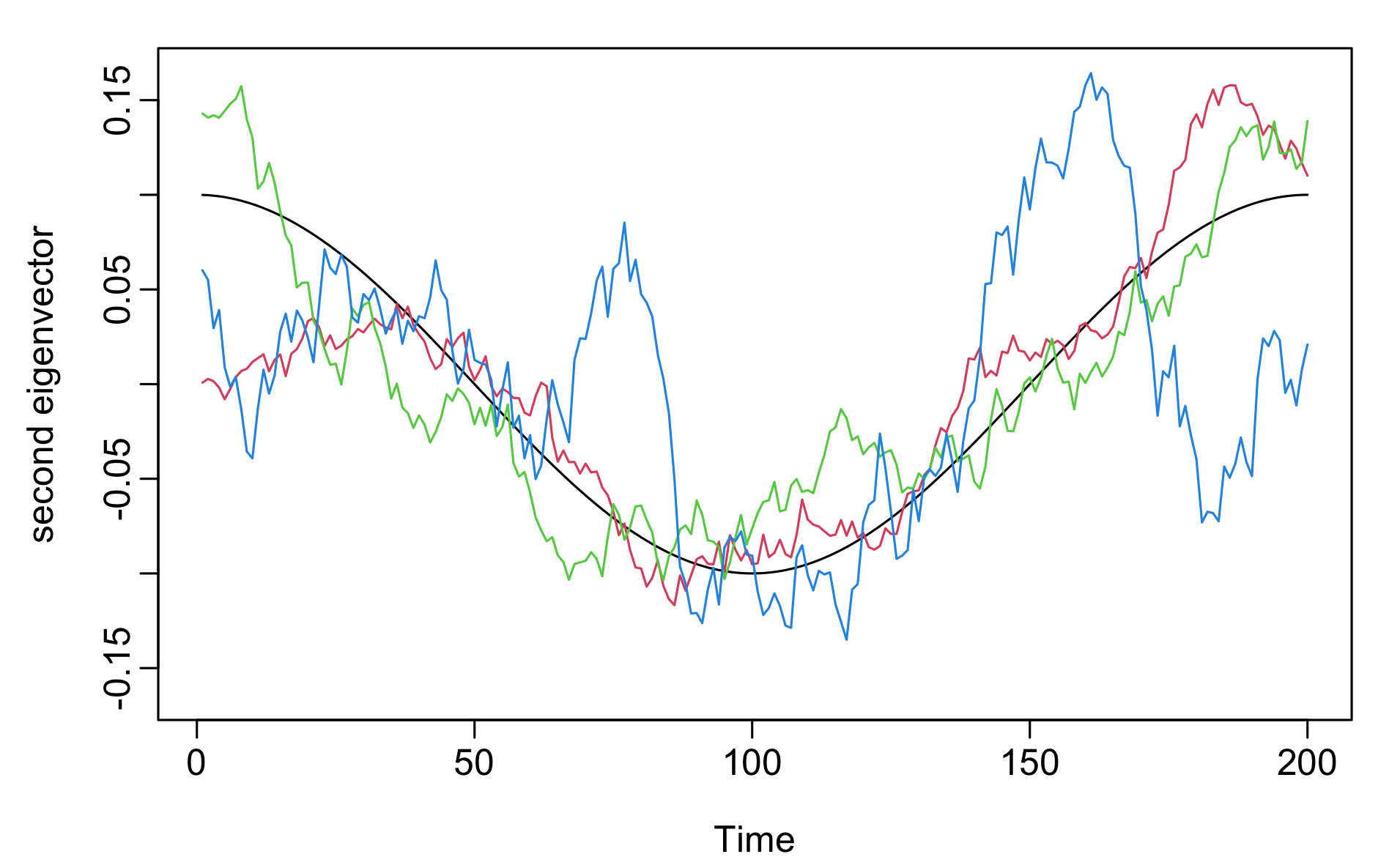}{.32}
    \subf{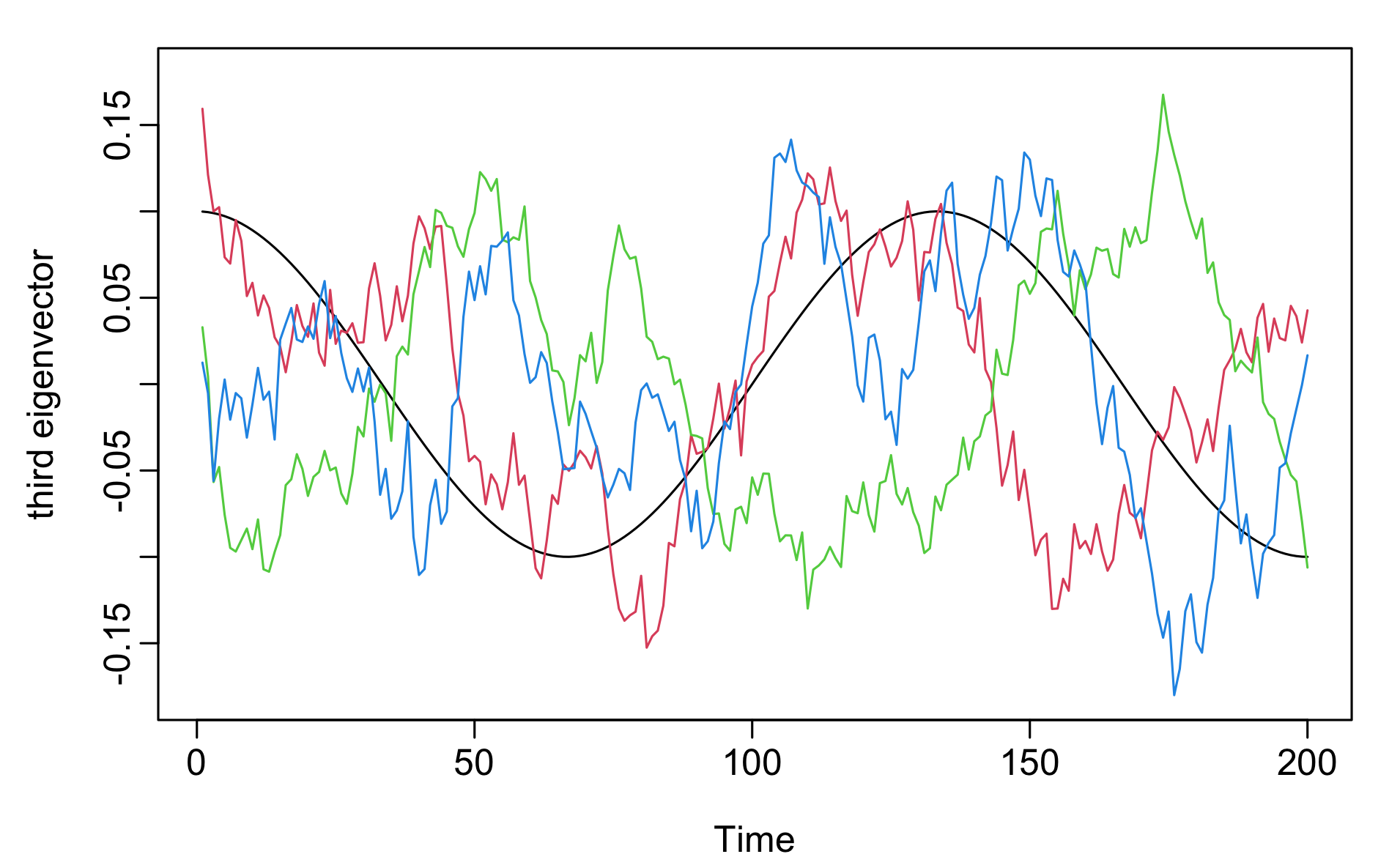}{.32}
    \subf{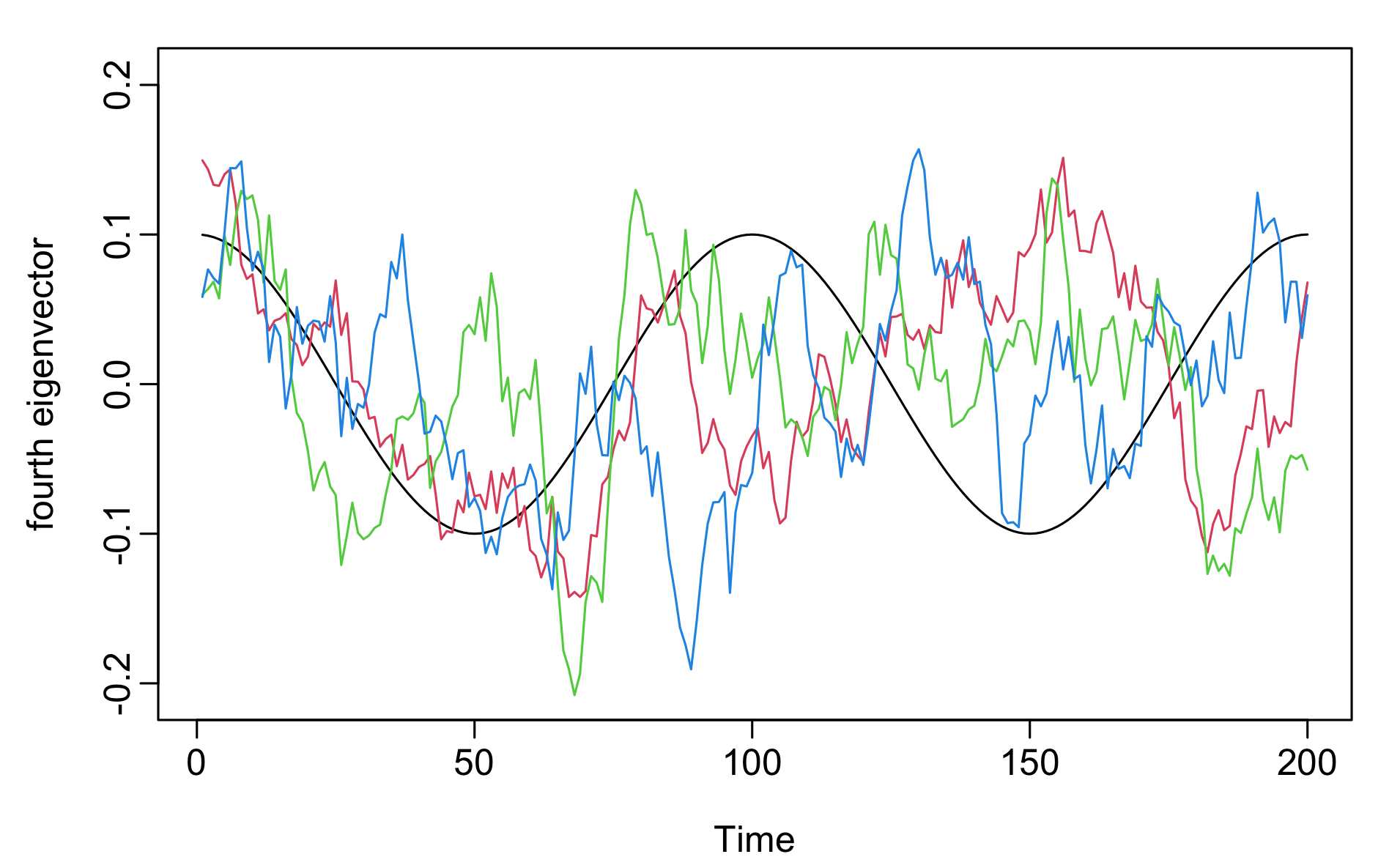}{.32}
    \subf{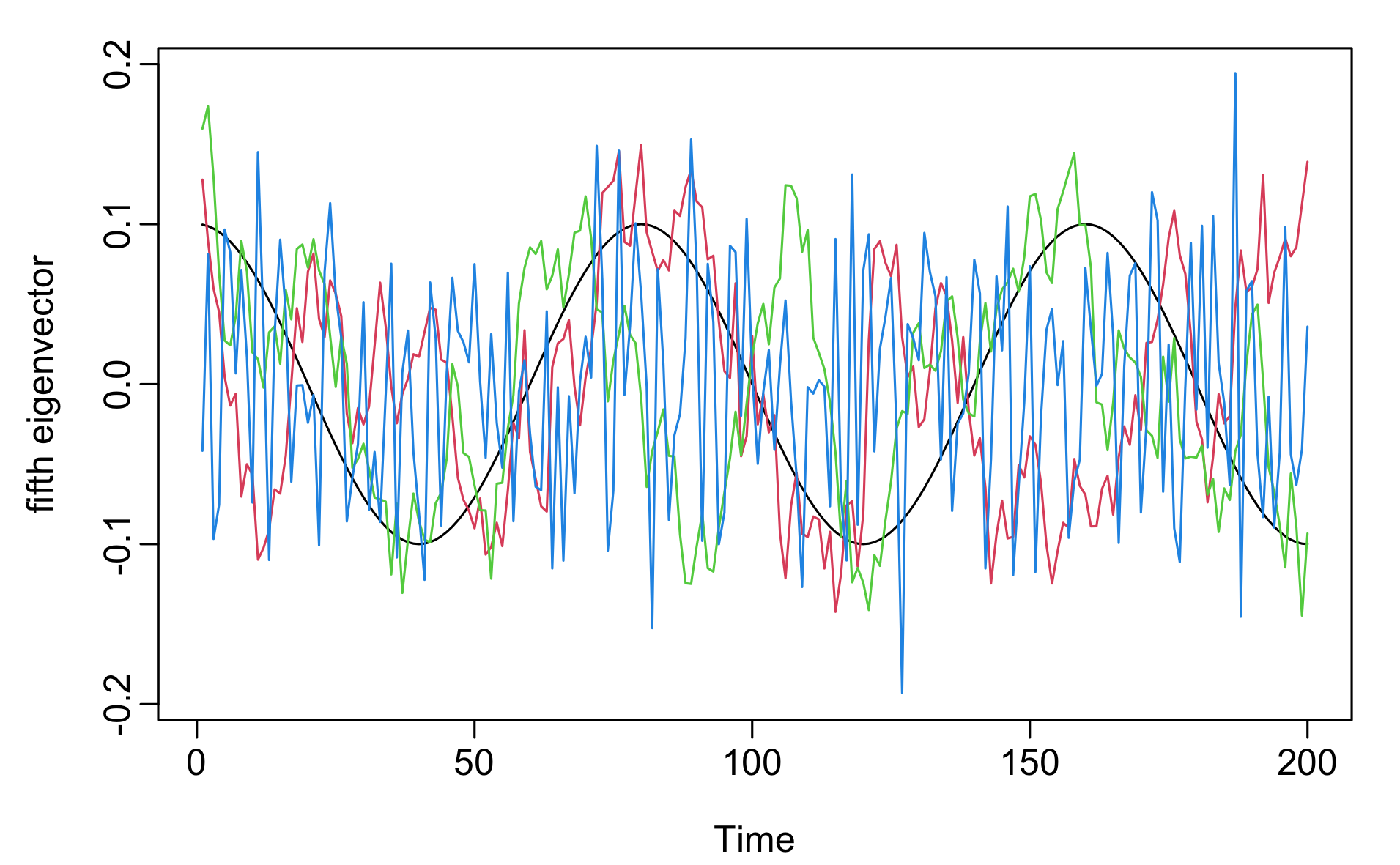}{.32}
    \subf{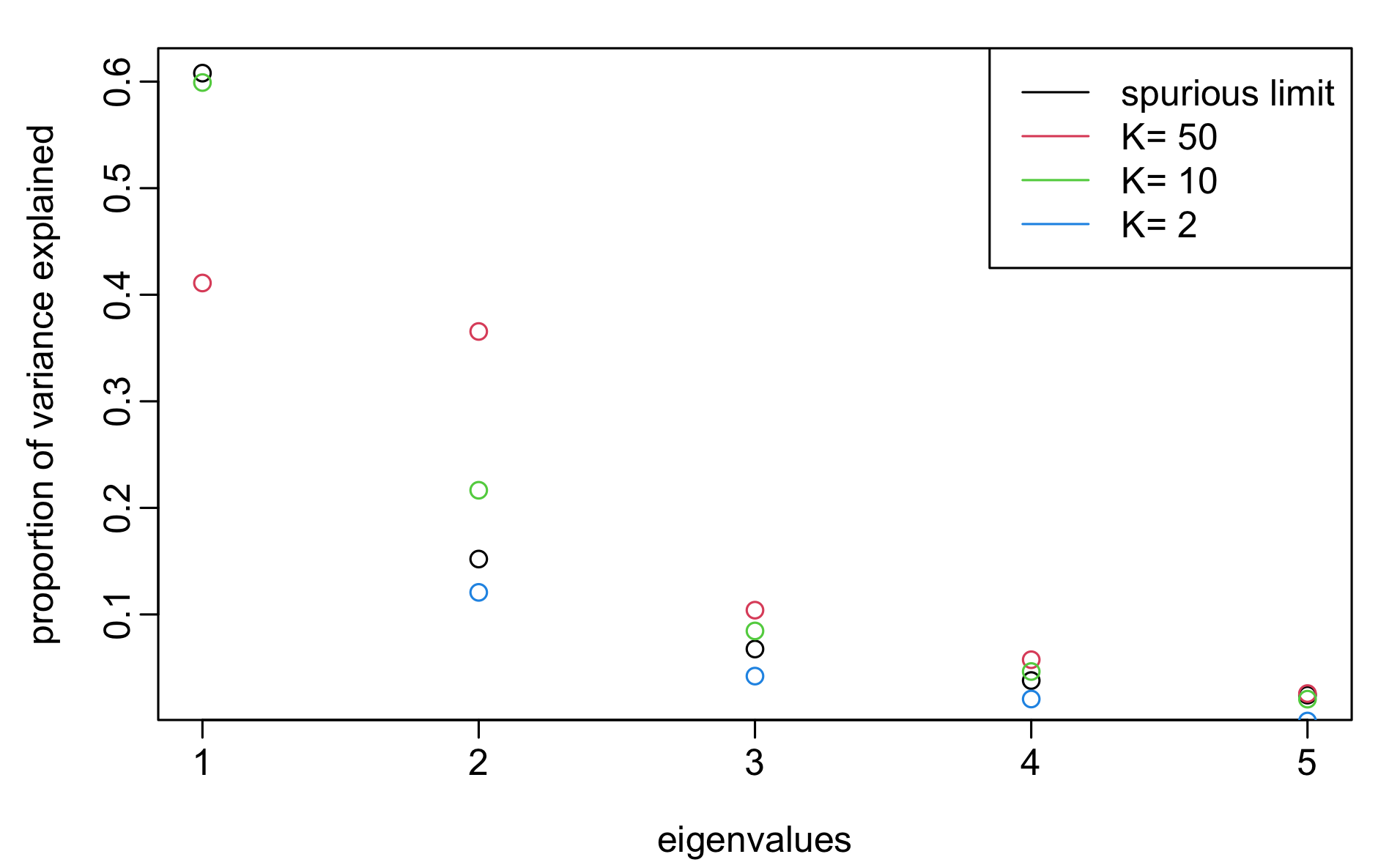}{.32}
    \caption{Setting 6 - $C_\epsilon$ is of rank 2, loading matrices are of low-effective rank}\label{figure - very_local loweffrank}
\end{figure}

In Figure~\ref{figure - very_local iid}, the case $K=2$ significantly deviates from the spurious limit. This effect is similar to Setting 3 but is visibly more pronounced, since the effective rank of $C_\epsilon$, which is an important quantity due to Theorem \ref{theorem - upper lower bounds},  is even smaller in Setting 5. 
In Setting 6, the sample eigen-structures all deviate from the spurious limit regardless of the value $K$, similar to Setting 4. The deviations are more visibly pronounced compared to Setting 4 due to the decrease in the rank of $A_n$.

\section{Empirical applications}\label{section - empirical}
We consider two age-specific mortality rate datasets as an empirical illustration of our results. Our first data set contains age-specific and gender-specific mortality rates obtained from the Human Mortality Database. The dataset consists of mortality rates for $p=32$ countries observed from 1960 to 2013 ($T=54$). 
Since exposures and death counts are sparse at high ages, 
we aggregate data with age over 100 into a single observation. The raw data is transformed into functional data via smoothed using the \verb|demography| package in \verb|R| and the logarithm of the mortality rates are computed. 
This dataset was considered in \cite{TangShangYang2022} where a novel clustering algorithm is proposed based on functional principal component analysis. 

From the smoothed log mortality rates, we computed the sample covariance matrix as defined in~\eqref{equation - gram matrix} and computed its eigenvalues and eigenvectors. Figure~\ref{figure - multi} plots the first four eigenvectors against the spurious limits as found in Theorem \ref{theorem - main}. 
\begin{figure}[!htb]
    \centering
    \subf{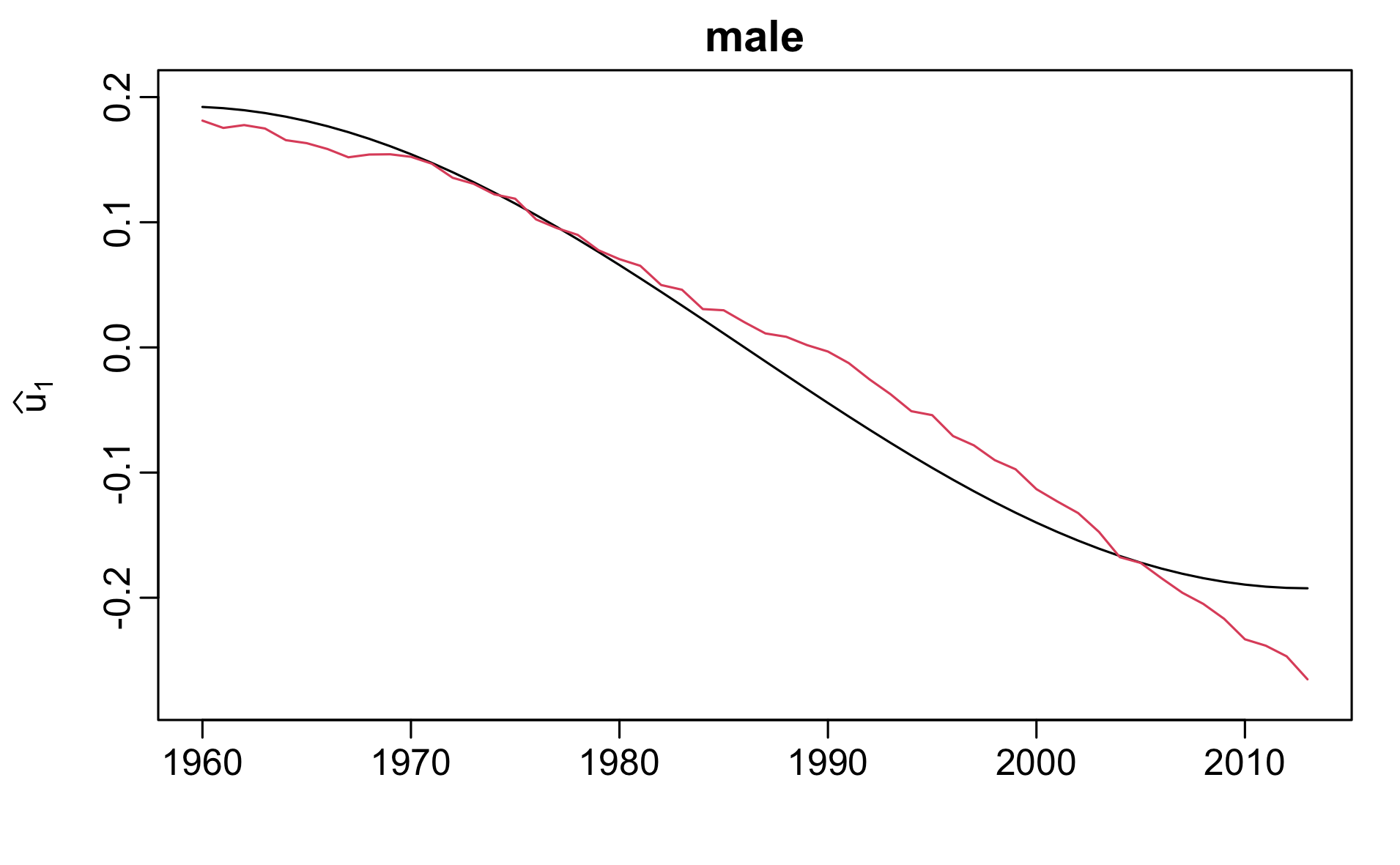}{.30}
    \subf{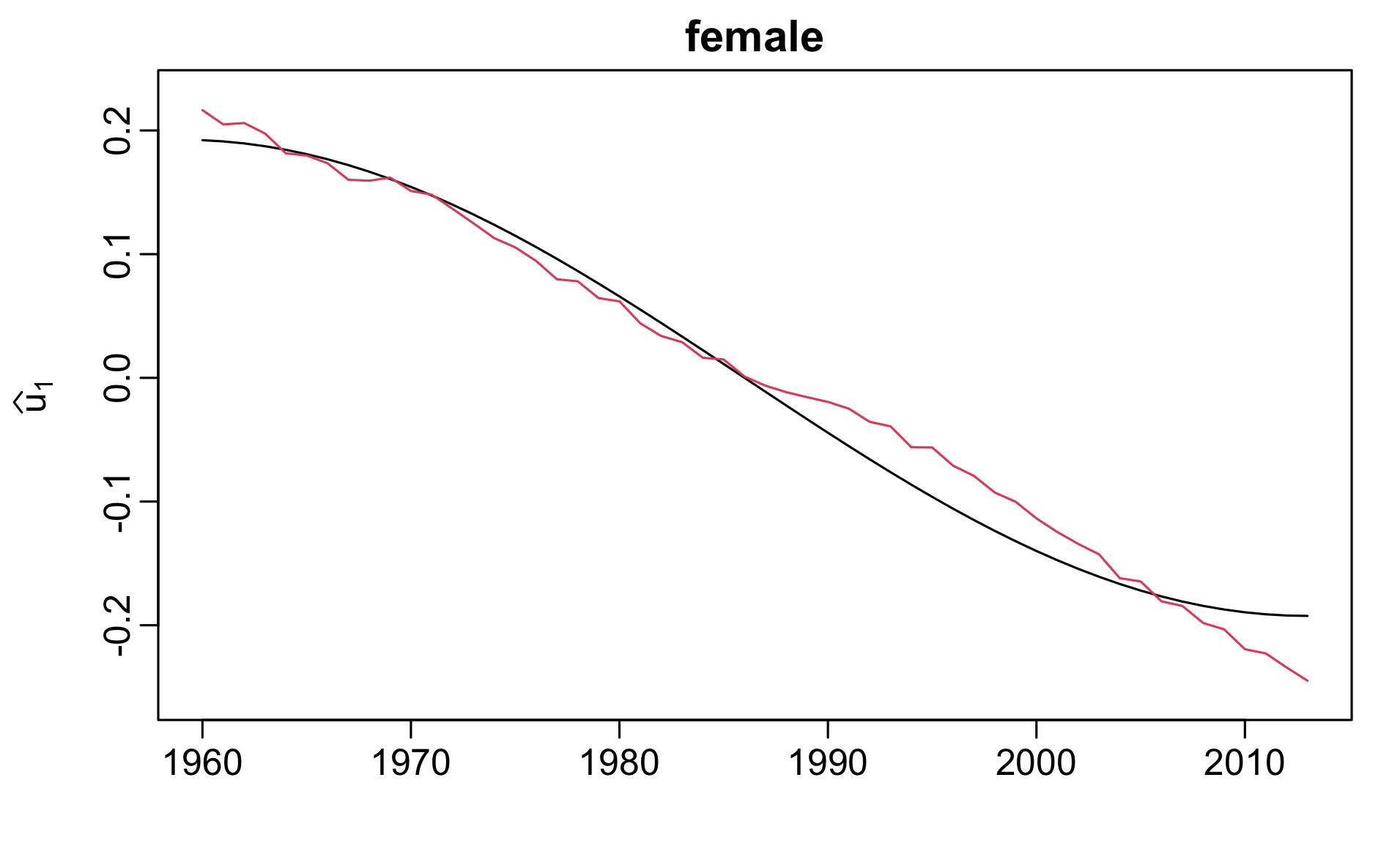}{.30}
    \subf{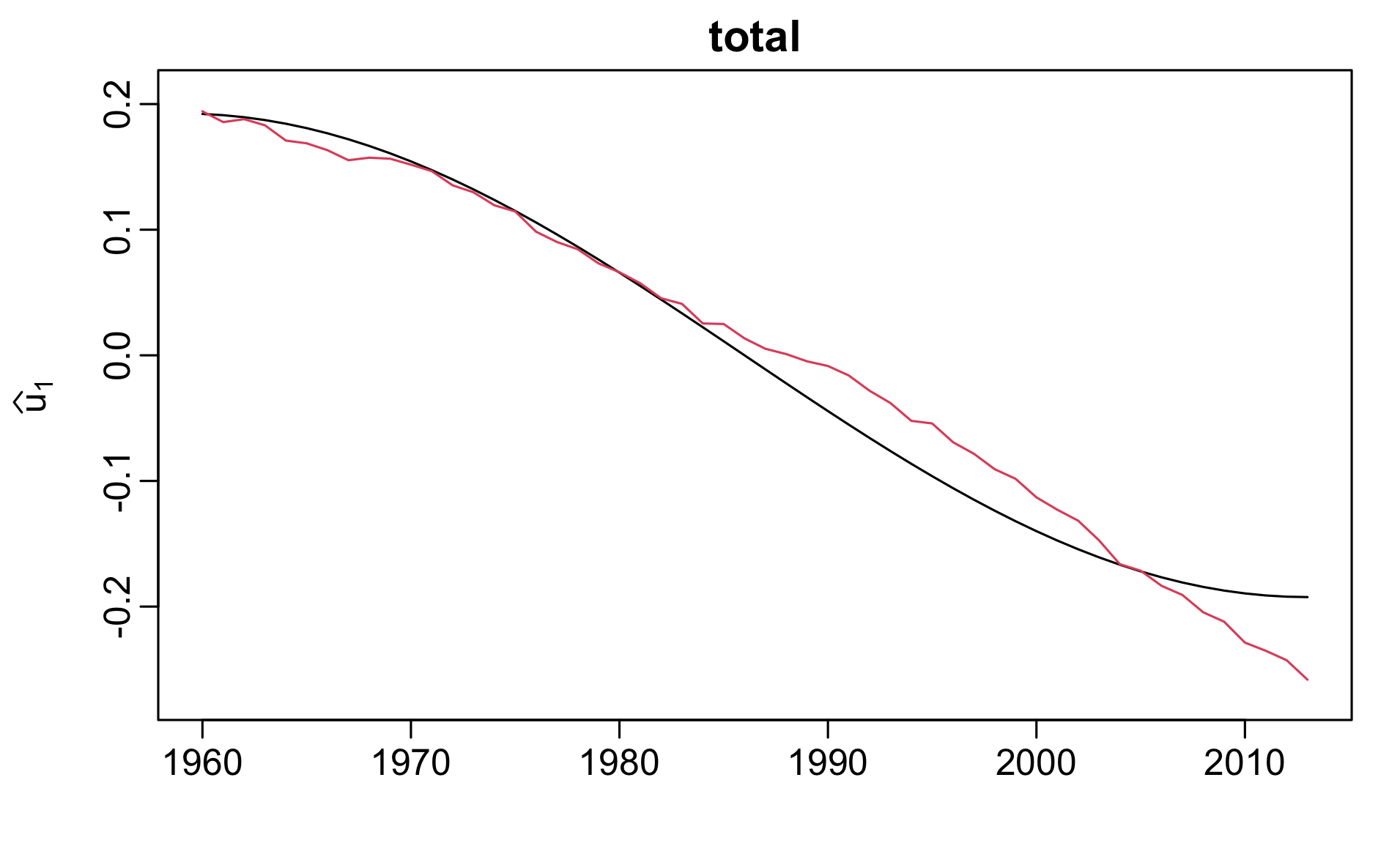}{.30}

    \subf{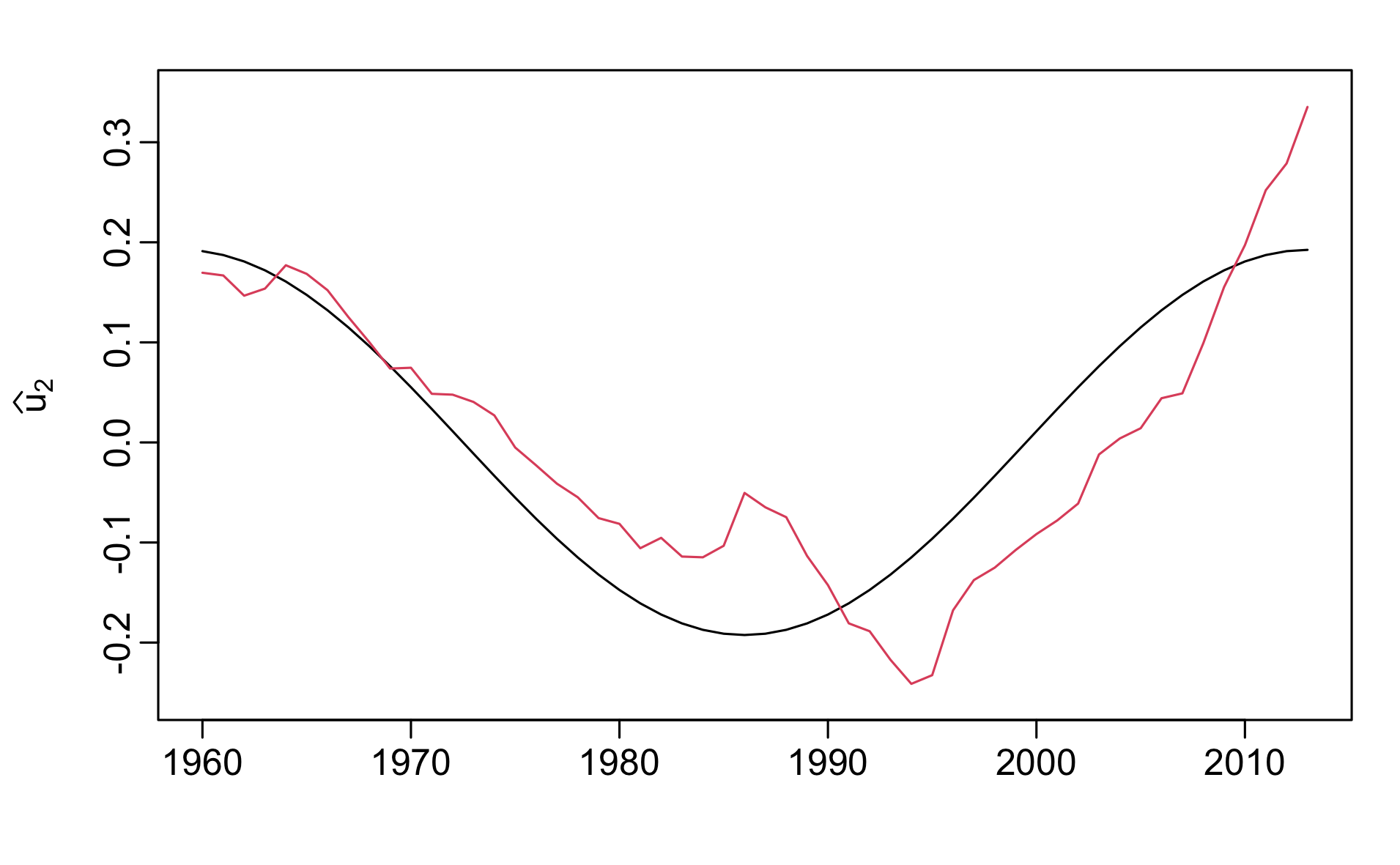}{.30}
    \subf{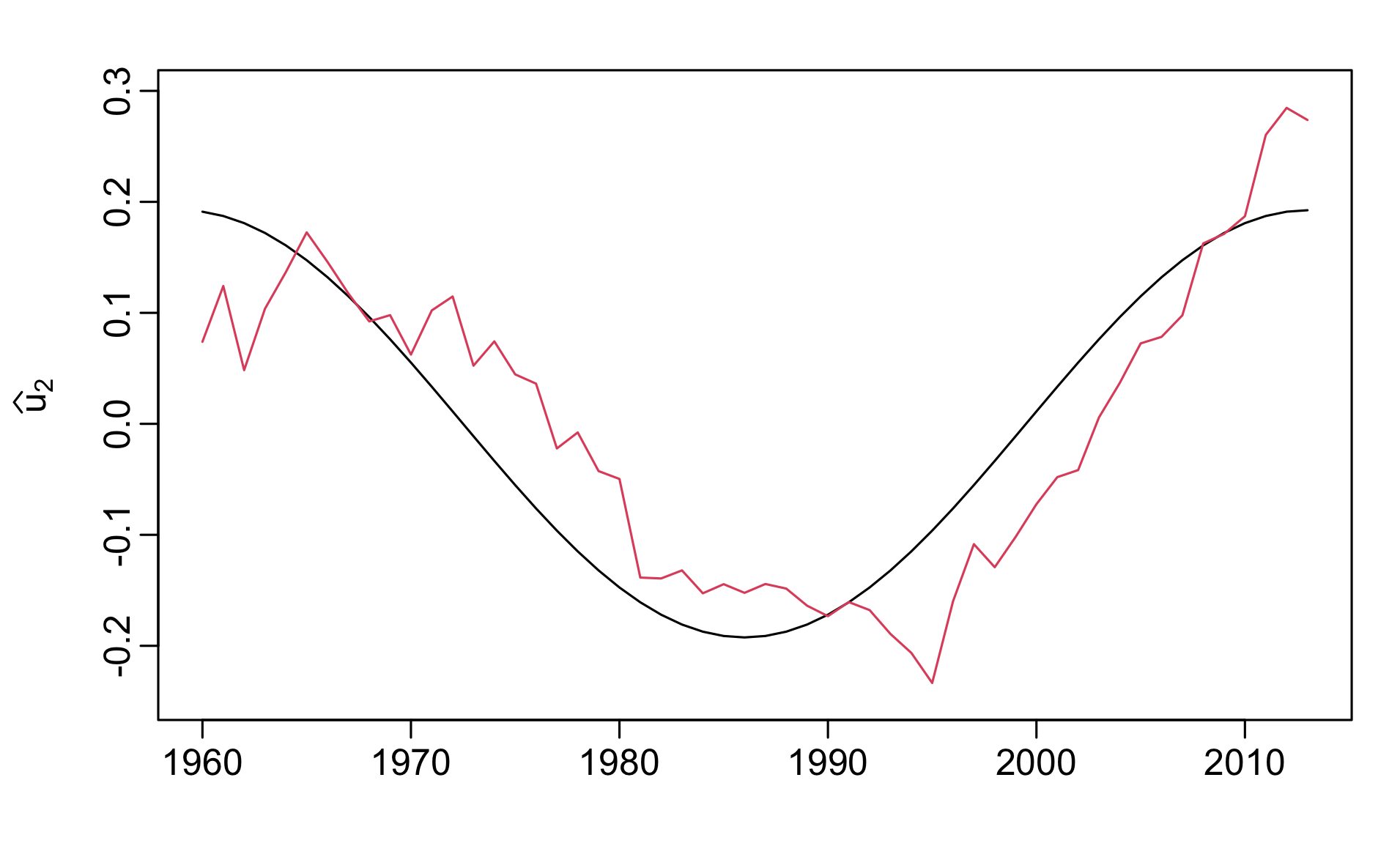}{.30}
    \subf{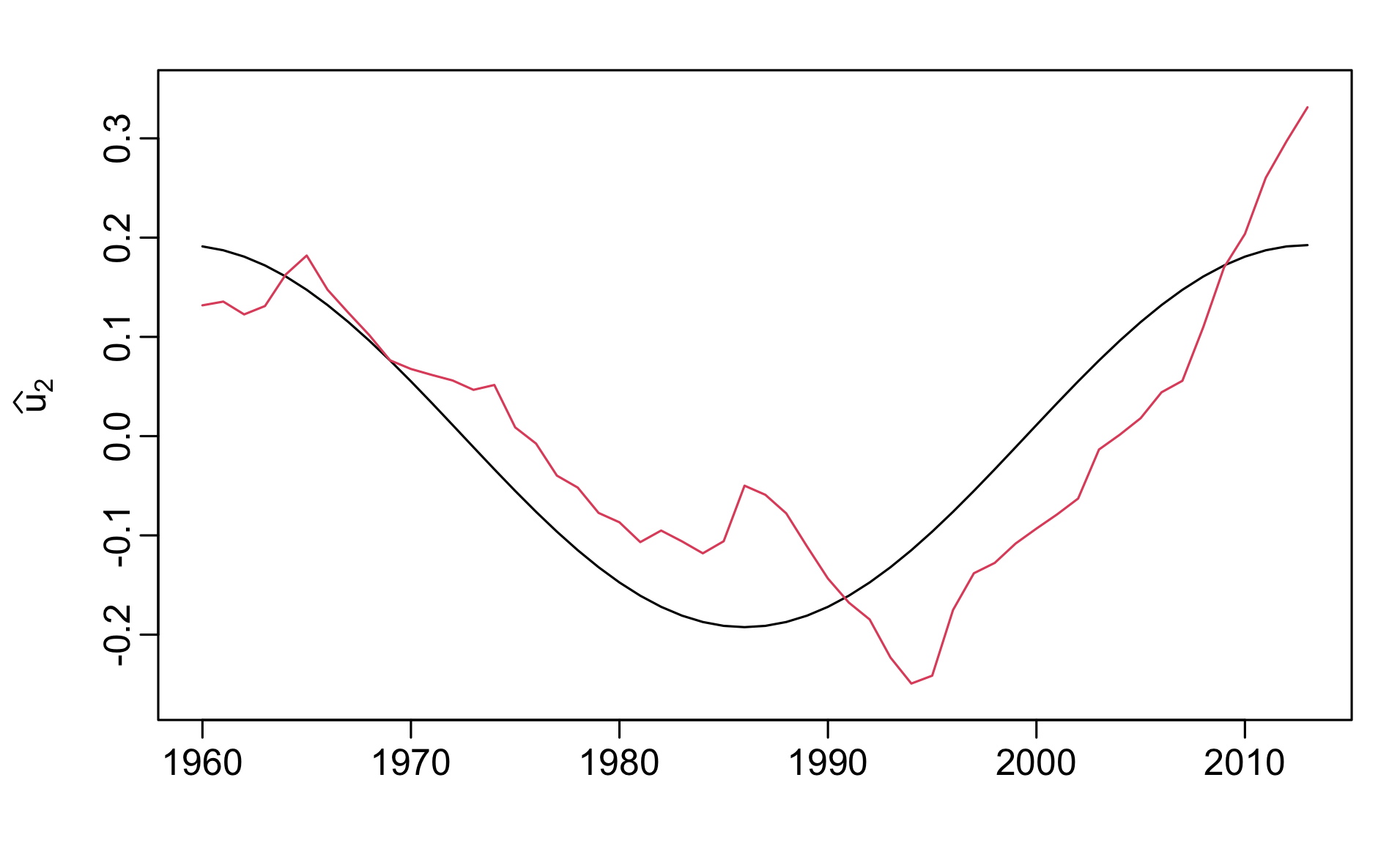}{.30}

    \subf{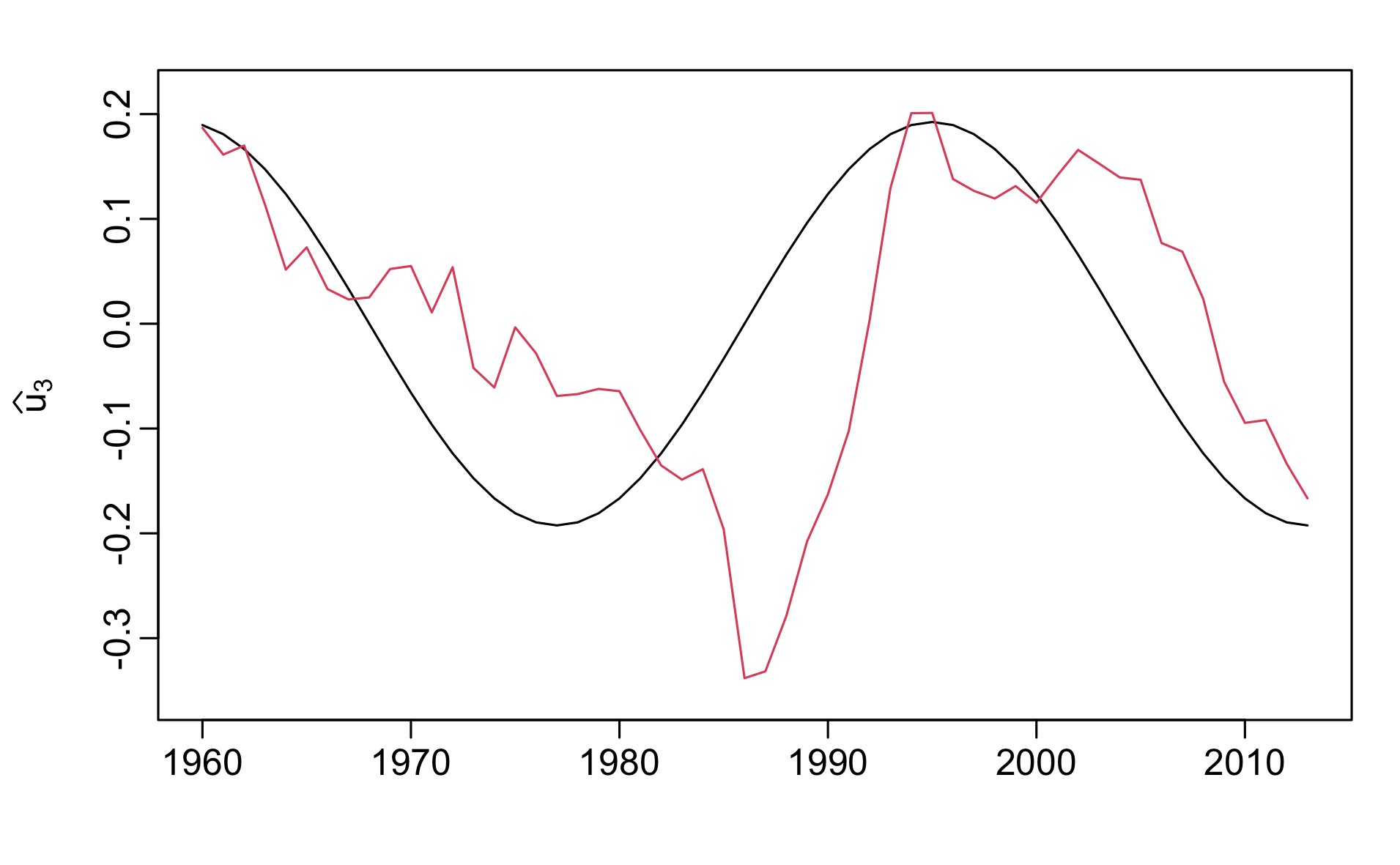}{.30}
    \subf{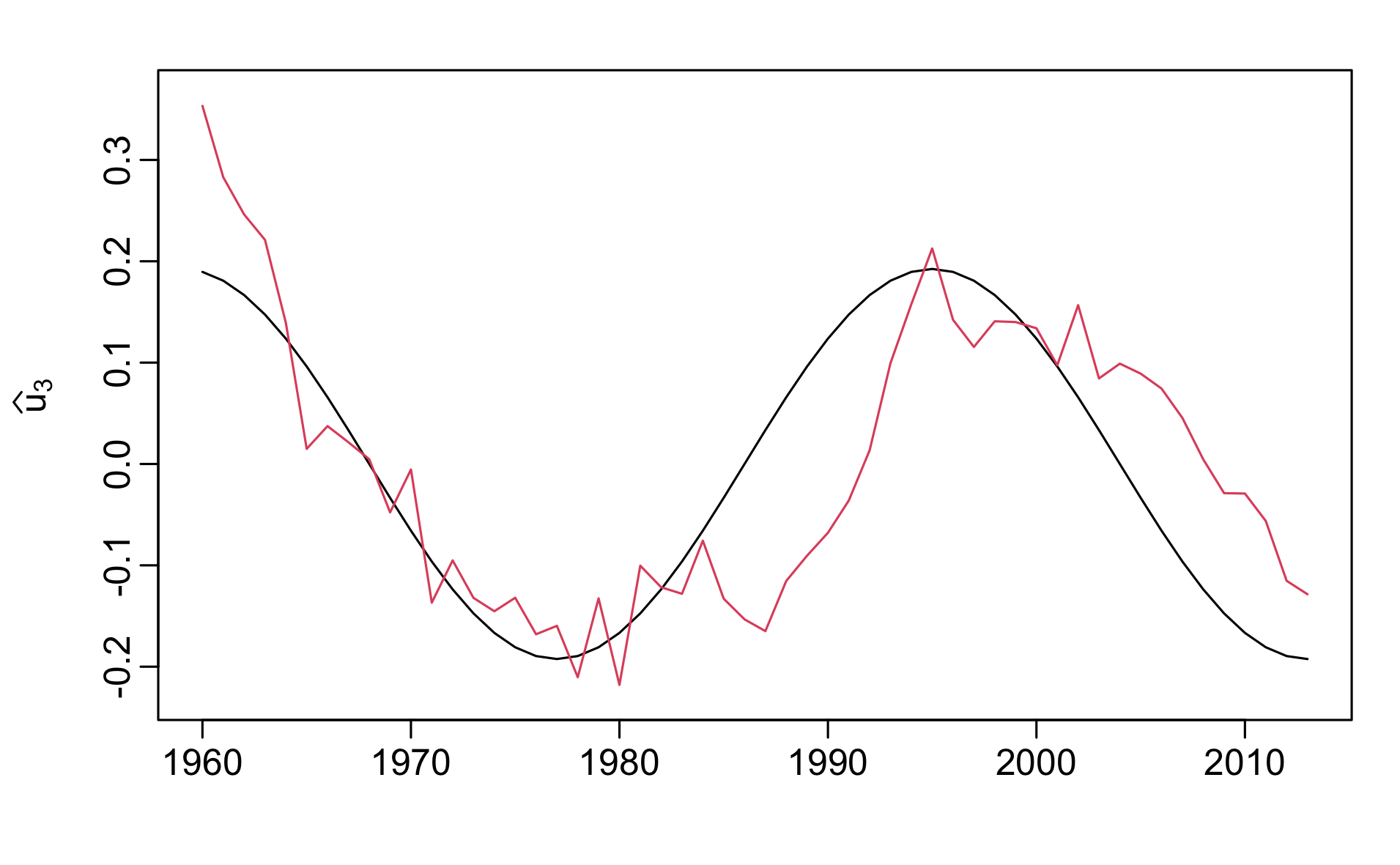}{.30}
    \subf{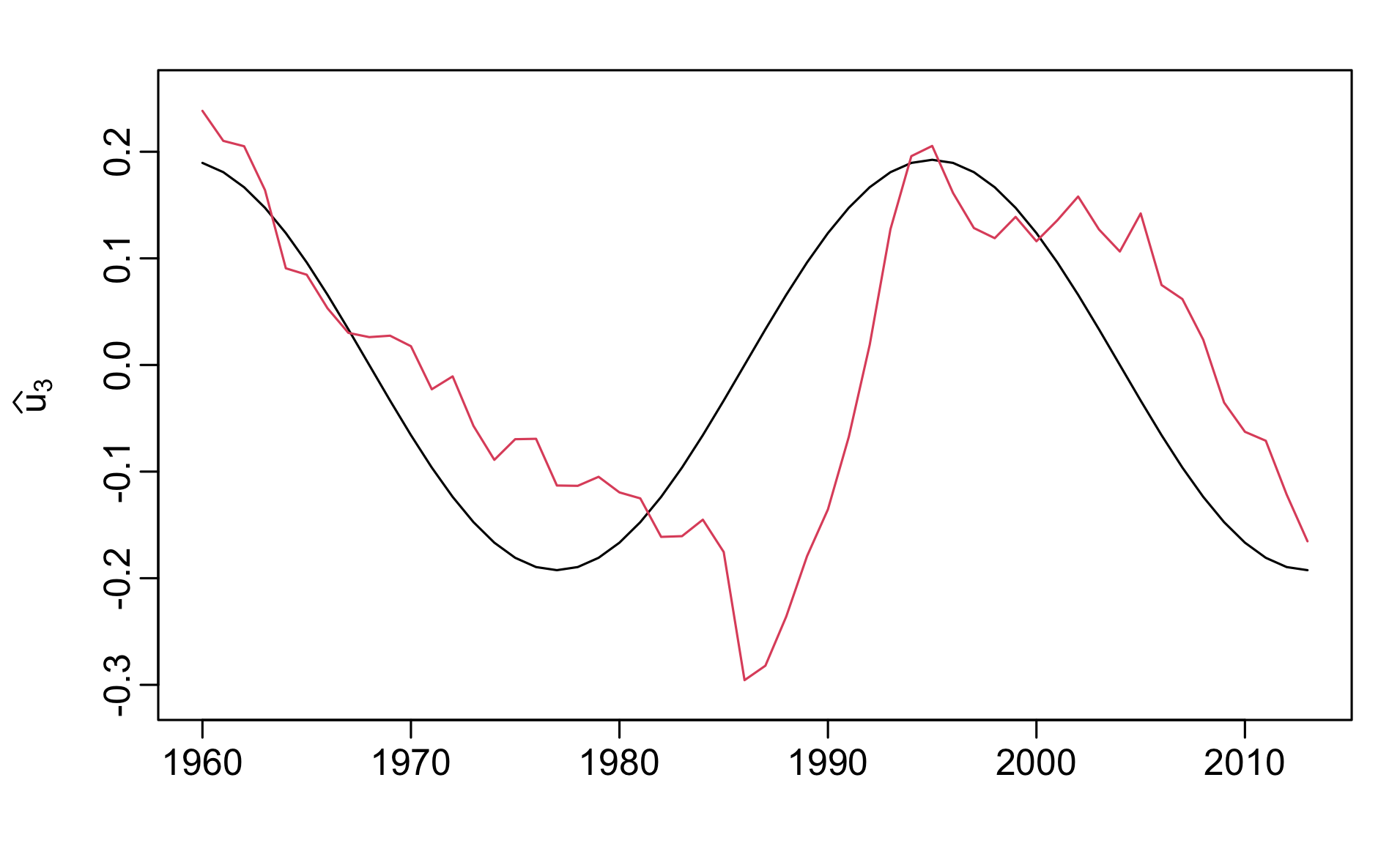}{.30}

    \subf{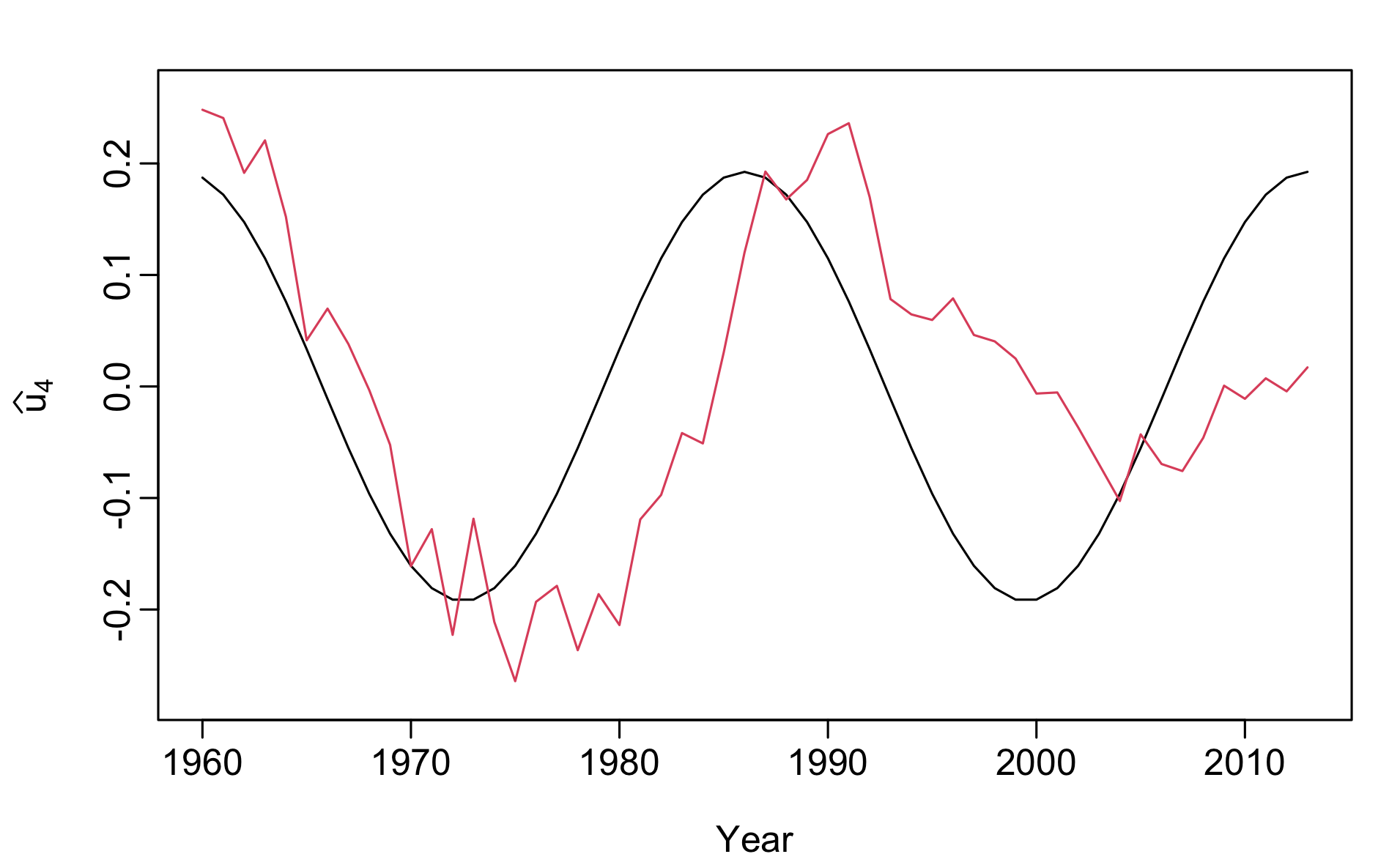}{.30}
    \subf{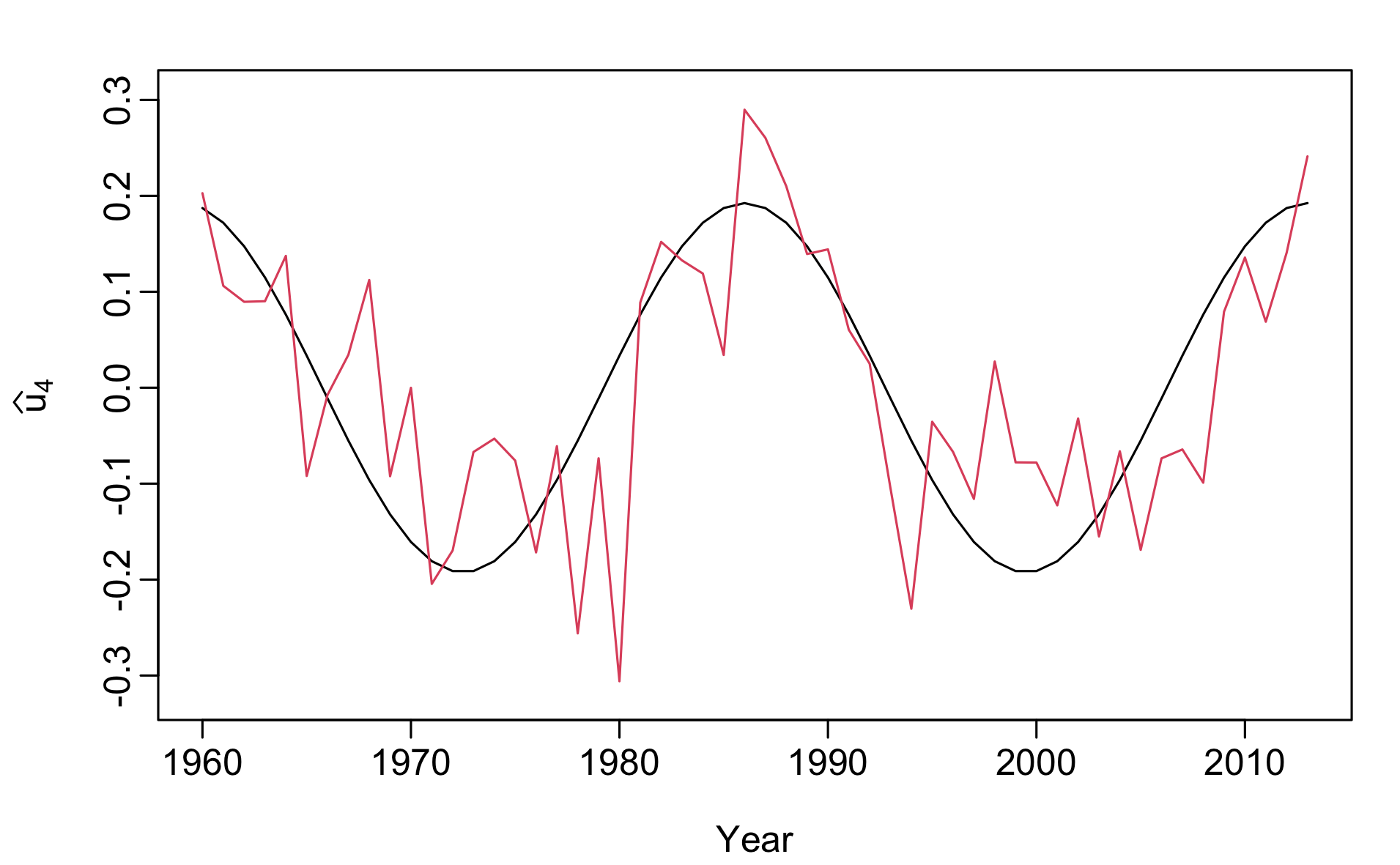}{.30}
    \subf{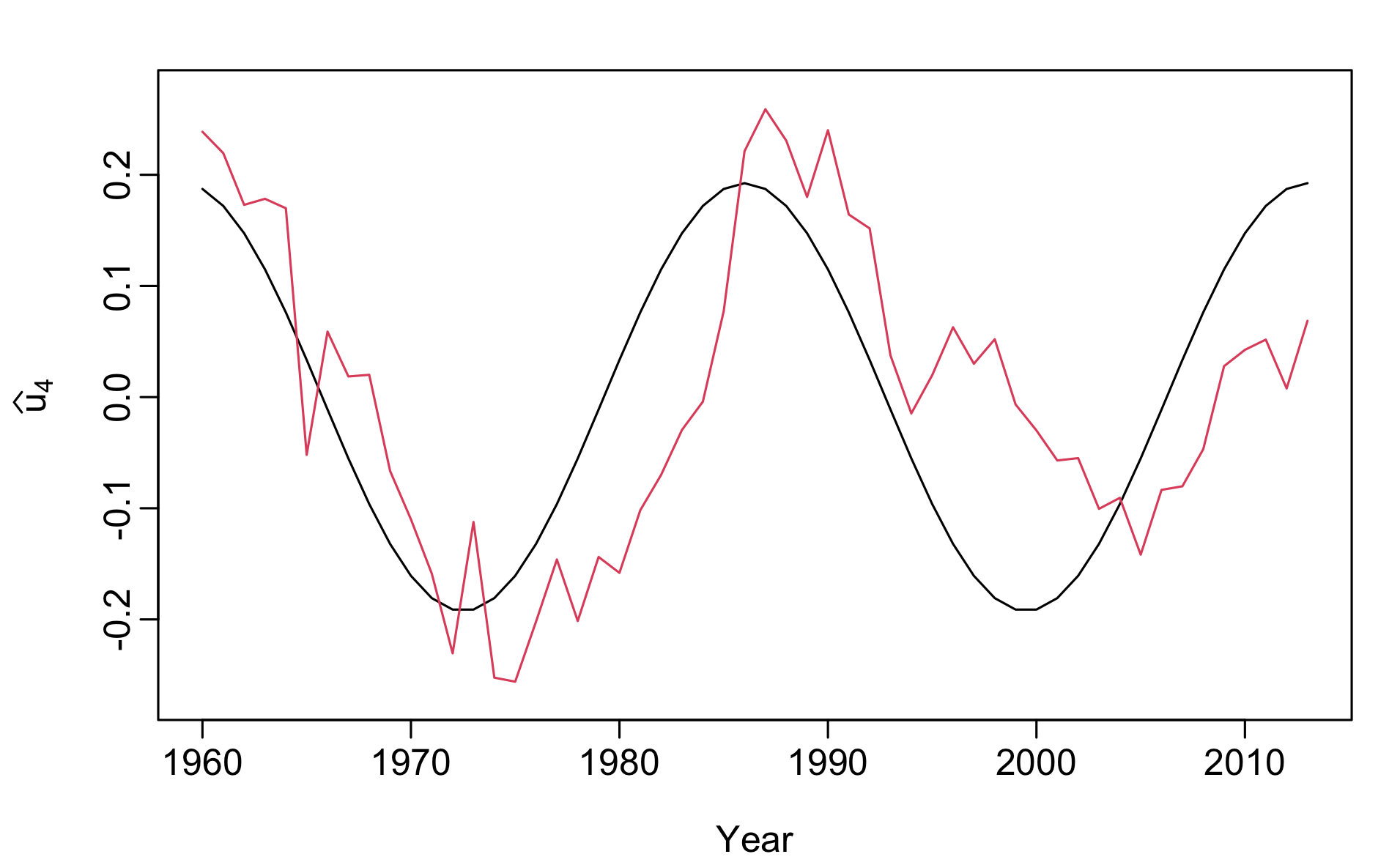}{.30}
    \caption{Sample eigenvectors of the covariance matrix of the male, female and total age-specific mortality rates of 32 countries from 1960 to 2013}
    \label{figure - multi}
\end{figure}
As can be seen, the sample eigenvectors closely resemble the trigonometric curves in the spurious limit. 
This suggests that the mortality rates are likely non-stationary and that the number of non-stationary components in the data, as measured by the effective rank $\mathcal{R}\sim\sqrt{qK}$ of the model, may be large. A large value of $\mathcal{R}$ implies three possible scenarios: Case 1 (large $q$ and small $K$): a large number of common factors ($K$) are shared across locations; Case 2 (large $K$ and small $q$): a large number of common factors ($q$) contribute to the functional variation; and Case 3 (moderately large $q$ and $K$): both the cross-location and functional dimension are driven by a relatively large number of common factors. Case 1 is consistent with findings in the mortality literature \citep{GaoShangYang2019, ZhangGaoPanYang2025, GuoQiaoWangEtAl2024}, which suggest that only a small number of common factors exist across locations. Cases 2 and 3 represent new findings for multi-location mortality data.
This raises concerns about the validity of principal component analysis based methods when analyzing this dataset without differencing. 

For comparison, we also consider the Japanese sub-national mortality rates observed for~$p=47$ prefectures over 1975-2022~($T = 48$). Since the mortality rates are recorded from different prefectures in the same country, it is reasonable to expect a higher level of homogeneity within the across the $47$ covariates. If the cross-sectional dependence (factor strength) is strong enough, i.e. if the model has low effective rank in the sense of condition~(\ref{equation - key condition}), it is possible to observe genuine factors present in the model instead of the spurious limit.

\begin{figure}[!htb]
\centering
    \subf{jpn_male_1}{.30}
    \subf{jpn_female_1}{.30}
    \subf{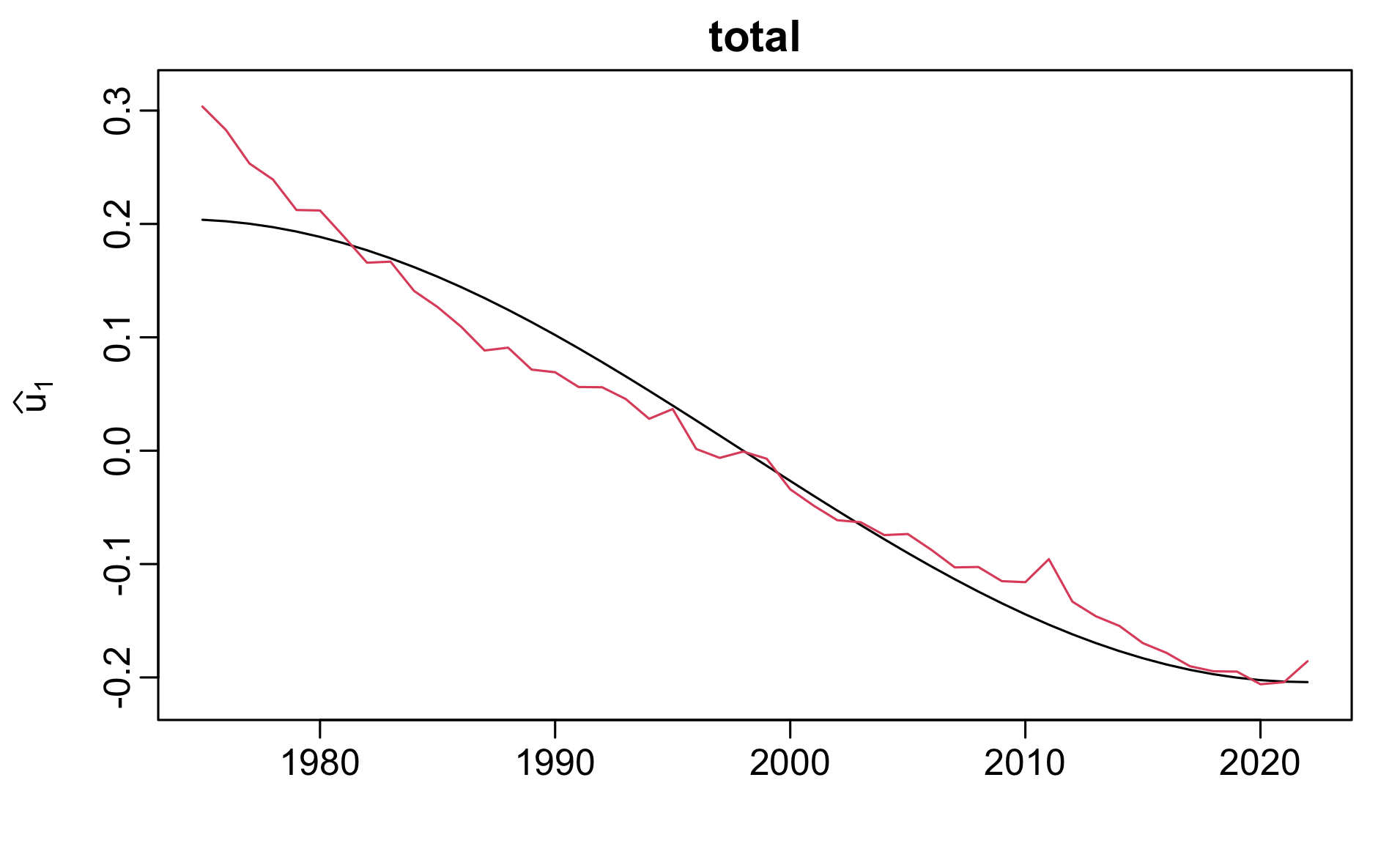}{.30}

    \subf{jpn_male_2}{.30}
    \subf{jpn_female_2}{.30}
    \subf{jpn_total_2}{.30}

    \subf{jpn_male_3}{.30}
    \subf{jpn_female_3}{.30}
    \subf{jpn_total_3}{.30}

    \subf{jpn_male_4}{.30}
    \subf{jpn_female_4}{.30}
    \subf{jpn_total_4}{.30}

    \caption{Sample eigenvectors of the covariance matrix of the male, female and total mortality rates of 47 Japanese prefectures from 1975 to 2022}
    \label{figure - jpn}
\end{figure}

Similar to the multi-country data set mentioned above, we aggregate ages above 100 before smoothing and taking the logarithm. Figure~\ref{figure - jpn} shows the sample eigenvectors of the covariance matrix defined in~\eqref{equation - gram matrix} plotted against the spurious limits of Theorem~\ref{theorem - main}. 
The eigenvectors still resemble the spurious limits, but the effects are not as visually pronounced as in the multi-country dataset. The second and third sample eigenvectors show significant deviation from the trigonometric limits at the year 2011. A possible explanation is the Tohoku Earthquake and Tsunami that took place in March of 2011 which results in nearly 20,000 deaths \citep{NakaharaIchikawa2013}. In this case, the eigenvectors could be thought of as genuine factors that contains relevant information on the dataset and instead of being labeled as spurious. Nevertheless, the overall resemblance between the sample eigenvectors and the trigonometric limits suggests that additional care should be taken when dealing with the non-stationary nature of the dataset. 

\section{Conclusion and Future Work}\label{section - conclusion}

We study the asymptotic spectral behaviour of high-dimensional functional factor models with non-stationary factors. Under a rank-type condition on the covariance operator, we show that conventional functional principal component analysis can produce spurious results, extending \cite{OnatskiWang2021} to the functional setting. In particular, the sample eigenvalues and eigenfunctions converge to limits unrelated to the true covariance structure, so that leading components may explain a large proportion of the variation even in the absence of genuine factors.
Moreover, due to the intrinsic high dimensionality of functional data, such spurious behaviour may arise even with a single non-stationary factor, in contrast to \cite{OnatskiWang2021} and \cite{HeZhang2024}. Simulation studies illustrate the theory and identify regimes in which spurious behaviour occurs. An empirical analysis of two multivariate age-specific mortality datasets further demonstrates the phenomenon under standard functional PCA.

Our simulations reveal asymptotic features of the sample eigenvectors that are not explained by existing theory. Formalizing these observations and establishing the associated conjecture is left for future work, which would provide practical tools for detecting spurious behaviour and for identifying key properties of common factors, including their effective dimensionality.
Another direction is to study the asymptotic distribution of the sample eigenvalues, in the spirit of \cite{HeZhang2024}. Such results would enable formal inference for testing the presence of high-rank components and assessing the emergence of spurious behaviour in functional factor models.

\newpage

\section*{Supplementary Material to ``Attribution of Spurious Factors from High-Dimensional Functional Time Series''}

This material consists of two parts: the first part provides proofs of theoretical results in the main paper; and the second part proposes a new conjecture on empirical eigenvectors.

\appendix
\appendixone
\section{Proofs}\label{section - proofs}
\subsection{Preliminary results}
We first give the proof to Proposition \ref{proposition - matrix of op}. 
\begin{proof}[of Proposition \ref{proposition - matrix of op}]
	For $f = (f_i)$ and $g = (g_i)$ in $\cH^K$, we have
		\begin{align*}
			\<\Omega f,g\>
			& =  \sum_{i} \< \sum_j \Omega_{ij} f_j, g_i\>
			=\sum_{j} \< f_j,  \sum_i\Omega_{ij}^* g_i\>
			=\<f, (\Omega_{ji}^*)_{ij} g\>.
		\end{align*}
		Similarly, we have
		\begin{align*}
			\<C_1 \Omega C_2 f,g\>
			& = 
			\sum_i \<C_1 \sum_j \Omega_{ij} C_2 f_j, g_i\>
			=
			\sum_j \<  f_j,  \sum_i C_2^* \Omega_{ij}^* C_1^* g_i\>
		\end{align*}
		which shows $(C_1 \Omega C_2)^* = C_2^* \Omega^* C^*$. 

		Let $(e_n)_n$ be an orthonormal basis for $\cH$. Let $(\tilde e_n)_n$ be a set of vectors in $\cH^K$ given by $\tilde e_{(n-1)K+m}:= (0, \ldots, e_{n},  \ldots,0)$ where $n=1,\ldots$,  $m=1,\ldots, K$ and $e_n$ appears in the $m$-th co-ordinate. Clearly $(\tilde e_n)$ is an orthonormal basis of $\cH^K$. Then we have
		\begin{align*}
			\sum_n\<C_1\Omega C_2 \td e_n,\td e_n\>
			& = \sum_n  \sum_{ij}\<C_1 \Omega_{ij} C_2 \td e_n^{j},\td e_n^{i}\>,
		\end{align*}
		where $\td e_n^i$ denotes the $i$\textsuperscript{th} co-ordinate of $\td e_n$. By definition, at least one of of $\td e_n^i$ and $\td e_n^j$ is 0 whenever $i\ne j$. Therefore
		\begin{align*}
			\sum_n\<C_1\Omega C_2 \td e_n,\td e_n\>=
			\sum_i \sum_n \<C_1\Omega_{ii}C_2 e_n, e_n\> = \sum_i \<C_1\Omega_{ii}C_2\>,
		\end{align*}
		which gives the trace of $C \Omega$. 
		Finally, we have
		\begin{align*}
		 	\norm{C_1 \Omega C_2}_F^2=\langle C_1 \Omega C_2C_2^* \Omega^* C_1^* \rangle
		 	& = 
		 	\sum_i \<C_1 (\Omega C_2 C_2^* \Omega^*)_{ii} C_1^*\>
		 	,
		 \end{align*} 
		where the last equality follows from the second claim of this Proposition. 
		Since $(\Omega^*)_{ji} = \Omega^*_{ij}$ by part 1 of this Proposition, we have
		\begin{align*}
			\norm{C_1 \Omega C_2}_F^2
			=
		 	\sum_{ij} \<C_1 \Omega_{ij} C_2 C_2^* (\Omega^*)_{ji} C_1^*\>
			=\sum_{ij} \<C_1 \Omega_{ij} C_2 C_2^* \Omega^*_{ij} C_1^*\>
			=\sum_{ij} \norm{C_1 \Omega_{ij} C_2}_F^2
		\end{align*}
		where $\norm{C_1 \Omega_{ij} C_2}_F$ denotes the Hilbert-Schmidt norm of $C_1 \Omega_{ij} C_2$ on $\cH$. 
\end{proof}
The following lemma is the key technical result underpinning  Theorem \ref{theorem - main}. It computes moments of certain bilinear forms associated with the random matrix $W$ appearing in \eqref{eq- def of W}. 
\begin{lemma}\label{lemma - moments of a'Wb}
	Let $W$ be the $T\times T$ random matrix defined in \eqref{eq- def of W}. Under Assumption \ref{assumptions - epsilon iid}, 
	for any deterministic vectors $a, b\in \Rm^T$, we have
	\begin{align*}
		\Em[a'W b ] = \<a,b\> \langle C_\epsilon \Omega \rangle,
		\numberthis\label{equation - lemma bilinear form, E[aWb]}
	\end{align*}
	Furthermore, for any deterministic vectors $a,b,c,d \in\Rm^T$, we have
	\begin{align*}
		\Big|\Cov(a'Wb, c'Wd)
		-  \<a,c\>\<b,d\>\norm{C_\epsilon^{1/2} \Omega C_\epsilon^{1/2}}_F^2 - &\<a,d\>\<b,c\> \norm{C_\epsilon^{1/2} \Omega C_\epsilon^{1/2}}_F^2 \Big|
		  \\
		& \le 2 \kappa_4^*
		\sum_t a_t b_t c_t d_t
		\norm{C_\epsilon^{1/2} \Omega C_\epsilon^{1/2}}_F^2,
	\end{align*}
	where $a_t, b_t, c_t$ and $d_t$ are the $t$-th coordinate of $a,b,c$ and $d$ respectively.

	\begin{proof}
	From the definition of $W$ we may write
	\begin{align*}
		a'Wb& = 
		\sum_{i,j=1}^K \iint
		\<a, \epsilon_{i\cdot}(v)\>  \<b, \epsilon_{j\cdot}(v)\> \Omega_{ij}(v, w) dv dw,
	\end{align*}
	where $\epsilon_{i\cdot}(v)\in\Rm^T$ is the $i$-th column of the matrix $\epsilon'(v)$. 
	Since $(\epsilon_{it})_{it}$ is a matrix of uncorrelated elements by assumptions, we have
	\begin{align*}
		\Em[a'Wb]
		& =  \sum_{i,j=1}^K \iint \sum_{s,t=1}^T   \Em[a_s \epsilon_{is}(v) b_t \epsilon_{jt}(w)] \Omega_{ij}(v,w) dv dw
		\\
		& =  \sum_{t=1}^T  a_t b_t \sum_{i=1}^K\iint  \Em[\epsilon_{it}(v) \epsilon_{it}(w)] \Omega_{ii}(v,w) dv dw
		\numberthis\label{equation - C_e tr W in proof of bilinear form lemma}
		  \\
		& =  \<a,b\> \sum_{i=1}^K\iint  C_\epsilon(v,w) \Omega_{ii}(v,w) dv dw.
	\end{align*}
	By Proposition \ref{proposition - matrix of op} we obtain
	\begin{align*}
		\Em[a'Wb] 
		=\sum_{i=1}^K \<C_\epsilon \Omega_{ii}\>
		& =   \<a, b\> \langle C_\epsilon \Omega \rangle.
	\end{align*}
	This gives the equality in \eqref{equation - lemma bilinear form, E[aWb]}.
	For second moments, we write 
	\begin{align*}
		\Em[a'Wb, c'Wd]
		 = 
		\sum_{i,j,k,l=1}^K 
		\iiiint
		A_{ijkl}(v,w, \alpha, \beta)
		\Omega_{ij}(v,w) \Omega_{kl}(\alpha, \beta)
		dv dw d \alpha d \beta,
	\end{align*}
	where we have defined 
	\begin{align*}
		A_{ijkl}(v,w, \alpha, \beta):& = 
		\Em[ \<a, \epsilon_{i\cdot}(v)\> \<b, \epsilon_{j\cdot}(w)\>
		\<c, \epsilon_{k\cdot}(\alpha)\>\<d, \epsilon_{l\cdot}(\beta)\>].
	\end{align*}
	We will first compute $A_{ijkl}$ and $\sum_{ijkl}\Omega_{ij}(v,w) \Omega_{kl}(\alpha, \beta)$ for all the possible partitions of index set $\{ i,j,k,l\}$. For ease of notation we will write
	\begin{align*}
		\tilde \Omega:=  \sum_i \Omega_{ii}(v,w) \Omega_{ii}(\alpha, \beta).
	\end{align*} 
	On the set $\{i=j\ne k=l\}$, we have
	\begin{align*}
		A_{iikk} (v,w, \alpha, \beta)
		& = 
		\Em[\<a, \epsilon_{i\cdot}(v)\> \<b, \epsilon_{i\cdot}(w)\>]
		\Em[\<c, \epsilon_{k\cdot}(\alpha)\> \<d, \epsilon_{k\cdot}(\beta)\>]
		\\
		& = \<a,b\> \<c,d\>C_\epsilon(v,w)C_\epsilon(\alpha,\beta),
	\end{align*}
	where the last line follows from similar computations as in \eqref{equation - C_e tr W in proof of bilinear form lemma}, and
	\begin{align*}
		\sum_{i=j\ne k=l}  \Omega_{ij}(v,w) \Omega_{kl}(\alpha, \beta)
		& = \sum_{i\ne k} \Omega_{ii}(v,w) \Omega_{kk}(\alpha, \beta)
	\\
		& =  \tr \Omega(v, w) \tr \Omega(\alpha, \beta)
		- \tilde \Omega.
	\end{align*}
	Similarly, on the set $\{i=k\ne j=l\}$, we have
	\begin{align*}
		A_{ijij} (v,w, \alpha, \beta)
		& = 
		\Em[ \<a, \epsilon_{i\cdot}(v)\> 
		\<c, \epsilon_{i\cdot}(\alpha)\>\>]
		\Em[  \<b, \epsilon_{j\cdot}(w)\>
		\<d, \epsilon_{j\cdot}(\beta)\>]
		\\
		& =
		\<a,c\> \<b,d\> C_\epsilon(v,\alpha)C_\epsilon(w,\beta),
	\end{align*}
	and
	\begin{align*}
		\sum_{i=k\ne j=l}  \Omega_{ij}(v,w) \Omega_{kl}(\alpha, \beta)
		& = \sum_{i\ne j} \Omega_{ij}(v,w) \Omega_{ij}(\alpha, \beta)
	\\
		& =  \tr \Omega(v, w)' \Omega(\alpha, \beta)
		- \tilde \Omega.
	\end{align*}
	On the set $\{i=l\ne j=k\}$, we have
	\begin{align*}
		A_{ijji} (v,w, \alpha, \beta)
		& = 
		\Em[ \<a, \epsilon_{i\cdot}(v)\> 
		\<d, \epsilon_{i\cdot}(\beta)\>\>]
		\Em[  \<b, \epsilon_{j\cdot}(w)\>
		\<c, \epsilon_{j\cdot}(\alpha)\>]
		\\
		& =
		\<a,d\> \<b,c\> C_\epsilon(v,\beta)C_\epsilon(w,\alpha),
	\end{align*}
		and
	\begin{align*}
		\sum_{i=l\ne j=k}  \Omega_{ij}(v,w) \Omega_{kl}(\alpha, \beta)
		& = \sum_{i\ne j} \Omega_{ij}(v,w) \Omega_{ji}(\alpha, \beta)
	\\
		& =  \tr \Omega(v, w) \Omega(\alpha, \beta)
		- \tilde \Omega.
	\end{align*}
	In cases where exactly three of the four indices coincide we have $A_{i,j,k,l\equiv0}$ since $\Em[\epsilon]=0$.  Finally on the diagonal $\{i=j=k=l\}$, we have
	\begin{align*}
		A_{iiii}(v,w, \alpha, \beta)
		=
		\Em[ \<a, \epsilon_{i\cdot}(v)\>\<b, \epsilon_{i\cdot}(w)\>\<c, \epsilon_{i\cdot}(\alpha)\>\<d, \epsilon_{i\cdot}(\beta)\> ]
	\end{align*}
	and
	\begin{align*}
		\sum_{i=j=k=l}  \Omega_{ij}(v,w) \Omega_{kl}(\alpha, \beta)
		& = \sum_{i} \Omega_{ii}(v,w) \Omega_{ii}(\alpha, \beta) = \tilde \Omega.
	\end{align*}
	Expanding the expression for $A_{iiii}$ and partitioning the index set in a similar way, we get
	\begin{align*}
		A_{iiii}&(v,w,\alpha, \beta)
		 = 
		\sum_{s,t, m,n}
		\Em[a_s b_t c_m d_n \epsilon_{is}(v) \epsilon_{it}(w)\epsilon_{im}(\alpha)	\epsilon_{in}(\beta)]
		\\
		& =
		\left(\sum_{s=t\ne m=n} + \sum_{s=m\ne t=n} + \sum_{s=n\ne t=m}
		+ \sum_{s=t=m=n}\right)
		\Em[a_s b_t c_m d_n \epsilon_{is}(v) \epsilon_{it}(w)\epsilon_{im}(\alpha)	\epsilon_{in}(\beta)],
	\end{align*}
	since $\Em[\epsilon_{is}\epsilon_{it}\epsilon_{im}\epsilon_{in}]=0$ on sets of exactly three coinciding indices. Therefore
	\begin{align*}
		A_{iiii}(v,w, \alpha, \beta)
		= 
		\<a,b\>&\<c,d\>C_\epsilon(v, w)C_\epsilon (\alpha, \beta)
		  \\
		& 
		+ \<a, c\> \<b,d\> C_\epsilon(v, \alpha)C_\epsilon (w, \beta)
		  \\
		& +
		\<a,d\>\<b,c\>C_\epsilon(v, \beta)C_\epsilon (w, \alpha)
		  \\
		& + \sum_t \kappa_{it}(v,w,\alpha, \beta)  a_t b_t c_t d_t
	\end{align*}
	where we have defined
	\begin{align*}
		\kappa_{it} (v,w, \alpha, \beta):=
		\Em[\epsilon_{it}(v) \epsilon_{it}(w) \epsilon_{it}(\alpha) \epsilon_{it}(\beta)-3].
	\end{align*}
	Recall from \eqref{equation - lemma bilinear form, E[aWb]} that
	\begin{align*}
		\Em[a'Wb] = \<a,b\>\langle C_\epsilon \Omega \rangle,
		\quad
		\Em[c'Wd] =\<c,d\>\langle C_\epsilon \Omega \rangle.
	\end{align*}
	After collecting the above terms and simplifying, we note that 
	all terms containing the quantity $\tilde \Omega$ cancel out and we obtain
\begin{align*}
\Cov(a'Wb,&\ c'Wd) = 
\<a,c\>\<b,d\> \int_{\cI^4} C_{\epsilon} (v,\alpha)C_\epsilon(w,\beta) \tr \Omega(v,w)' \Omega(\alpha, \beta) dv dw  d \alpha d \beta
\\
		& 
		+ \<a,d\>\<b,c\> \int_{\cI^4} C_{\epsilon} (v,\beta)C_\epsilon(w,\alpha) \tr \Omega(v,w) \Omega(\alpha, \beta)
		dv dw  d \alpha d \beta
		  \\
		&
		+
		\sum_t a_t b_t c_t d_t
		\sum_i\int_{\cI^4}
		\kappa_{it}(v,w,\alpha, \beta) \Omega_{ii}(v, w) \Omega_{ii}(\alpha, \beta)dv dw  d \alpha d \beta.
	\end{align*}
	Recall that $C_\epsilon$ is self-adjoint. 
	By definition of $\Omega$ and Proposition \ref{proposition - matrix of op}, we have
	\begin{align*}
		\Omega_{ij}(u,v) = \Omega_{ji}^*(u,v) =  \Omega_{ji}(v,u).
		\numberthis\label{equation - Omega ij and adjoints}
	\end{align*}
	By linearity, the first term in the above expression simplifies to
	\begin{align*}
		& \int C_{\epsilon} (v,\alpha) C_\epsilon(w,\beta) \tr \Omega(v,w)' \Omega(\alpha, \beta) dv dw  d \alpha d \beta
		\\ 
        = & 
		\sum_{i,j} \int_{\cI^4} C_\epsilon(\beta,w) \Omega_{ij}(w,v)C_\epsilon(v,\alpha) \Omega_{ji}(\alpha,  \beta) 
		dv dw  d \alpha d \beta
		  \\
		= & \sum_{i,j}
		\<C_\epsilon \Omega_{ij} C_\epsilon \Omega_{ji}\>
	\end{align*}
	Using \eqref{equation - Omega ij and adjoints} and the self-adjointness of $C_\epsilon$ again, we obtain
	\begin{align*}
		\sum_{i,j}
		\<C_\epsilon \Omega_{ij} C_\epsilon \Omega_{ji}\>
		& =\sum_{i,j}
		\<C_\epsilon^{1/2} \Omega_{ij} C_\epsilon^{1/2} C_\epsilon^{1/2} \Omega_{ij}^* C_\epsilon^{1/2}\>
		  \\
		& =
		\sum_{i,j}
		\norm{C_\epsilon^{1/2} \Omega_{ij} C_\epsilon^{1/2}}_F^2
		=
		\norm{C_\epsilon^{1/2} \Omega C_\epsilon^{1/2}}_F^2.
	\end{align*}
	Similarly we have
	\begin{align*}
		\int_{\cI^4} C_{\epsilon} (v,\beta)C_\epsilon(w,\alpha) \tr \Omega(v,w) \Omega(\alpha, \beta)
		dv dw  d \alpha d \beta
		= \norm{C_\epsilon^{1/2} \Omega C_\epsilon^{1/2}}_F^2. 
	\end{align*}
	Therefore we have
\begin{align*}
& \Cov(a'Wb, c'Wd)
	 -  (\<a,c\>\<b,d\> + \<a,d\>\<b,c\>) \norm{C_\epsilon^{1/2} \Omega C_\epsilon^{1/2}}_F^2
\\
= & \sum_t a_t b_t c_t d_t \sum_i\int_{\cI^4}\kappa_{it}(v,w,\alpha, \beta) \Omega_{ii}(v, w) \Omega_{ii}(\alpha, \beta)dv dw  d \alpha d \beta
\end{align*}
and it remains to bound the last term. By Assumption \ref{assumptions}.\ref{assumptions - epsilon 4th moemnt}, we have
\begin{align*}
&\left|\sum_i\int_{\cI^4}\kappa_{it}(v,w,\alpha, \beta) \Omega_{ii}(v, w) \Omega_{ii}(\alpha, \beta)dv dw  d \alpha d \beta\right|
\\
& \qquad\qquad \le \sum_i  \int_{\cI^4} C_\epsilon(\beta,v) \Omega_{ii}(v,w)
		C_{\epsilon}(w, \alpha) \Omega_{ii}(\alpha, \beta) dv dw d \alpha d \beta
\\
& \qquad\qquad
=\sum_i \<C_\epsilon \Omega_{ii} C_\epsilon \Omega_{ii}\>.
\end{align*}
By Proposition \ref{proposition - matrix of op}(\ref{proposition - matrix of op - trace}), we have
\begin{align*}
\sum_i \<C_\epsilon \Omega_{ii} C_\epsilon \Omega_{ii}\>
	& = 
	\sum_i \norm{C_\epsilon^{1/2} \Omega_{ii} C_\epsilon^{1/2}}_F^2
	\le 
	\sum_{i,j} \norm{C_\epsilon^{1/2} \Omega_{ij} C_\epsilon^{1/2}}_F^2
	= \norm{C_\epsilon^{1/2} \Omega C_\epsilon^{1/2}}_F^2,
\end{align*}
where the last equality follows from Proposition~\ref{proposition - matrix of op}(\ref{proposition - matrix of op - HS}). 
\end{proof}
\end{lemma}

The following lemma gives the spectral decomposition of the matrix $M \Theta'$. 
\begin{lemma}
	\label{lemma - spectral decomp of MTheta}
	The matrix $M\Theta'$ admits a singular value decomposition
	\begin{align*}
		M \Theta' = \sum_{t=1}^T \sigma_t w_t v_t', 
	\end{align*}
	where $\sigma_t = (2\sin(\frac{t\pi}{2T} ))^{-1}$ and 
	\begin{align*}
	 	w_{tn}
	 	=
	 	-\sqrt{\frac{2}{T}} \cos\left( \frac{(n-1/2)\pi t}{T}\right),
	 	\quad
	 	v_{tn}
	 	=
	 	\sqrt{\frac{2}{T}} \sin\left( \frac{(n-1)\pi t}{T}\right)
	\end{align*} 
	for $t<T$, and $\sigma_T =0$, $w_{Tn} = 1/\sqrt{T}$ and $v_{Tn} = 1_{\{n=1\}}$. 
	Furthermore, we have $\sigma_n\sim T$ for any fixed $n$ as $T\to\infty$. 
	\begin{proof}
		The spectral decomposition follows from Lemma 5 of \cite{OnatskiWang2021}. 
		For the last estimate, note that  $\sin(\frac{n\pi}{2T}) \approx  \frac{n\pi}{2T}$ as $T\to\infty$, so $\sigma_n \sim T$. 
	\end{proof}
\end{lemma}

Using Lemma \ref{lemma - moments of a'Wb} and Lemma \ref{lemma - spectral decomp of MTheta}, we may derive the following concentration bounds for certain quadratics forms of $W$.
\begin{lemma}\label{lemma - concentration of a'Wb}
	Let $(\sigma_i)_i, (v_i)_i$ be as defined in Lemma \ref{lemma - spectral decomp of MTheta}. 
	Under Assumptions \ref{assumptions}-\ref{assumptions - zeta}, for any $i\le j\le T$, we have the following estimate on the quadratic form of $W$:
	\begin{align*}
		v_i' W v_j = \delta_{ij}\langle C_\epsilon \Omega \rangle + \langle C_\epsilon \Omega \rangle o_p(1).
	\end{align*}
	Furthermore, we have
	\begin{align*}
		\sum_{k=i}^j  \sigma_k^2 v_k' W v_k
		& = \langle C_\epsilon \Omega \rangle 
		\sum_{k=i}^j \sigma_k^2 + \langle C_\epsilon \Omega \rangle o_p(T^2). 
	\end{align*}
	\begin{proof}
		Since $\<v_i, v_j\> = \delta_{ij}$, by
		Lemma \ref{lemma - moments of a'Wb}
		we have $\Em[v_i' W v_j] = \delta_{ij} \langle C_\epsilon \Omega \rangle$. 
		Furthermore, 
		\begin{align*}
			\Var(v_i'W v_j) 
			& = 
			(1+ \delta_{ij})\norm{C^{1/2}_\epsilon \Omega C_\epsilon^{1/2}}_F^2
			+ 
			O(2 	\kappa_4^* \sum_t v_{it}^2 v_{jt}^2 \norm{C^{1/2}_\epsilon \Omega C_\epsilon^{1/2}}_F^2)
			  \\
			& = O(\norm{C^{1/2}_\epsilon \Omega C_\epsilon^{1/2}}_F^2),
		\end{align*}
		where the last estimate holds since $\sum_t v_{it}^2v_{jt}^2 \le \sum_{s,t} v_{it}^2 v_{js}^2 =1$. 
		By Chebyshev's inequality
		\begin{align*}
			v_i' W v_j
			=
			\delta_{ij} \langle C_\epsilon \Omega \rangle
			+ O_p ( \norm{C_\epsilon^{1/2} \Omega C_\epsilon^{1/2}}_F) 
		\end{align*}
		and the first claim follows from  Assumption \ref{assumptions - Omega rank}.
		For the second estimate, recall from Lemma \ref{lemma - spectral decomp of MTheta} that $\sigma_k^{-1} \sim \sin (k/T) $ as $T\to\infty$. Clearly this implies $\sigma_k\sim T$. Since $i,j$ are fixed, this implies $\sum_{k=i}^j \sigma_k^2 = O(T^2)$ and the second claim follows. 
	\end{proof}
\end{lemma}

The following algebraic identity from Lemma 3.3 of \cite{YinEsserXin2014} gives a geomatric interpretation on the ratio between $\ell^1$ and $\ell^2$ norms of a vector. 
\begin{lemma}[Lemma 3.3 of \cite{YinEsserXin2014}]\label{lemma - ratio l1/l2}
	Let $x \in\Rm^n $. Then
	\begin{align*}
		\frac{{2\alpha(x)^{1/2}}}{1+ \alpha(x)} \norm{x}_0^{1/2}
		\le
		\frac{\norm{x}_1}{\norm{x}_2} \le  \norm{x}_0^{1/2}
	\end{align*}
	where $\norm{x}_0:=|\{i\le n, x_i\ne0 \}|$
	and $\alpha(x) $ is given by
	\begin{align*}
		\alpha(x):= \frac{\inf_{i\le n, x_i\ne 0} x_i}{\sup_{i\le n, x_i\ne0} x_i}.
	\end{align*}
\end{lemma}


Finally,  Lemma \ref{lemma - same spectrum} is a linear algebraic result used to compute the spectrums of certain operators. 
\begin{lemma}\label{lemma - same spectrum}
	Let $\cH$ be a Hilbert space with inner product $\<\cdot, \cdot\>$ and $p,q, T$ be positive integers. 
	Suppose $f_1, \ldots, f_T$ and $g_1, \ldots, g_T$ are elements of $\cH^p$ with components $f_t = (f_{1t}, \ldots, f_{pt})$ and $g_t = (g_{1t}, \ldots, g_{pt})$ respectively. 
	Let $A=(A_{st})$ be the  $T\times T$ matrix given by
	\begin{align*}
		A_{st} = \sum_{i=1}^p \<f_{is}, g_{it}\>, \quad s,t=1,\ldots, T
	\end{align*}
	and  $\mathcal A$ be the compact operator on $\cH^p$ given by
	\begin{align*}
		\mathcal A  = \sum_{t=1}^T f_t\otimes g_t. 
	\end{align*}
	Then $A$ and $\mathcal A$ share the same non-zero singular values. Furthermore, we have
	\begin{align*}
		\norm{A}^2 \le \norm{A_f} \norm{A_g},
	\end{align*}
	where $A_f$ and $A_g$ are the matrices whose entries are given by 
	\begin{align*}
		(A_f)_{st} = \sum_{i=1}^p \<f_{is}, f_{it}\>, 
		\quad
		(A_g)_{st} = \sum_{i=1}^p \<g_{is}, g_{it}\>.
	\end{align*}
\begin{proof}
	Let $F: \Rm^T \to \cH^p$ be the linear operator given by
	\begin{align*}
		Fu =
		\sum_{t=1}^T u_t f_t,\quad u = (u_1, \ldots, u_T)'\in \Rm^T.
	\end{align*}
	For $h = (h_1, \ldots, h_p)\in\cH^p$, since
	\begin{align*}
		\<h, Fu\>
		=
		\sum_{i=1}^p \sum_{t=1}^T\<h_i, u_t f_{it}\>
		=
		\sum_{t=1}^T  u_t \<h, f_t\>,
	\end{align*}
	the adjoint of $F$ is given by $F^*:\cH^p\to \Rm^T$
	\begin{align*}
		F^* h = (\<h, f_1\>, \ldots, \<h, f_T\>)',\quad h\in\cH^p.
	\end{align*}
	Similarly, define
	$G:\cH^p \to \Rm^T$ by
	\begin{align*}
		G u = \sum_{t=1}^T u_t g_t, \quad u \in\Rm^T
	\end{align*}
	whose adjoint is given by
	\begin{align*}
		G^* h = (\<h, g_1\>, \ldots, \<h, g_T\>)',\quad h\in\cH^p.
	\end{align*}
	Let $(e_t)_{t=1,\ldots, T}$ be the standard basis of $\Rm^T$, then we have
	\begin{align*}
		\<e_s, F^* G e_t\>
		=
		\<  F e_s,  G e_t   \> =  \<f_s, g_t\> = A_{st}
	\end{align*}
	the matrix $A$ can be written as $A = F^{*} G$. Similarly, we have $\mathcal A = F G^*$ since
	\begin{align*}
		\<h_1, FG^* h_2\>
		=
		\<F^* h_1, G^* h_2\>
		=
		\sum_{t=1}^T  \<h_1, f_t\> \<h_2, g_t\>  = \<h_1, \mathcal A h_2\>
	\end{align*}
	for $h_1, h_2\in\cH^p$. It remains to observe that $A = F^*G$ shares the same non-zero singular values as $A^* = G^*F$ which shares the same non-zero singular values as $\mathcal A = F G^*$. This gives the first claim.
	The second claim follows from the inequality
	\begin{align*}
		\norm{A}^2 \le \norm{F}^2 \norm{G}^2
		= \norm{FF^*} \norm{GG^*}
	\end{align*}
	and the fact that $A_f = FF^*$ and $A_g = GG^*$. 
\end{proof}
\end{lemma}

\subsection{Proof of the main results}
\begin{proof}[of Theorem \ref{theorem - main}]
	First we show that
	show that under Assumption \ref{assumptions}-\ref{assumptions - zeta}, Theorem \ref{theorem - main} holds for the matrix $\hat S$ whenever it holds for the covariance matrix of the non-stationary part of the model $\td S$, i.e. the perturbation  of the stationary error term $\zeta$ is asymptotically negligible. 

We will show below that the spectral gaps of $\td S$ are asymptotically of size
\begin{align*}
|\lambda_{k}(\td S) - \lambda_{k+1}(\td S)|\sim \norm{\td S}\sim T^2 p^{-1}\langle C_\epsilon \Omega\rangle.\label{equation - gap of tilde S}
\numberthis
\end{align*}
By standard perturbation theory \citep{Kato2013}, to prove  that the perturbation of $\zeta$ is asymptotically negligible, it suffices to show that $\norm{\hat S - \td S} = o_p(T^2 p^{-1}\langle C_\epsilon \Omega\rangle)$. This is what we show now. Since $X(u) := (\Psi F)(u)+ \zeta(u)$, by the triangle inequality we have
\begin{align*}
\norm{    \hat S -  \td S   } \lesssim
\normm{     \frac{1}{p}\int_\cI M\zeta(u)'\zeta(u)M du   } + 
\normm{     \frac{1}{p}\int_\cI M(\Psi F)(u)'\zeta(u)M du   }.
\numberthis\label{equation - perturbation proof S - S}
\end{align*}
where $M$ denotes the orthogonal projection onto the orthogonal complement of the $T$-dimensional vector of ones, i.e. $M = I_T - T^{-1} 1_T 1_T'$. We will use the fact that $\norm{M} = 1$ throughout the rest of the proof without mention. 

	For the first term in \eqref{equation - perturbation proof S - S}, by Lemma \ref{lemma - same spectrum}, we have
	\begin{align*}
		\normm{\frac{1}{p}\int_\cI M\zeta(u)'\zeta(u)M du   }
		\le
		\frac{1}{p}\normm{\sum_{t=1}^T \zeta_t \otimes \zeta_t}
		\le \frac{1}{p}\sum_{t=1}^T \norm{ \zeta_t}^2_{\cH^p}. 
	\end{align*}
	Since $\Em[\norm{\zeta_t}^2]= \<C_\zeta\>$, by Markov's inequality we have
	\begin{align*}
		\normm{     \frac{1}{p}\int_\cI M\zeta(u)'\zeta(u)M du   }
		=
		O_p(Tp^{-1} \< C_\zeta\>)
		=o_p(T^2p^{-1} \<C_\epsilon \Omega\>), 
		\numberthis\label{equation - order of zeta^2}
	\end{align*}
	where the last estimate follows from Assumption \ref{assumptions - zeta}.
	Similarly, by the second claim of Lemma \ref{lemma - same spectrum}, the second term in \eqref{equation - perturbation proof S - S} is bounded by
	\begin{align*}
\normm{     \frac{1}{p}\int_\cI M(\Psi F)(u)'\zeta(u)M du   }
		& \le
		\normm{ \frac{1}{p}\int_\cI M(\Psi F)(u)'(\Psi F)(u)M du   }^{1/2}
		\normm{     \frac{1}{p}\int_\cI M \zeta(u)'\zeta(u) Mdu   }^{1/2}
		  \\
		& =
		\norm{ \td S  }^{1/2}
		\normm{     \frac{1}{p}\int_\cI M\zeta(u)'\zeta(u) Mdu   }^{1/2}
		  \\
		& =o_p( T^2 p^{-1} \<C_\epsilon \Omega\>  ),
	\end{align*}
	where the last estimate follows from \eqref{equation - order of zeta^2} and the bound
	\begin{align*}
		\norm{\td S} = O_p(T^2 p^{-1} \< C_\epsilon \Omega\>)
		\numberthis\label{equation - order of tilde S}
	\end{align*}
	which will be shown below. 
	From this we conclude that $\norm{\td S - \hat S} = o_p(T^2 p^{-1}\<C_\epsilon \Omega\>)$
	which shows that the perturbation of $\zeta$ is asymptotically negligible. 

	It remains to show that Theorem \ref{theorem - main} holds for the covariance operator $\td S$, which also proves the two estimated  \eqref{equation - gap of tilde S} and \eqref{equation - order of tilde S} we used above. 
	This proof is adapted from the arguments in Section A.3.1 of \cite{OnatskiWang2021}. We sketch the key components of the argument adapted to the functional setting of our model.

	Let $\td \lambda_n, \td u_n, n=1,\ldots, T$ be the eigenvalues of $\td S$ arranged in non-ascending order and the corresponding eigenvectors. 
	We will focus on the largest eigenvalues $\tilde \lambda_1$ and the corresponding eigenvector $\tilde u_1$; the general case can be obtained by mathematical induction. We omit the details of the induction argument, since it can be found on \citet[][p.610-611]{OnatskiWang2021}.

	Recall from Lemma \ref{lemma - spectral decomp of MTheta} that the set $\{w_n, n=1, \ldots, T\}$ forms an orthonormal basis of $\Rm^T$ and $w_T$ is proportional to a vector of ones. Since $\tilde u_1$ is orthogonal to $w_T$ by construction, we have the representation
	\begin{align*}
		\tilde u_1 = \sum_{t=1}^{T-1} \alpha_t w_t
	\end{align*}
	for some set of weights $(\alpha_t)$. Note that since $(w_t)$ is an orthonormal set, we have
	\begin{align*}
		1 = \norm{\td u_1}^2 = \norm{\sum_{t=1}^{T-1} \alpha_t w_t}
		= \sum_{t=1}^{T-1} \alpha_t^2 \norm{w_t}^2 = \sum_{t=1}^{T-1} \alpha_t^2.
	\end{align*}
	By the definitions of $\td \lambda_1, \td u_1, \td S$ we have
	\begin{align*}
		\td \lambda_1 & =  \td u_1' \td S \td u_1
		=
		\sum_{s,t=1}^{T-1} \alpha_s \alpha_t w_s' \td S w_t. 
	\end{align*}
	Since $\td S = p^{-1}M \Theta' W \Theta M$ where $M \Theta' = \sum_t \sigma_t w_t v_t'$ by Lemma \ref{lemma - spectral decomp of MTheta}, we have
	\begin{align*}
		p\td \lambda_1
		=
		\sum_{s,t=1}^{T-1} \alpha_s  \alpha_t \td \sigma_s\td \sigma_t  v_s' W v_t.
	\end{align*}
The quadratic form $v_s'W v_t$ can be estimated using Lemma~\ref{lemma - concentration of a'Wb}, however, the sum of $T-1$ such terms has to be dealt with additional care since $T\to\infty$. Using a truncation argument similar to the one found in \citet[][p.607-p.609]{OnatskiWang2021}, we may show that 
\begin{align*}
p \hat \lambda_1 \le \langle C_\epsilon \Omega \rangle \left(\sum_{t=1}^{T-1} \alpha_t^2 \td \sigma_t^2 + o_p(T^2)\right).
\end{align*}
Since $\td \sigma_1$ is the leading eigenvalue, and $\sum_{t=1}^{T-1} \alpha_t^2 =1$, we have
\begin{align*}
\sum_{t=1}^{T-1} \alpha_t^2 \td \sigma_t^2 = \alpha_1^2 \td \sigma_1^2+
\sum_{t=2}^{T-1} \alpha_t^2 \td \sigma_t^2 \le \alpha_1^2 \td \sigma_1^2+
(1- \alpha_1^2) \td \sigma_2^2, 
\end{align*}
which gives
\begin{align*}
p \hat \lambda_1 \le \alpha_1^2 \td \sigma_1^2 \langle C_\epsilon \Omega \rangle + (1- \alpha_1^2) \td \sigma_2^2\langle C_\epsilon \Omega \rangle 
+ o_p(T^2) \langle C_\epsilon \Omega \rangle . 
\end{align*}
On the other hand, since $\td \lambda_1$ is the leading eigenvalue, we have
\begin{align*}
p \td \lambda_1 \ge p w_1' \td \Sigma w_1 = \sigma_1 v_1' W v_1 =
\sigma_1^2 \langle C_\epsilon \Omega \rangle + o_p(T^2) \langle C_\epsilon \Omega \rangle. 
\end{align*}
Combining the upper and lower bounds and simplifying, we obtain
\begin{align*}
(\sigma_1^2 - \sigma_2^2)(1- \alpha_1^2) \langle C_\epsilon \Omega \rangle
\le o_p(T^2) \langle C_\epsilon \Omega \rangle.
\end{align*}
Using the fact that $\sigma_1^2- \sigma_2^2\sim T^2$, we obtain $1- \alpha_1^2 = o_p(1)$, which implies $\alpha_1^2 = (\td u_1'w_1)^2 \to 1$ in probability as $T\to\infty$. This gives the first claim.

Using the upper and lower bounds again, we may obtain
\begin{align*}
|p\td \lambda_1 - \sigma^2 \langle C_\epsilon \Omega \rangle|
\le |1- \alpha_1^2|(\sigma_1^2+ \sigma_2^2) \langle C_\epsilon\Omega \rangle + o_p(T^2)\langle C_\epsilon\Omega \rangle = o_p(T^2)\langle C_\epsilon\Omega \rangle,
\end{align*}
since $\sigma_1^2\sim T^2/\pi^2$ and $1-\alpha_1^2 = o_p(1)$. This gives
\begin{align*}
\tilde \lambda_1 = \frac{T^2 }{\pi^2 p} \langle C_\epsilon \Omega \rangle(1+ o_p(1)),  
\end{align*}
which is the asymptotic limit of the largest empirical eigenvalue. This also implies the bounds~\eqref{equation - gap of tilde S} and~\eqref{equation - order of tilde S}. Finally, the last claim follows from the same arguments as p.609-610 of \cite{OnatskiWang2021}, and we omit the details. 
\end{proof}

\begin{proof}[of Theorem~\ref{theorem - upper lower bounds}]
	Using \eqref{eq- def of Omega} and \eqref{equation - Psi_ik special case in theorem 2}, the operator $\Omega$ can be written as
	\begin{align*}
		\Omega_{kl}(v,w) & = \sum_{i=1}^p \int_\cI  \Psi_{ik}(u,v)\Psi_{il}(u,w) du
		  \\
		& =
		\sum_{i=1}^p 
		\sum_{n=1}^\infty 
		a_{nik} a_{nil}
		\phi_n(v)
		\phi_n(w)
		  \\
		& = \sum_{n=1}^\infty (A_n' A_n)_{kl} \phi_n(v) \phi_n(w).
	\end{align*}
	Since $ C_\epsilon^{1/2} = \sum_n c_n^{1/2} \phi_n\otimes \phi_n$, we have
	\begin{align*}
		(C_\epsilon^{1/2} \Omega_{kl} C_\epsilon^{1/2})(u,v)
		& =  \int_{\cI} C_\epsilon^{1/2}(u,\alpha) \Omega_{kl}(\alpha,\beta) C_\epsilon^{1/2}(\beta,v) d \alpha d \beta
		  \\
		&
		=\sum_{n=1}^\infty c_n  (A_n' A_n)_{kl} \phi_n(u) \phi_n(v).
	\end{align*}
	By Proposition \ref{proposition - matrix of op}, the trace of $\cC$ is given by
	\begin{align*}
		\langle \cC \rangle
		& = \sum_{k=1}^K \<  C_\epsilon \Omega_{kk} \>
		=
		\sum_{k=1}^K \sum_{n=1}^\infty c_n (A_n' A_n)_{kk} 
		=
		\sum_{n=1}^\infty c_n \<B_n\>. 
	\end{align*}
	Similarly, the Hilbert-Schmidt norm of $\cC$ is given by
	\begin{align*}
		\norm{\cC}_2^2
		=
		\sum_{kl} \norm{C_\epsilon^{1/2} \Omega_{kl} C_\epsilon^{1/2} }^2
		=
		\sum_{n=1}^\infty \sum_{kl} c_n^2 (A_n'A_n)_{kl}^2
		=
		\sum_{n=1}^\infty  c_n^2 \norm{B_n}^2_2.
	\end{align*}
	Since $1= \<C_\epsilon\> = \sum_n c_n$ and $\norm{C_\epsilon}_2^2 = \sum_n c_n^2$, we have the following upperbound
	\begin{align*}
		\frac{\langle \cC \rangle^2 }{\norm{\cC}_2^2}
		& =  \limsup_{q\to\infty}
		\frac{(\sum_{n=1}^q c_n \<B_n\>)^2}{\sum_{n=1}^q c_n^2 \norm{B_n}_2^2}
		  \\
		& \le 
		\inf_{Q\in\mathbb N} \sup_{q>Q}
		\frac{\max_{n\le q} \<B_n\>^2 }{\min_{n\le q} \norm{B_n}_2^2}
		\frac{(\sum_{n=1}^q c_n) ^2}{  {\sum_{n=1}^q c_n^2} }
		  \\
		& \le
		\frac{1}{\norm{C_\epsilon}_2^2}
		\inf_{Q\in\mathbb N} \sup_{q>Q}
		\frac{\max_{n\le q} \<B_n\>^2 }{\min_{n\le q} \norm{B_n}_2^2}.
	\end{align*}
	Let $k_q$ be the index such that the maxinum  in the numerator is attained, i.e.
	\begin{align*}
		k_{q}:= \mathrm{arg\ max}_{n\le q} \<B_n\>.
	\end{align*}
	Since $  \norm{B_n}_2 \gtrsim \norm{B_{k_q}}_2  $ uniformly in $n$ and $q$ by assumption, we have
	\begin{align*}
		\frac{\max_{n\le q} \<B_n\> }{\min_{n\le q} \norm{B_n}_2}
		\lesssim
		\frac{ \<B_{k_q}\> }{  \norm{B_{k_q}}_2  } \le 
		\max_{n\le q}
		\frac{ \<B_n\> }{\norm{B_n}_2},
	\end{align*}
	which gives the first upper bound
	\begin{align*}
		\frac{\langle \cC \rangle }{\norm{\cC}_2}
		\lesssim
		\frac{1}{\norm{C_\epsilon}_2}
		\inf_{Q\in\mathbb N} \sup_{q>Q}
		\max_{n\le q} \frac{ \<B_n\> }{\norm{B_n}_2}
		\le 
		\frac{1}{\norm{C_\epsilon}_2}
		\sup_{n}\frac{ \<B_n\> }{\norm{B_n}_2}.
	\end{align*}
	By the Cauchy-Schwarz inequality we have
	$\<B_n\> \le \norm{B_n}_0^{1/2} \norm{B_n}_2 $ and the second upper bound follows from the first upper bound.

	Similar to the above computations, using $\norm{B_n}_2\sim \norm{B_m}_2$ we obtain a lower bound of the form
	\begin{align*}
		\frac{\langle \cC \rangle }{\norm{\cC}_2}
		& \ge
		\frac{1}{\norm{C_\epsilon}_2}
		\sup_{Q\in\mathbb N} \inf_{q>Q}
		\frac{\min_{n\le q} \<B_n\> }{\max_{n\le q} \norm{B_n}_2}
		\gtrsim
		\frac{1}{\norm{C_\epsilon}_2}
		\sup_{Q\in\mathbb N} \inf_{q>Q}
		\min_{n\le q}
		\frac{\<B_n\> }{\norm{B_n}_2}.
	\end{align*}
	Since $\<B_n\>\ge \norm{B_n}_2$, we obtain the first half of the first lower bound
	\begin{align*}
		\frac{\langle \cC \rangle }{\norm{\cC}_2}
		\gtrsim \frac{1}{\norm{C_\epsilon}_2}.
	\end{align*}
	For the second  half, 
	using the non-negativity of the summand, we have
	\begin{align*}
		\frac{(\sum_{n=1}^q c_n \<B_n\>)^2}{\sum_{n=1}^q c_n^2 \norm{B_n}_2^2}
		&\ge
		\frac{\sum_{n=1}^q c_n^2 \<B_n\>^2}{\sum_{n=1}^q c_n^2 \norm{B_n}_2^2}
		\\
		& =
		\frac{c_m^2 \<B_m\>^2 + \sum_{n\ne m}^q c_n^2 \<B_n\>^2}{\sum_{n=1}^q c_n^2 \norm{B_n}_2^2},
	\end{align*}
	where $m\le q$ is arbitrarily chosen. 
	Since $\norm{B_n}_2\sim \norm{B_m}_2$ uniformly for all $n$, we have
	\begin{align*}
		\frac{(\sum_{n=1}^q c_n \<B_n\>)^2}{\sum_{n=1}^q c_n^2 \norm{B_n}_2^2}
		\gtrsim 
		\frac{c_m^2 \<B_m\>^2 + \sum_{n\ne m}^q c_n^2 \<B_n\>^2}{\norm{C_\epsilon}^2_2 \norm{B_m}_2^2}
		\ge 
		\frac{1}{\norm{C_\epsilon}^2_2 } c_m^2\frac{\<B_m\>^2}{\norm{B_m}_2^2}.
	\end{align*}
	Taking the limit in $q$ yields the first lower bound. 
	Finally, using Lemma \ref{lemma - ratio l1/l2}, we get
	\begin{align*}
		\frac{\<B_m\>}{\norm{B_m}_2}
		\gtrsim \norm{B_m}_0^{1/2} \alpha(B_m)^{1/2}
	\end{align*} 
	which completes the proof. 
	\end{proof}

\section{A conjecture on the asymptotics of sample eigenvectors}\label{section - conjecture}

Lastly, we end the paper with some interesting observations and a conjecture, which forms the basis of some future research. We observe that for the case $K=2$ in Setting V (Figure \ref{figure - very_local iid}), the fifth sample eigenvector seemingly behaves in a significantly different way to the first four. More specifically, the first four sample eigenvectors vagues resembles the spurious limit and are relatively smooth,  while the fifth eigenvector resembles a white noise process which has rough sample paths. 

The model in this setting is generated by two basis functions with equal weight, and along the direction of each basis function, the non-stationary factor is loaded by a matrix $A_n$ of rank two. Heuristically, the non-stataionary component of the model spans a subspace of dimension four, while the orthogonal complement of this subspace is spanned by the stationary components of the model. This is similar to the concept of the attractor and cointegrating space that has recently garnered much attention in the functional time series literature \citep[see][]{ChangKimPark2016, BeareSeoSeo2017, LiRobinsonShang2023}. In this line of work, the dimension of the non-stationary space is determined by model parameters, and the asymptotic behaviour of the eigenstructure differs on the stationary and non-stationary subspaces. 

It is therefore natural to conjecture that for models with $C_\epsilon$ and $A_n$ being of finite rank, only the first $\norm{C_\epsilon}_0\norm{A_n}_0$ eigenvectors will tend to the spurious limit, while the rest behave like white noise. This is because the limits in Theorem \ref{theorem - main} arise as eigenfunctions of Wiener processes, which are the scaling limits of the non-stationary part of the model.
This is confirmed empirically by simulations under a wide range of different settings, we include a simple example to illustrate what typically occurs. 
\begin{figure}[!htb]
    \centering
    \subf{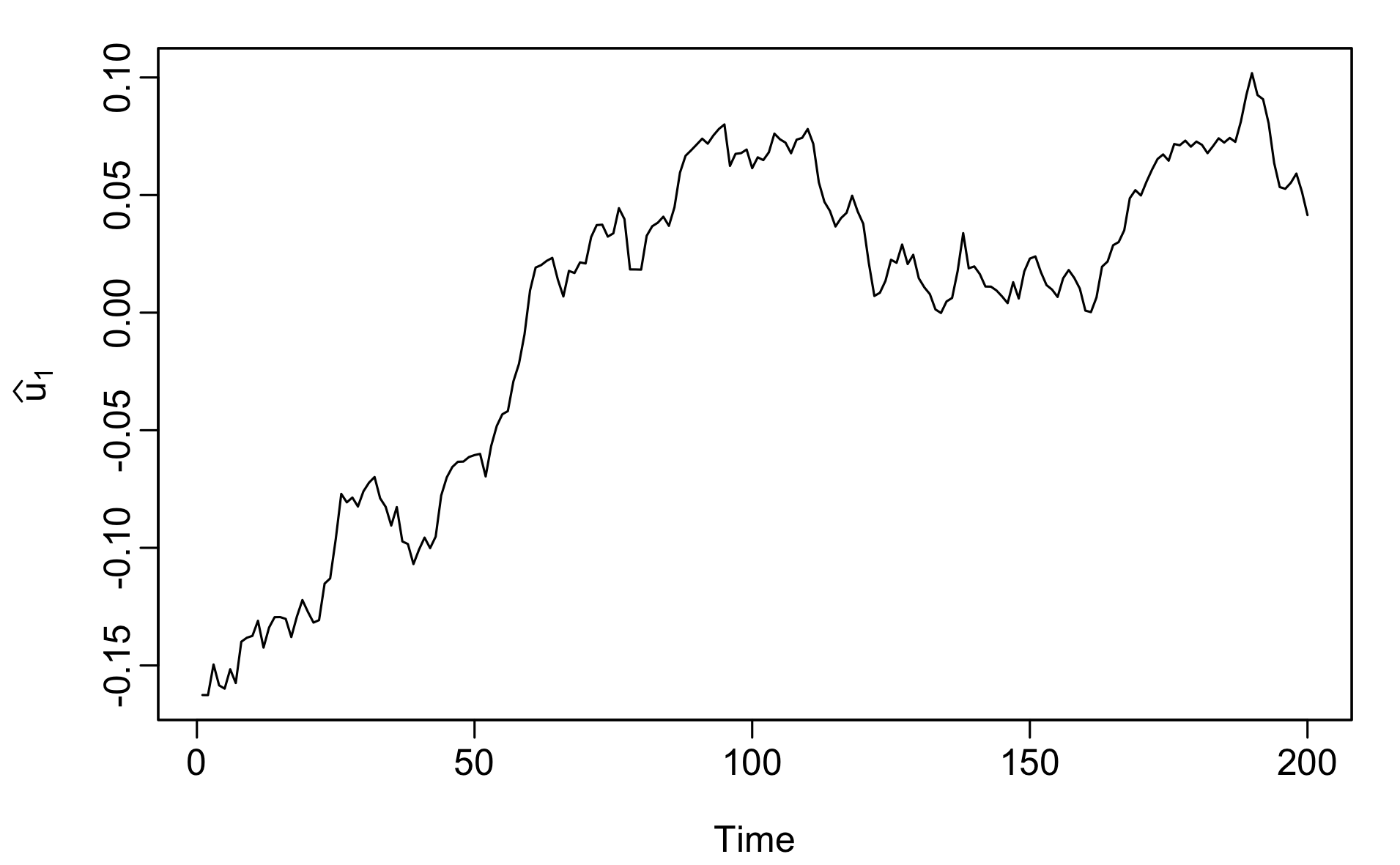}{.23}
    \subf{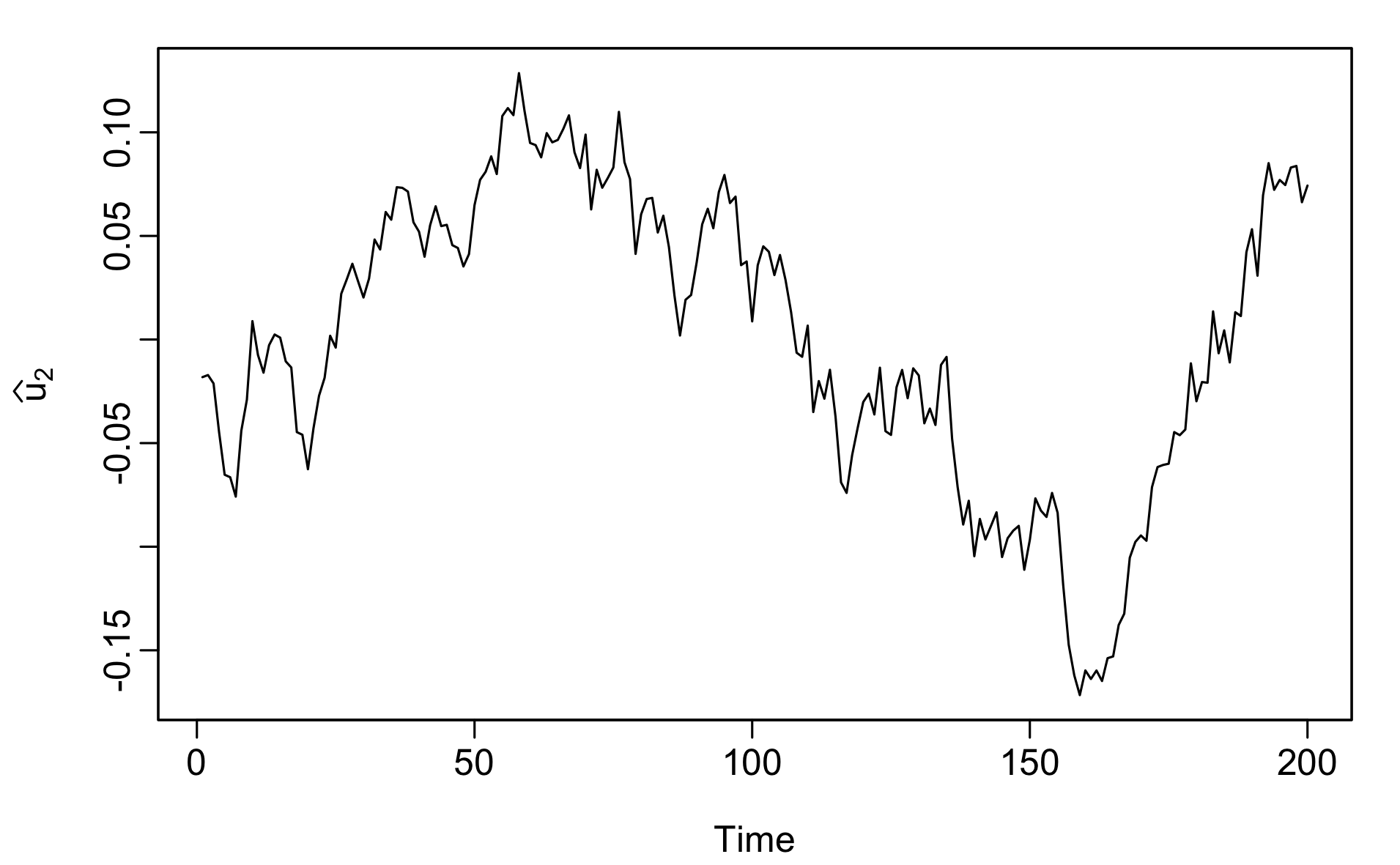}{.23}
    \subf{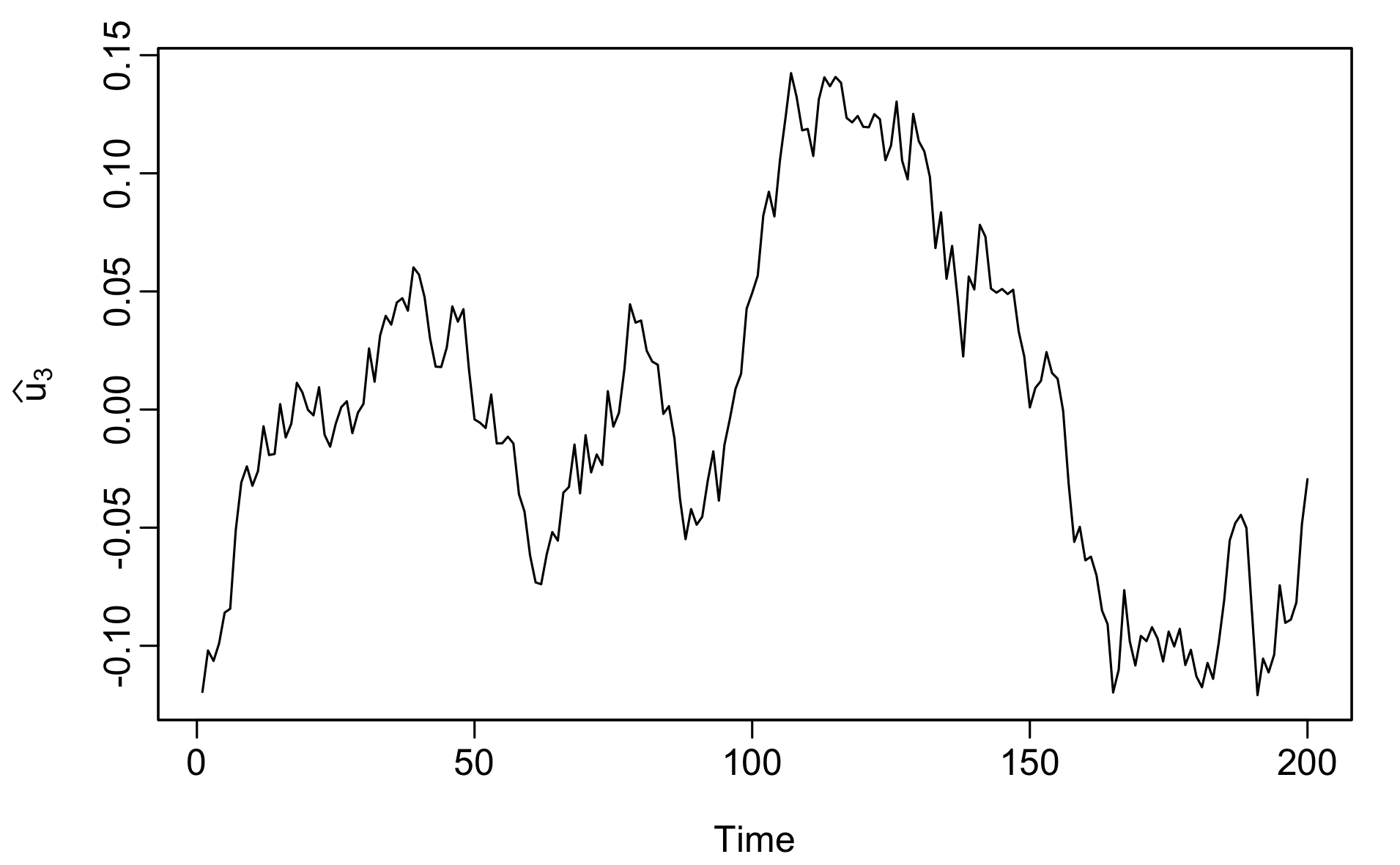}{.23}
    \subf{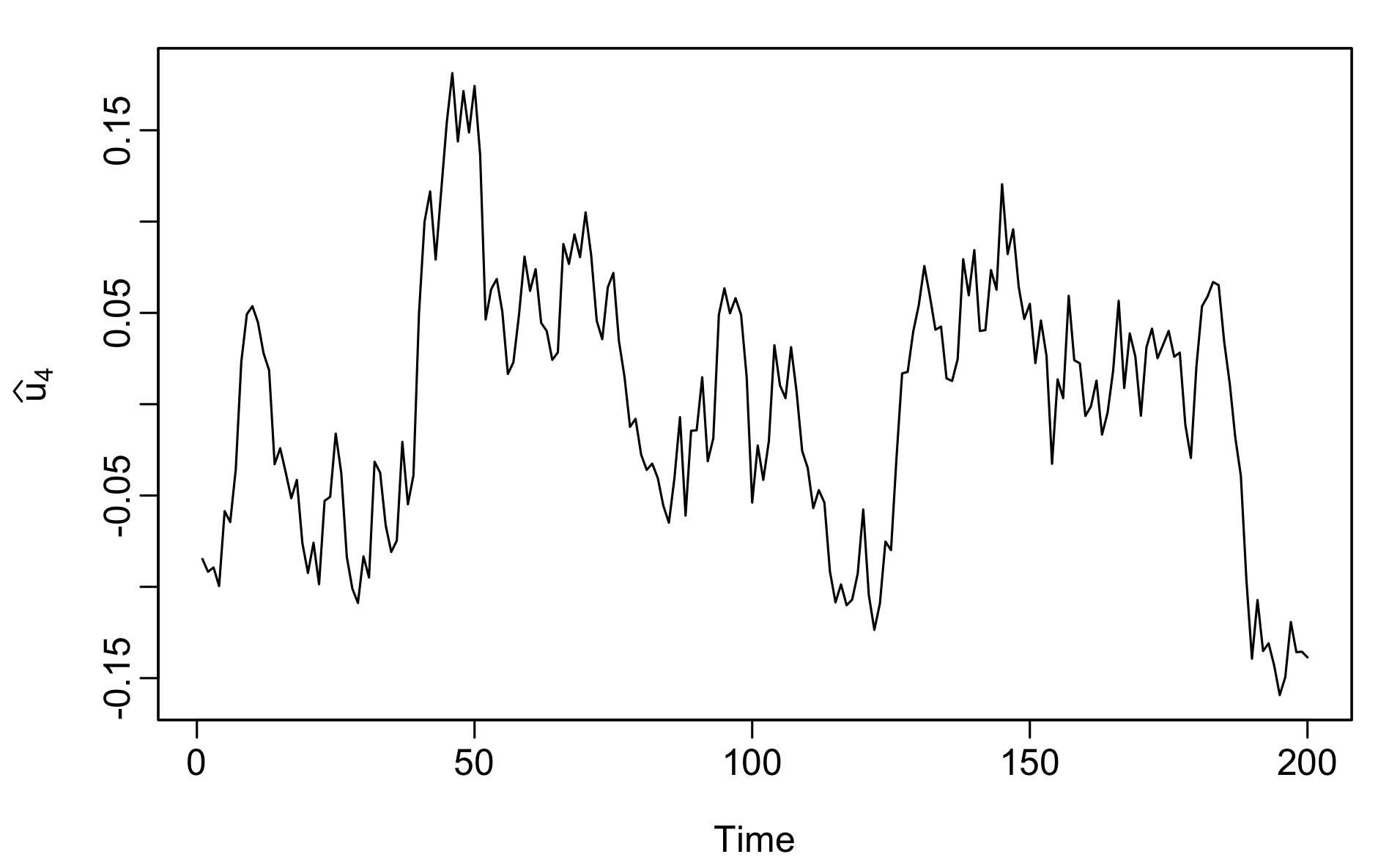}{.23}
    \subf{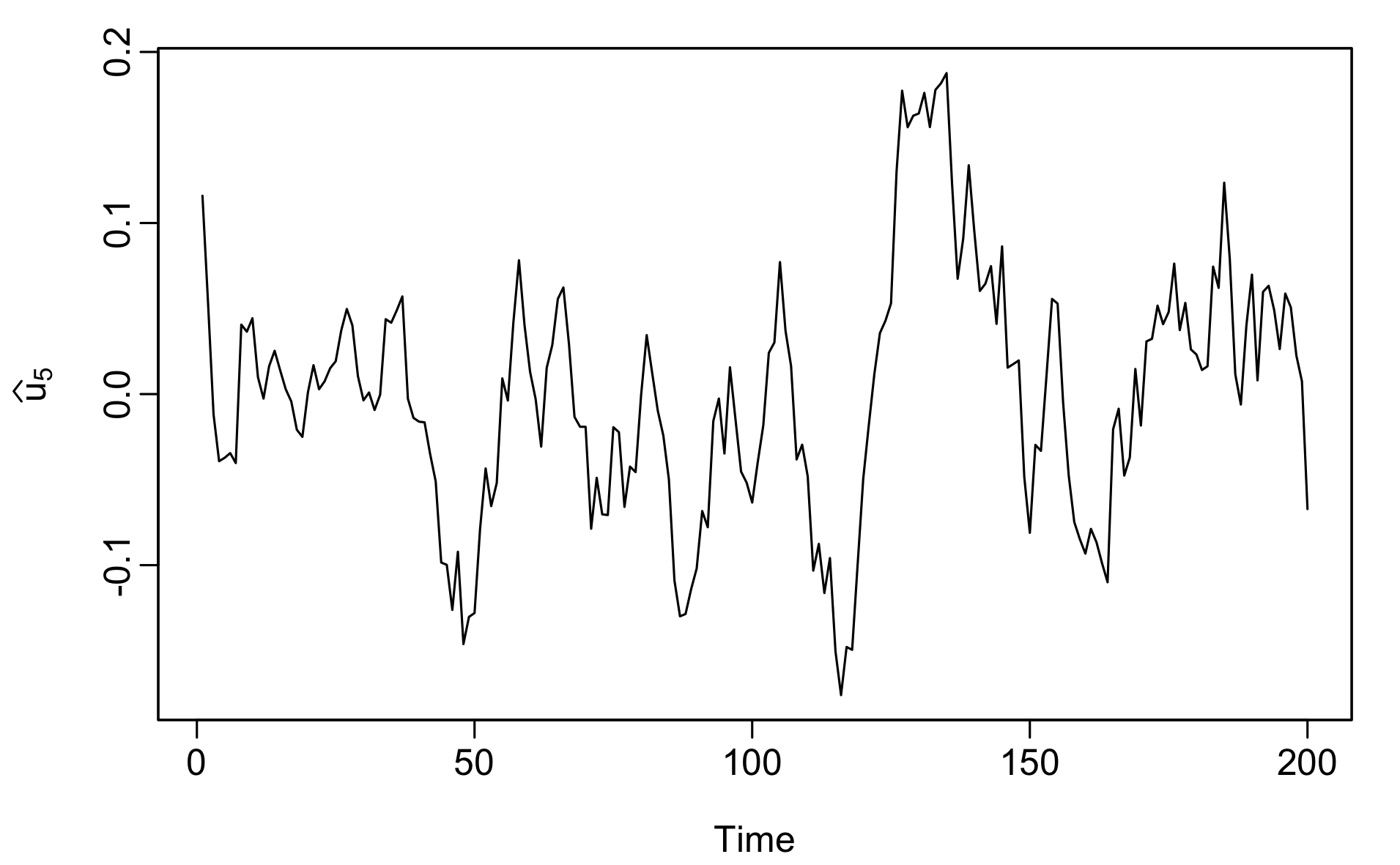}{.23}
    \subf{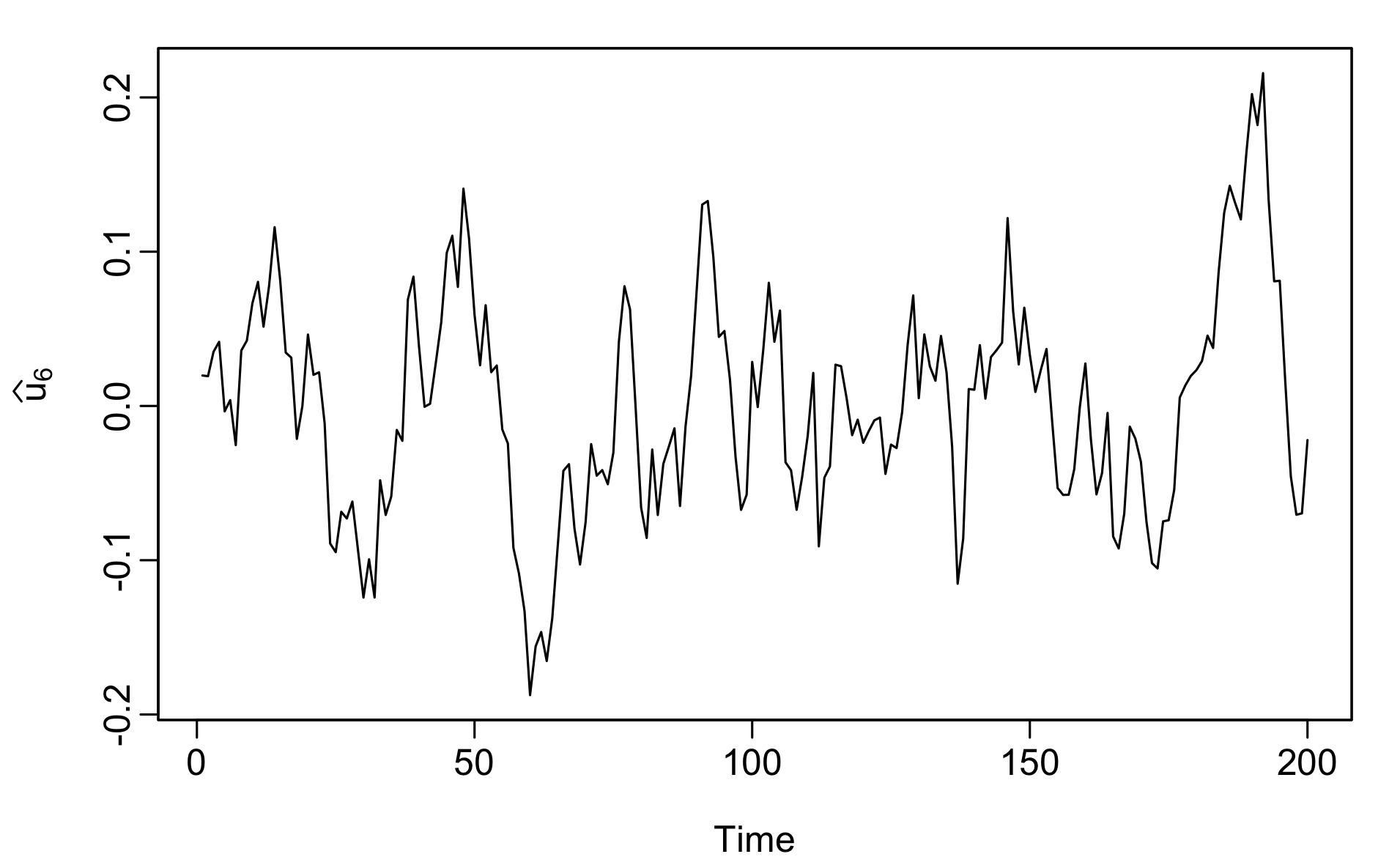}{.23}
    \subf{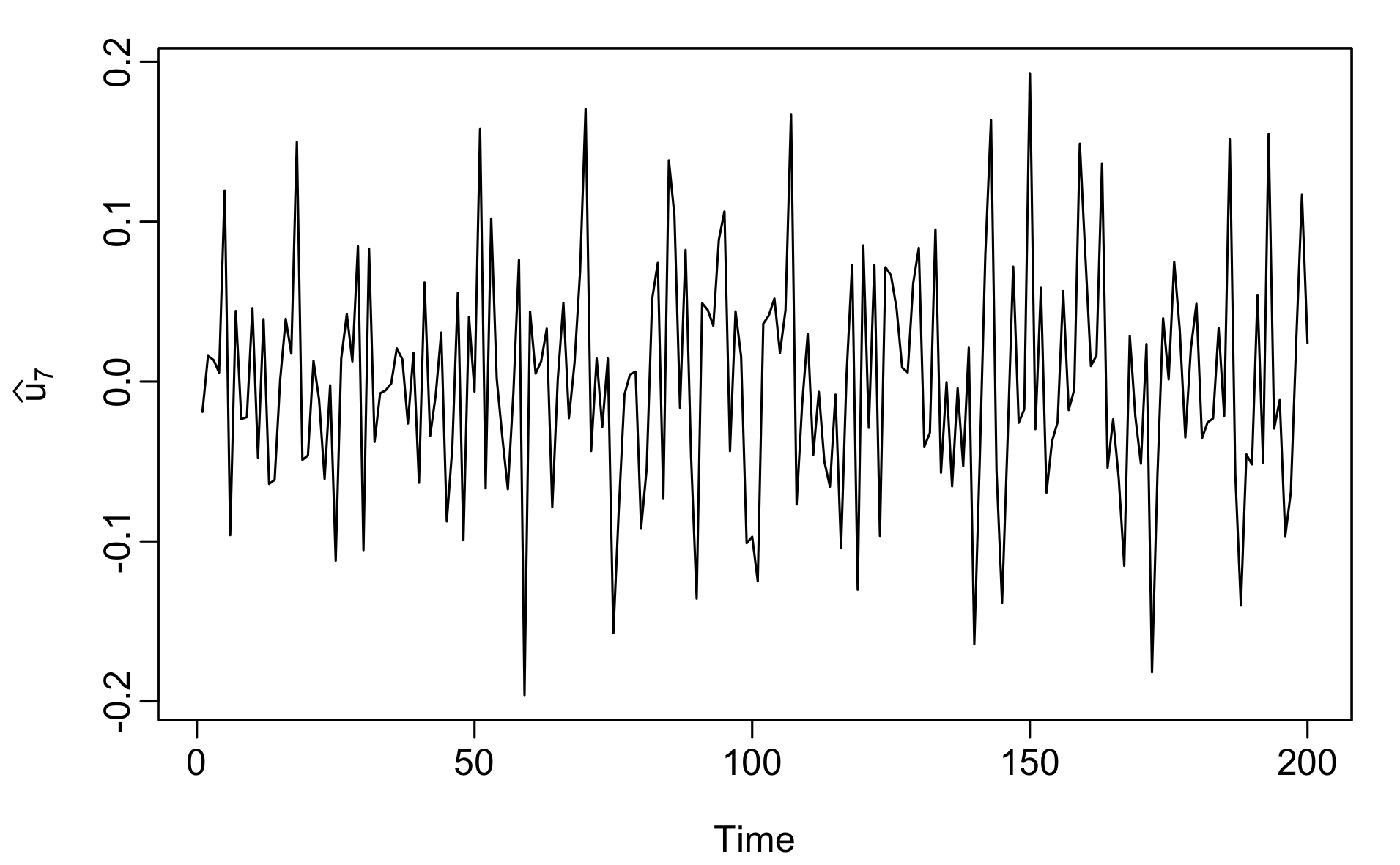}{.23}
    \subf{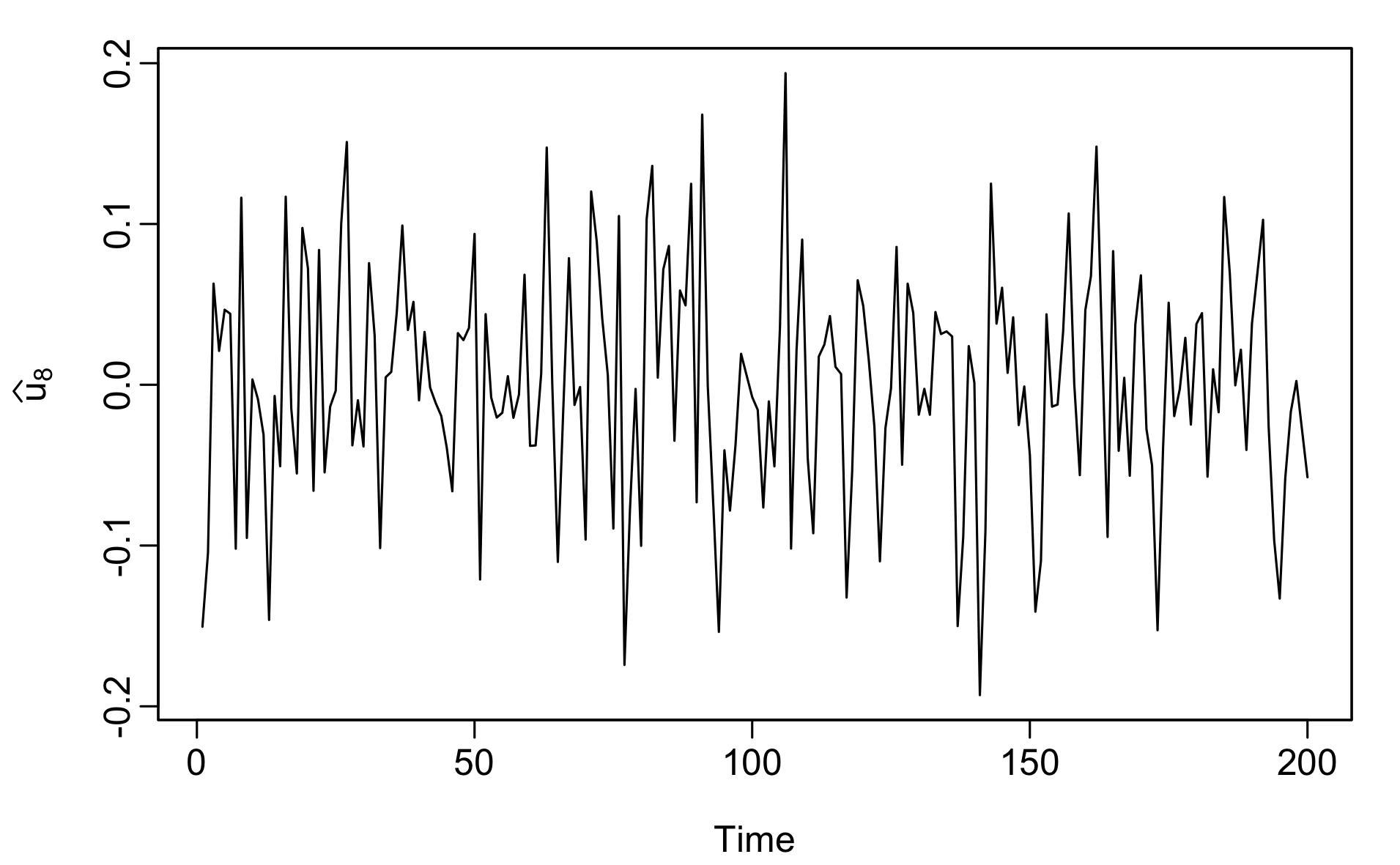}{.23}
    \subf{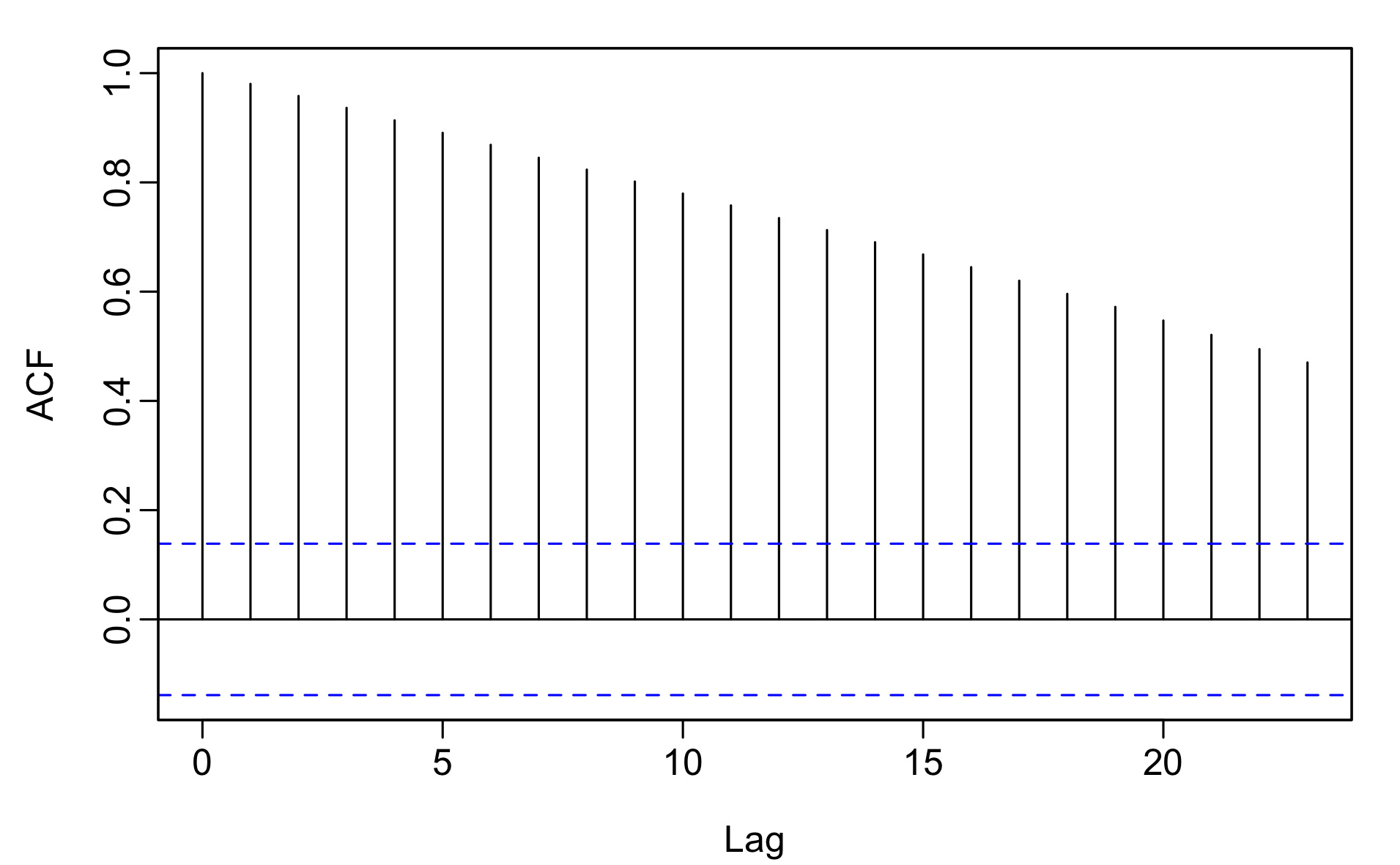}{.23}
    \subf{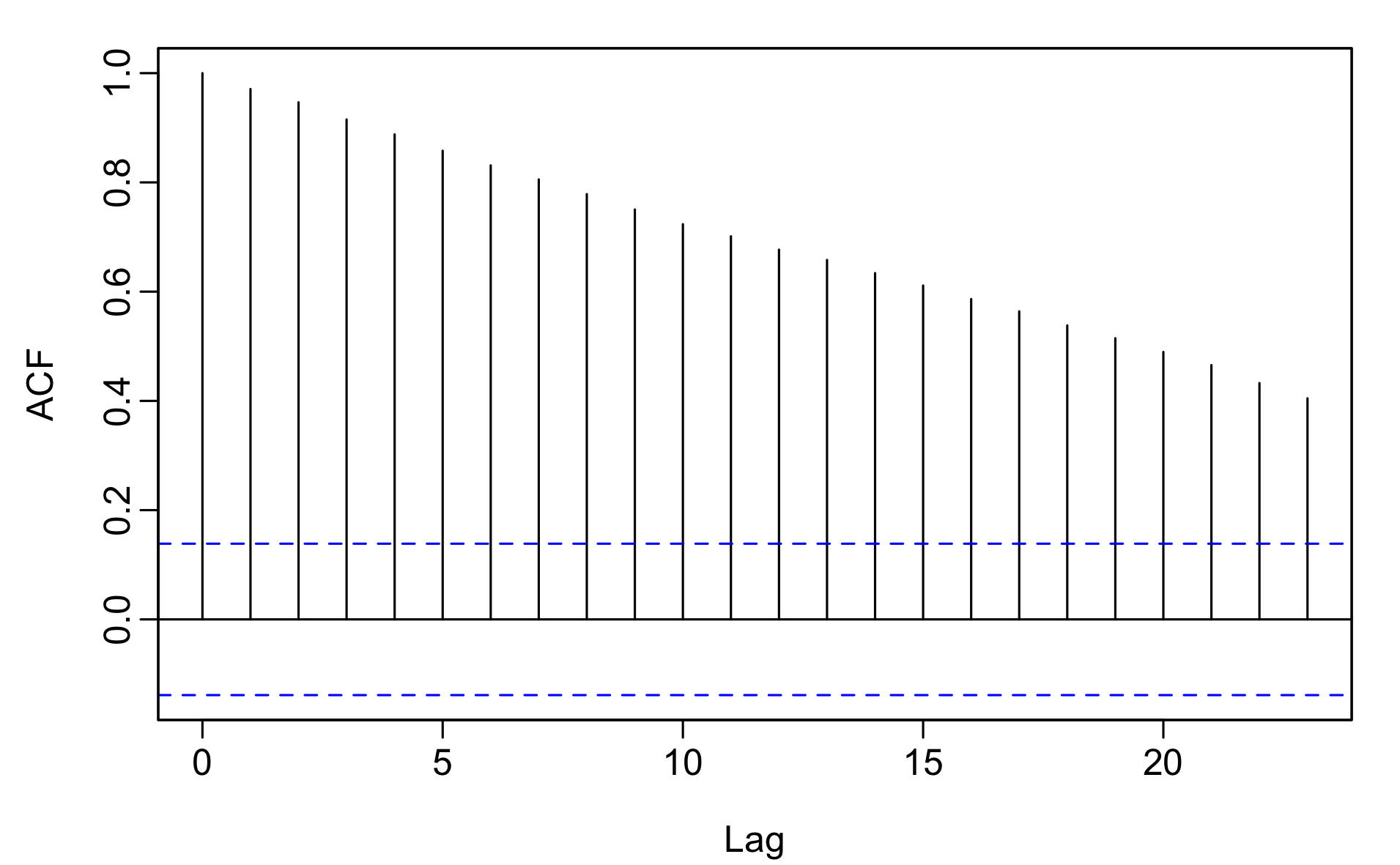}{.23}
    \subf{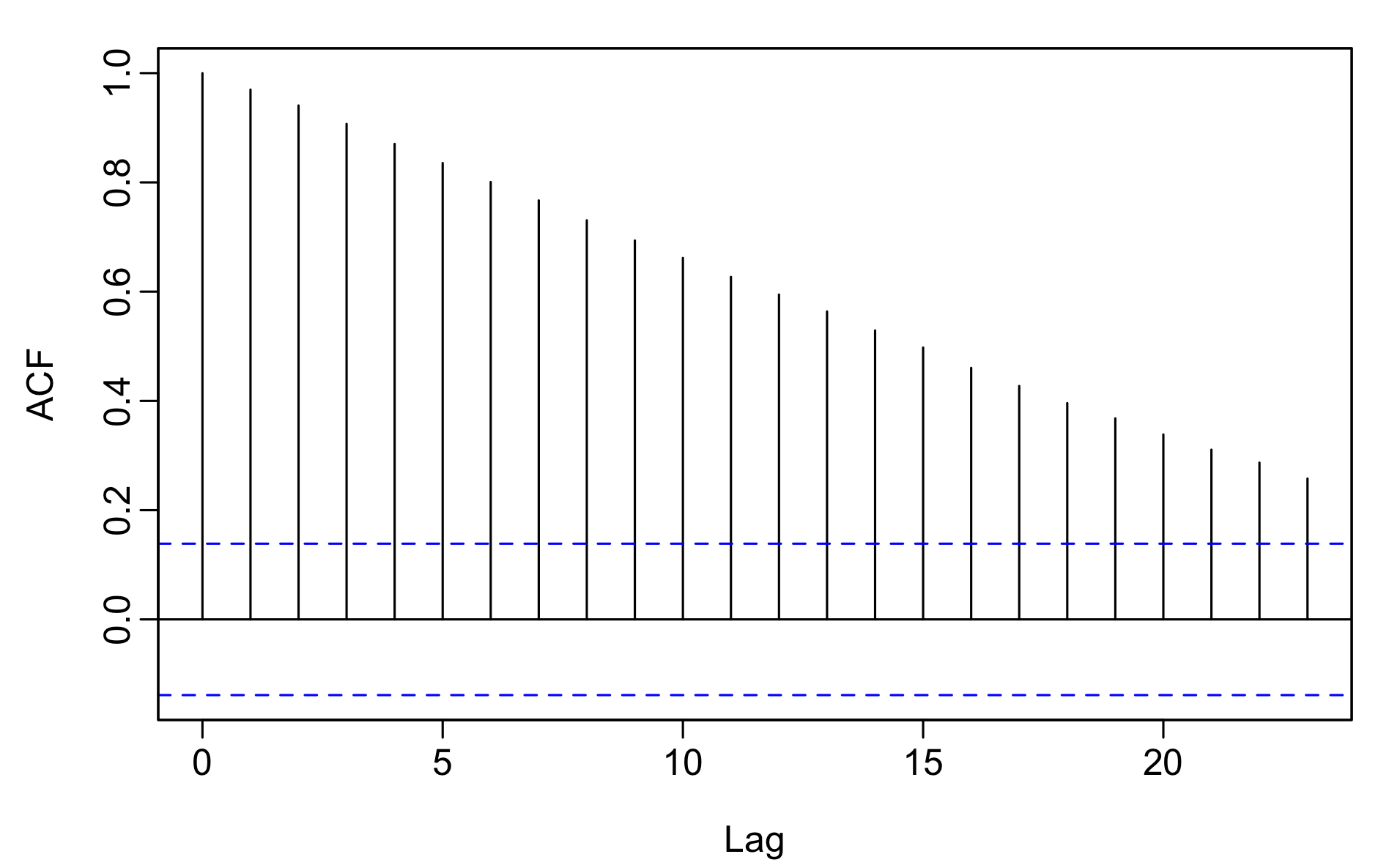}{.23}
    \subf{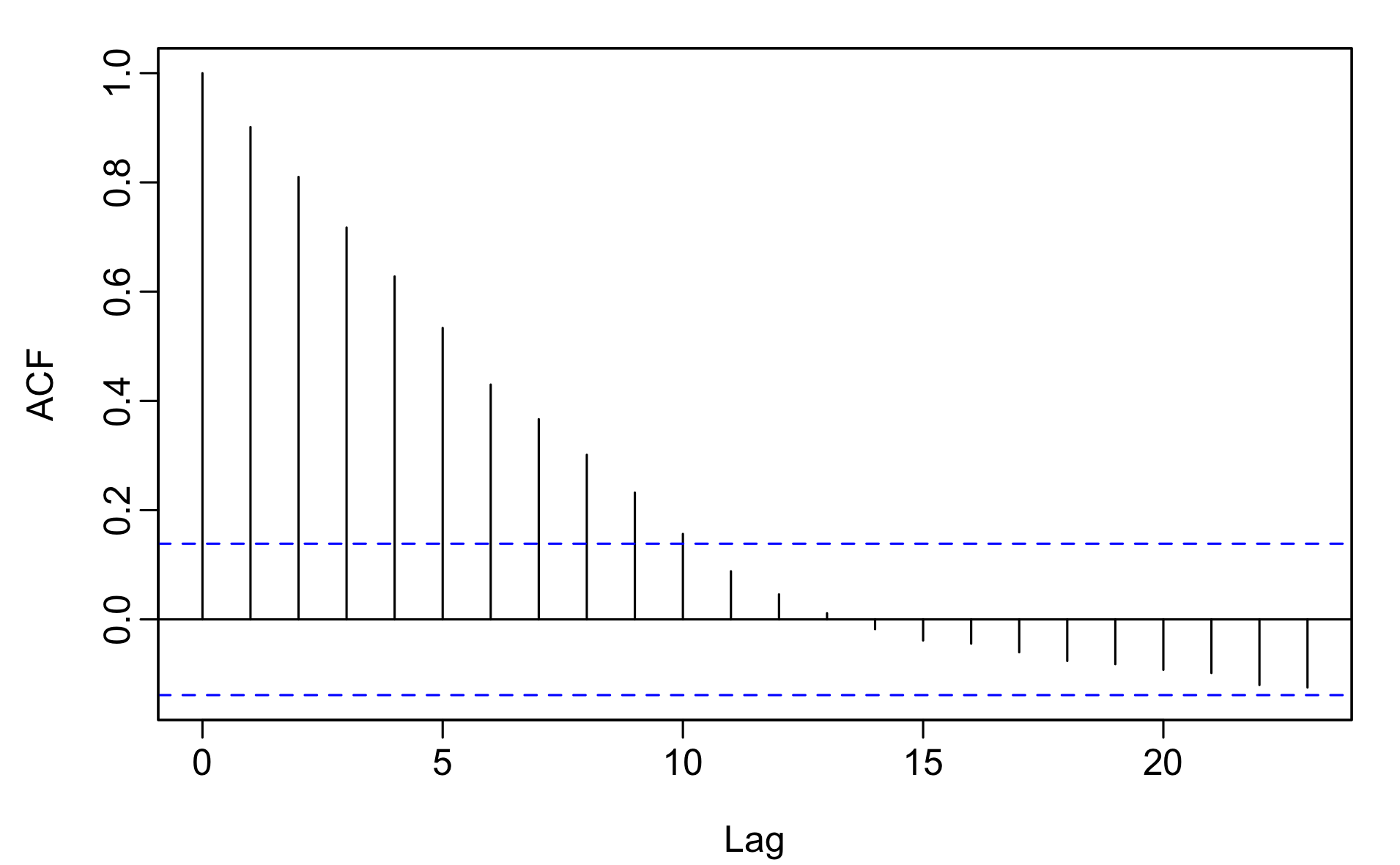}{.23}
    \subf{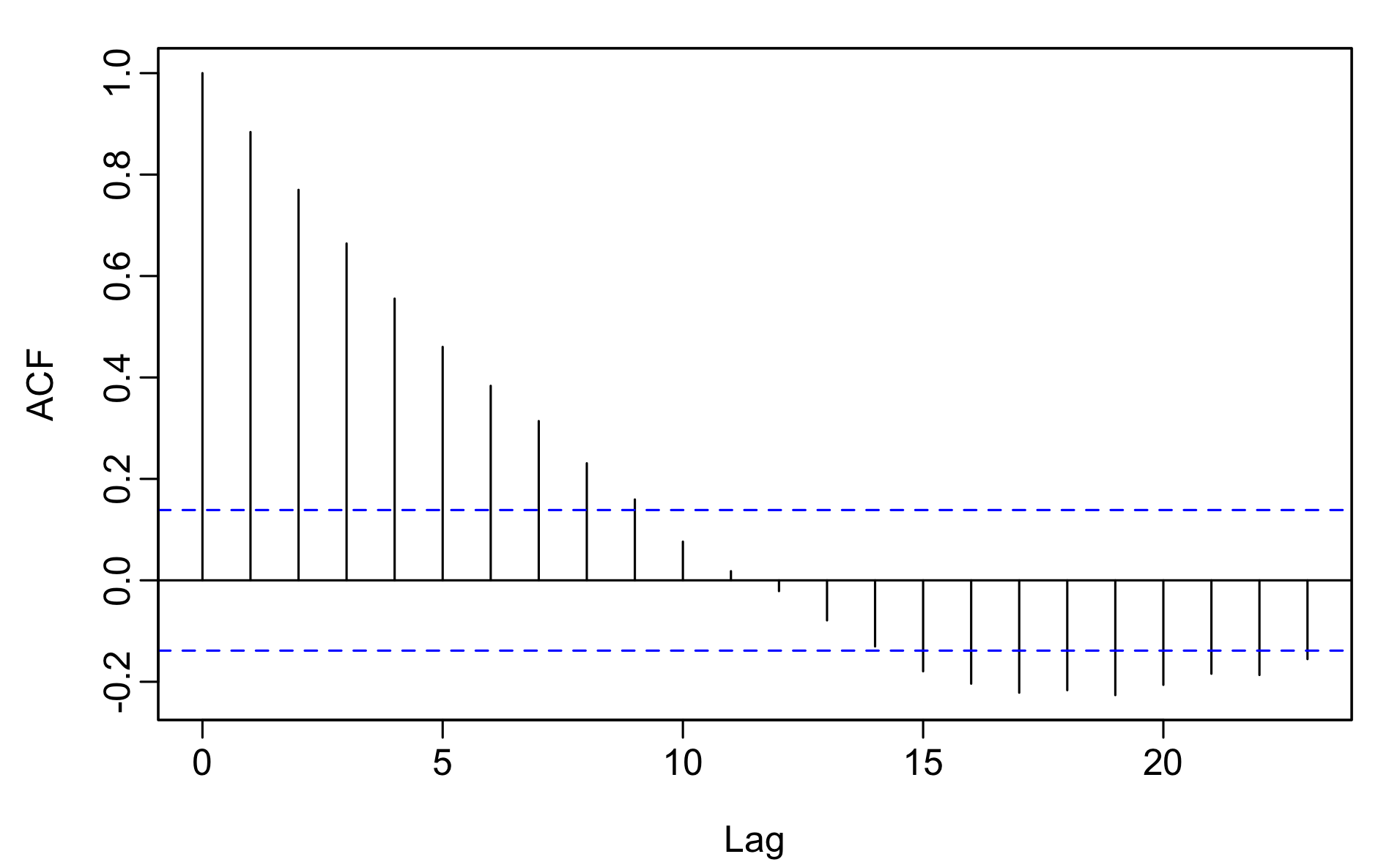}{.23}
    \subf{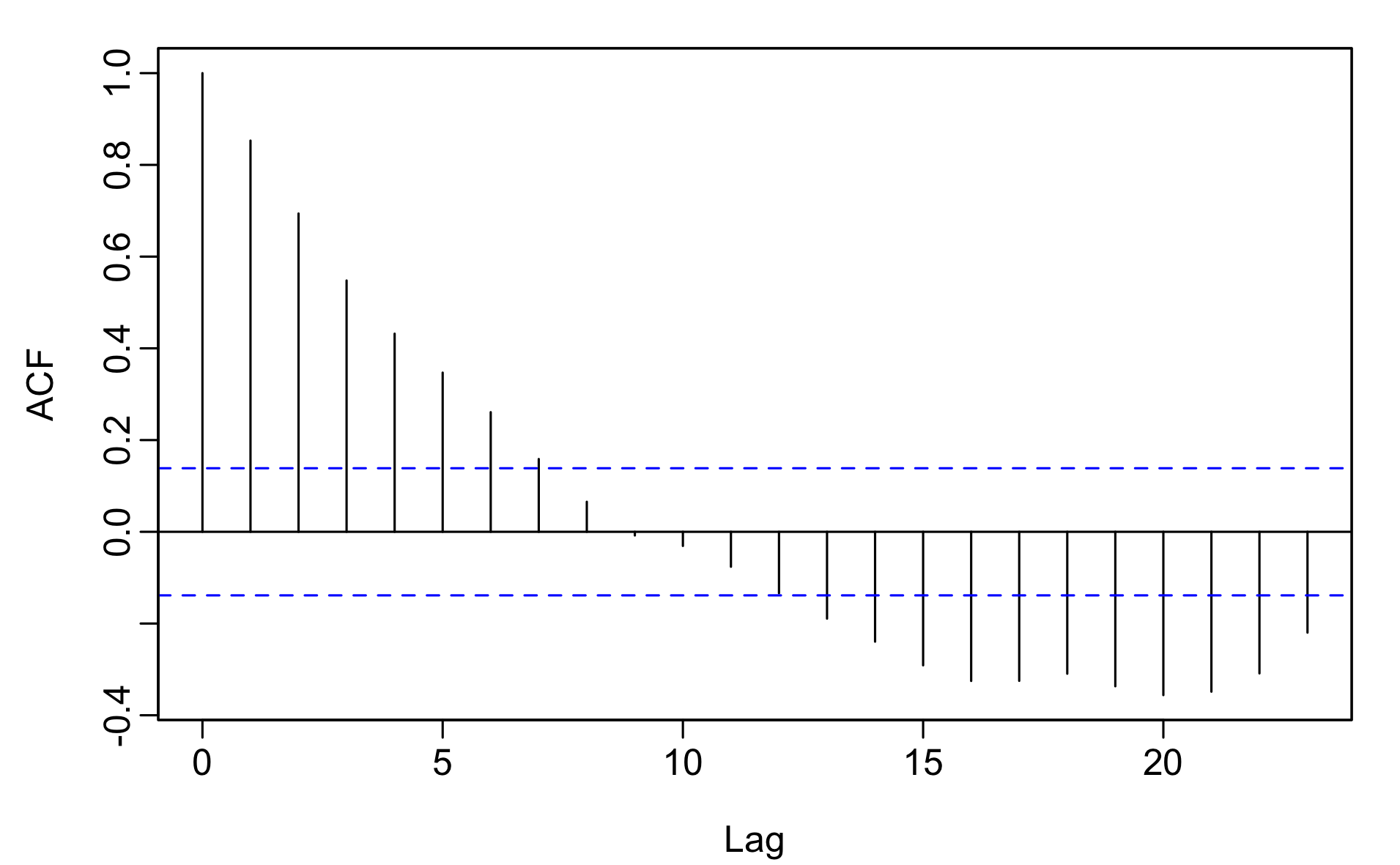}{.23}
    \subf{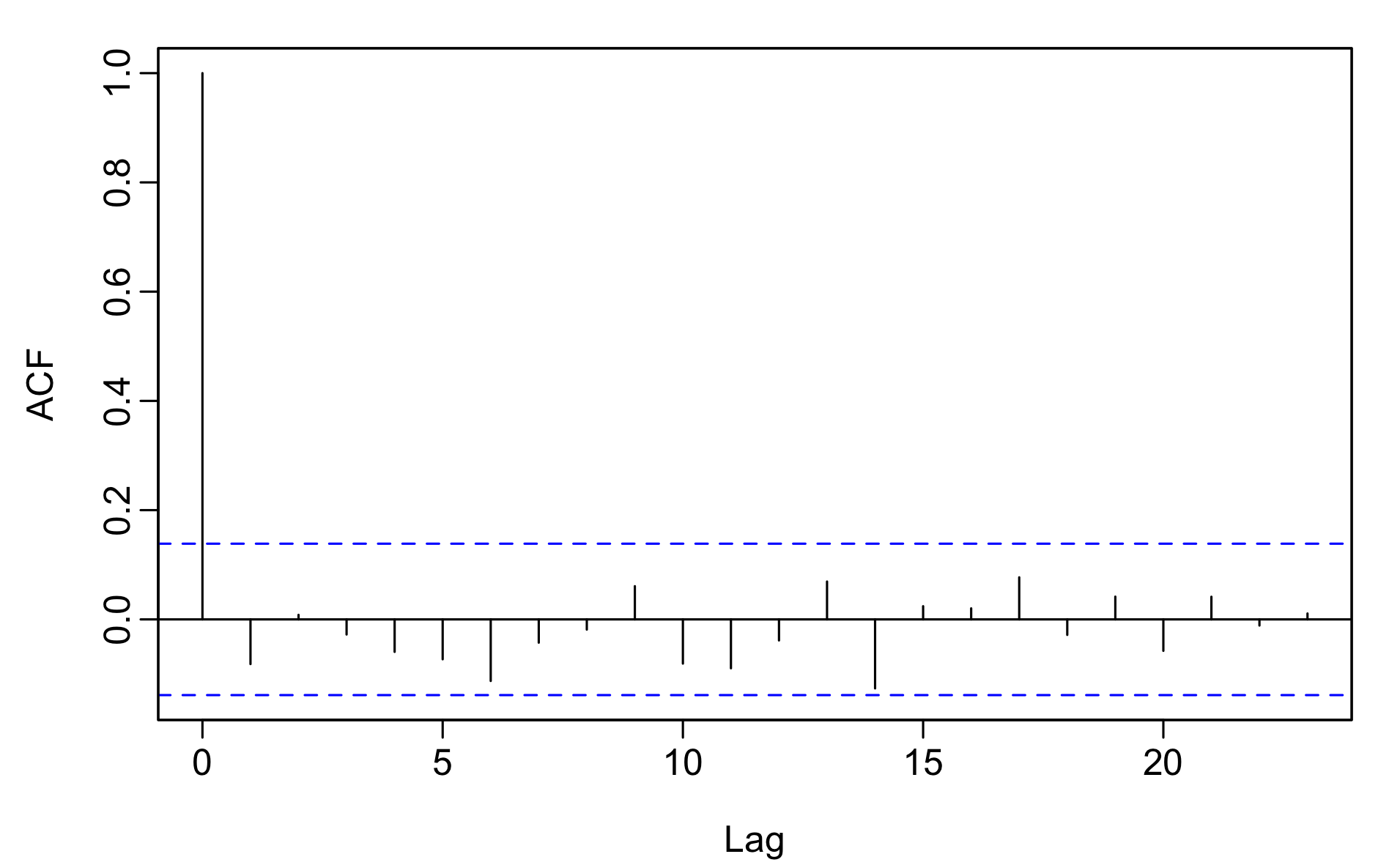}{.23}
    \subf{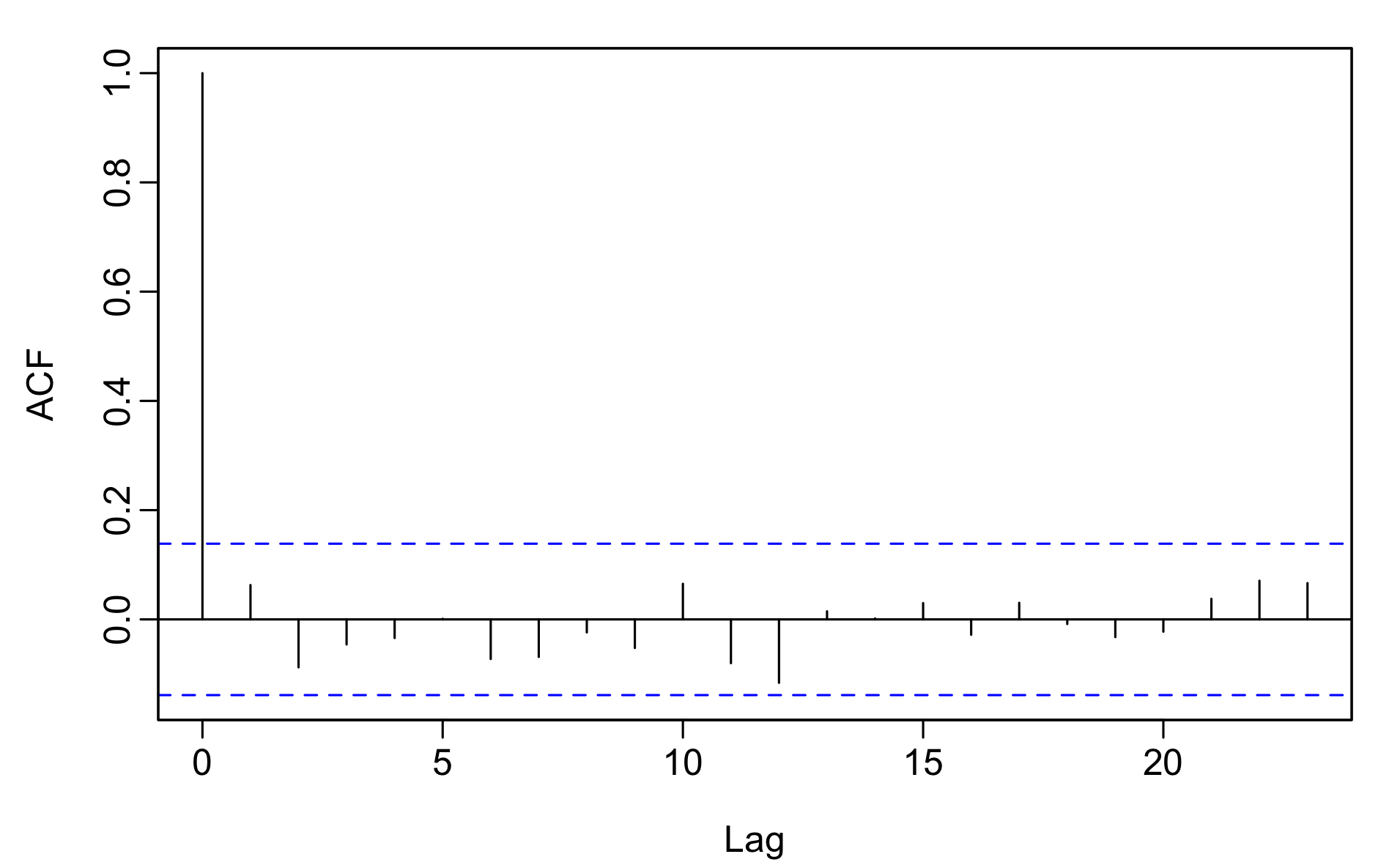}{.23}
    \caption{An eample where $C_\epsilon$ is of rank 2 and loading matrices are of rank 3}\label{figure - last one}
\end{figure}
The top eight subfigures of Figure~\ref{figure - last one} plot the sample eigenvectors of a model with $K=10$,  $C_\epsilon$ following Setting (C) and $A_n = G_{p\times 3} \times G_{3\times K}$, where $G_{m,n}$ denotes a $m\times n$ matrix of i.i.d. standard Gaussians. Clearly the first six eigenvectors behave differently than the last two, which is in line with our conjecture since here the ``effective rank'' of the model is equal to $\norm{C_\epsilon}_0\norm{A_n}_0=6$. The bottom eight figures plot the sample ACF of each of the sample eigenvectors as a time series of length 200. Clearly the first six eigenvectors exhibit persistent dependence which is reflected in their smoother sample paths, while the last two do not. 

A rigorous just justification of these observations and conjectures will serve as a valuable tool to identify non-stationary factors in a functional time series model, as well as distinguish spurious behaviour from genuine factors in the model. 
A much more intricate analysis of the eigenstructure of the model is needed to derive such results, and we will pursue this direction in a standalone work in the future. 

\bibliographystyle{plainnat}
\bibliography{ref2_cleaned}

\end{document}